\newif{\ifarxiv}
\arxivtrue

\documentclass[11pt]{article}
\pdfoutput=1 
\usepackage{amsthm}
\usepackage{amsmath,amssymb,stmaryrd}
\usepackage[T1]{fontenc}
\usepackage[utf8]{inputenc}
\usepackage{authblk}
\usepackage{cite}
\usepackage{amsfonts,stmaryrd}
\usepackage{amsmath}
\usepackage[inline]{enumitem}
\usepackage{url}
\usepackage{xypic}
\usepackage{xcolor}
\usepackage{orcidlink}

\usepackage{prooftree}

\newtheorem{theorem}{Theorem}[section]
\newtheorem{lemma}[theorem]{Lemma}
\newtheorem{fact}[theorem]{Fact}
\newtheorem{proposition}[theorem]{Proposition}
\newtheorem{corollary}[theorem]{Corollary}

\newtheorem{definition}[theorem]{Definition}
\newtheorem{remark}[theorem]{Remark}

\newcommand\kwfont[1]{\pmb{\mathtt{#1}}}
\newcommand\eqdef{\mathrel{\buildrel \text{def}\over=}}
\newcommand\diff{\setminus}
\newcommand\limp{\mathrel{\Rightarrow}}

\newcommand\sqrtop{\text{sqrt}}
\newcommand\sqrtkw{\underline{\sqrtop}}

\newcommand\intT{\kwfont{int}}
\newcommand\boolT{\kwfont{bool}}
\newcommand\unitT{\kwfont{unit}}
\newcommand\voidT{\kwfont{void}}
\newcommand\realT{\kwfont{real}}

\newcommand\muxop{\text{mux}}
\newcommand\muxkw{\underline{\muxop}}
\newcommand\posop{\text{pos}}
\newcommand\poskw{\underline{\posop}}
\newcommand\expop{\text{exp}}
\newcommand\expkw{\underline{\expop}}

\newcommand\ifkw{\kwfont{if}}

\newcommand\true{\underline{\trueop}}
\newcommand\trueop{\mathrm{true}}
\newcommand\false{\underline{\falseop}}
\newcommand\falseop{\mathrm{false}}

\newcommand\sample{\kwfont{sample}}
\newcommand\score{\kwfont{score}}
\newcommand\retkw{\mathop{\kwfont{ret}}\nolimits}

\newcommand\dokw{\mathop{\kwfont{do}}\nolimits}
\newcommand\reckw{\kwfont{rec}}

\newcommand\casekw{\kwfont{case}}

\newcommand\headop{\mathrm{head}}
\newcommand\headkw{\underline{\headop}}

\newcommand\tailop{\mathrm{tail}}
\newcommand\tailkw{\underline{\tailop}}

\newcommand\letkw{\kwfont{let}}
\newcommand\inkw{\kwfont{in}}
\newcommand\letbe[2]{\letkw\;{#1}\;\inkw\;{#2}}
\newcommand\nat{\mathbb{N}}

\newcommand\Z{\mathbb{Z}}
\newcommand\Eval[1]{\left\llbracket{#1}\right\rrbracket}

\newcommand{\real}{\mathbb{R}}
\newcommand\Ipref{\mathbf{I}}
\newcommand\IR{\Ipref\real}
\newcommand\realp{\real_+}
\newcommand{\creal}{\overline\real_+}

\newcommand\FV{\mathrm{fv}}

\newcommand\R{\mathrel{R}}
\newcommand\rover[1]{\overline{\overline{#1}}}

\newcommand\IRbb{\IR^\star} 
\newcommand\identity[1]{\mathrm{id}_{#1}}

\newcommand\Lform{\mathcal{L}}
\newcommand\Val{\mathbf{V}}

\newcommand\lfp{\mathop{\text{lfp}}}

\newcommand\Config{\Gamma}
\newcommand\upc{\mathop{\uparrow}}
\newcommand\dc{\mathop{\downarrow}}
\newcommand\uuarrow{\rlap{$\uparrow$}\raise.5ex\hbox{$\uparrow$}}
\newcommand\ddarrow{\rlap{$\downarrow$}\raise.5ex\hbox{$\downarrow$}}
\newcommand\Open{{\mathcal O}}
\newcommand\fin{\mathrm{f}}

\newcommand\cb[1]{\mathbf{#1}} 
\newcommand{\catc}{{\cb{C}}} 

\newcommand{\Dcpo}{{\mathbf{Dcpo}}}
\newcommand{\TOP}{{\mathbf{Top}}}

\newcommand\App{\mathrm{App}}

\newcommand\num{\mathrm{num}}
\newcommand\bin{\mathrm{bin}}

\begin{document}
%
\title{A Domain-Theoretic Approach to Statistical Programming Languages}

\author[1]{Jean Goubault-Larrecq\, \orcidlink{0000-0001-5879-3304}}
\author[2,3]{Xiaodong Jia\, \orcidlink{0000-0001-9310-6143}}
\author[1]{Cl\'ement Th\'eron\, \orcidlink{0000-0003-2703-4289}}
\affil[1]{Universit\'e Paris-Saclay, CNRS, ENS Paris-Saclay, Laboratoire M\'ethodes Formelles, 91190, Gif-sur-Yvette, France}
\affil[2]{School of Mathematics, Hunan University,
  Changsha, Hunan 410082, China}
\affil[3]{Department of Computer Science, Tulane University,
  New Orleans, LA 70118, USA}

\maketitle

\begin{abstract}
  We give a domain-theoretic semantics to a statistical programming
  language, using the plain old category of dcpos, in contrast to some
  more sophisticated recent proposals.  Remarkably, our monad of
  minimal valuations is commutative, which allows for program
  transformations that permute the order of independent random draws,
  as one would expect.  A similar property is not known for Jones and
  Plotkin' s monad of continuous valuations.  Instead of working with
  true real numbers, we work with exact real arithmetic, providing a
  bridge towards possible implementations.  (Implementations by
  themselves are not addressed here.)  Rather remarkably, we show that
  restricting ourselves to minimal valuations does not restrict us
  much: \emph{all} measures on the real line can be modeled by minimal
  valuations on the domain $\IR_\bot$ of exact real arithmetic.  We
  give three operational semantics for our language, and we show that
  they are all adequate with respect to the denotational semantics.
  We also explore quite a few examples in order to demonstrate that
  our semantics computes exactly as one would expect, and in order to
  debunk the myth that a semantics based on continuous maps would not
  be expressive enough to encode measures with non-compact support
  using only measures with compact support, or to encode measures via
  non-continuous density functions, for instance.  Our examples also
  include some useful, non-trivial cases of distributions on
  higher-order objects.
\end{abstract}


\maketitle

\ifarxiv
\relax
\else
\tableofcontents
\fi

\section{Introduction}

The purpose of this paper is to give a simple, domain-theoretic semantics
to statistical programming languages.

Statistical programming languages such as Church \cite{GMRBT:Church},
Anglican \cite{WvdMM:Anglican}, WebPPL \cite{GS:webppl} or Venture
\cite{MSP:venture}, were introduced as a convenient means to describe
and implement so-called stochastic generative processes.  Those are
randomized programs that describe probability distributions.

Initial proposals focused on implementations.  One of the first
proposals for a formal semantics of such a language, SFPC, is due to
V{\'a}k{\'a}r et al.\ \cite{VKS:SFPC}, and is based on
\emph{quasi-Borel predomains}, a notion that expands on the clever
notion of \emph{quasi-Borel spaces} \cite{HKSY:qBorel}, with
additional domain-theoretic structure.

The constructions of \cite{VKS:SFPC} are rather involved, and one may
wonder there would exist simpler denotational semantics for such
languages.  We will give one, based on domain theory alone.  Domain
theory is probably one of the oldest mathematical basis for
denotational semantics \cite{scott69}.  It is a common belief that it
would be inadequate for giving semantics to probabilistic languages.
This may be due to a superficial reading of a famous paper by A. Jung
and R. Tix \cite{jung98}.  And indeed, there are several purely
domain-theoretic semantics of probabilistic languages
\cite{jones90,jgl-jlap14,JGL-lics19}.

However, it is true that statistical probabilistic languages present
additional challenges to the semanticist.  Those are caused by several
additional features that one must take into account: a native type of
real numbers, continuous distributions, and perhaps most importantly,
\emph{soft constraints} \cite{SYHKW:soft:constr}.  The latter is a
convenient way of implementing the computation of conditional
distributions, or smoothed versions therefore, and is implemented by a
primitive called $\score$ in SFPC.

One difference between the quasi-Borel semantics of SFPC
\cite{VKS:SFPC}, or that of PCFSS \cite{DLH:geom:bayes}, or that of
PPCF \cite{EPT:PPCF}, with our domain-theoretic semantics does not lie
in probabilities or soft constraints, but with the way we handle real
numbers.  SFPC, PCFSS, and PPCF rely on \emph{true} real numbers, that
is, values of type $\realT$ are interpreted as elements of $\real$.
We interpret values of type $\realT$ as so-called \emph{exact} real
numbers, namely as elements of a dcpo $\IR_\bot$ of interval
approximations of real numbers, as in RealPCF and other proposals for
so-called exact real arithmetic
\cite{BC:exact:real,lester92,sunderhauf95c,escardo96a,escardo96,edalat97,Plume:ERC,marcial04,Ho:ERCE}.
This is the most natural choice with domain-theoretic semantics, as
$\real$ itself is not a dcpo (more precisely, the topology of $\real$
cannot be obtained as a Scott topology), but embeds naturally in a
domain of interval approximations.  This is also a natural bridge to
implementations; although we will not pursue this topic in depth, the
final section of this paper will give hints.

\paragraph{Contributions.}

The main contribution of this paper is, therefore, a simple, purely
domain-theoretic denotational semantics for a statistical programming
language with exact real arithmetic, continuous distributions and soft
constraints, featuring full recursion.  Additionally, and as in
previous proposals such as \cite{VKS:SFPC}, our monad of so-called
\emph{minimal valuations} implementing probabilistic choice is
commutative.  This is crucial in establishing the correctness of
run-of-the-mill program transformations such as permuting the order in
which two independent random variables are drawn.  In domain-theoretic
semantics of probabilistic programming languages based on Jones and
Plotkin's continuous valuations \cite{jones89,jones90}, it was not
known how to achieve this, at least until recent work by Jia,
Lindenhovius, Mislove and Zamdzhiev \cite{JLM:prob:quant}.  (We made
that discovery at the same time that they did: see the final related
work section for details.)  For the moment, let us just say that there
are two differences between our work and theirs.  The less significant
one is probably that we work with (minimal) valuations that are
unbounded, as required to give semantics to $\score$, while they work
with subprobability valuations.  The more significant one is that they
do not handle \emph{continuous} distributions.  It is a priori unclear
which continuous distributions on $\real$ can be represented as
\emph{minimal} valuations.  Being minimal, indeed, is a genuine
restriction: as we will see, Lebesgue measure is a continuous
valuation on $\real$ that is not minimal.  In spite of this, and this
is perhaps the most pleasing aspect of the current work, \emph{all}
measures on $\real$---including Lebesgue measure---can be realized as
minimal valuations on $\IR_\bot$, as we will show in
Section~\ref{sec:lebesgue-r-valuation}.  We will also provide an
extensive list of examples in Section~\ref{sec:examples}, by which we
hope to demonstrate that our approach is capable of defining a rich
set of distributions, including distributions on higher-order objects.

\paragraph{Outline.}

We give some preliminary definitions in
Section~\ref{sec:preliminaries}, where we also describe some of the
challenges in more detail.  In Section~\ref{sec:monads-continuous}, we
recapitulate the construction and basic properties of the monad $\Val$
of continuous valuations.  The commutativity of that monad is
equivalent to a form of the Fubini-Tonelli theorem, which in its most
basic form expresses an interchange property for double integrals.
Rather subtly, and perhaps paradoxically, such a Fubini-Tonelli
theorem is easy to obtain for the $\Val$ monad on the category $\TOP$
of topological spaces, but is an open problem on the subcategory
$\Dcpo$ of dcpos.  We will explain the issue in
Section~\ref{sec:fubini-theorem-what}, where we will see that
everything boils down to the fact that products in $\Dcpo$ are in
general different from products in the larger category $\TOP$.

In order to obtain a Fubini-Tonelli theorem on $\Dcpo$, we may opt to
restrict to, say, continuous dcpos, but this runs into some trouble,
as we will have seen in Section~\ref{sec:preliminaries}.  Our solution
is much simpler: we restrict continuous valuations to a submonad
$\Val_m$ of so-called \emph{minimal} valuations, and we show that
$\Val_m$ is a commutative monad on $\Dcpo$ in
Section~\ref{sec:minim-valu-fubini}.  Now, since we will restrict our
valuations to be minimal, doesn't this exclude some interesting
continuous distributions?  And indeed that will seem to be the case:
as we have already mentioned, we will show that Lebesgue measure, seen
as a continuous valuation on $\real$, is not minimal.  But (and again,
the shift is subtle), it \emph{is} minimal on the domain $\IR_\bot$
that serves to do exact real arithmetic.  In fact, as we show in
Section~\ref{sec:lebesgue-r-valuation}, \emph{every} measure on
$\real$ gives rise to a minimal valuation on $\IR_\bot$.  We will even
give a simple, explicit description of the corresponding minimal
valuation as a supremum of a countable chain of simple valuations.

Having done all this preliminary work, we introduce a higher-order
statistical programming language, ISPCF, with continuous distributions
and soft constraints, as well as full recursion, in
Section~\ref{sec:simple-sfpc-calculus}.  Its denotational semantics is
somehow straightforward, considering our preparatory steps, and is
given in Section~\ref{sec:semantics}.  In order to get a grasp of what
we can express in ISPCF, and more importantly, how one can reason
about ISPCF programs using that semantics, we provide a rather
extensive list of examples in Section~\ref{sec:examples}.  This
culminates with non-trivial examples of distributions on some
higher-order data types.

Finally, we will explore a few operational semantics for ISPCF, paving
the way for formally verified implementations.  (We will not address
implementations per se in this paper.)  The first operational
semantics for ISPCF we give is similar to some other earlier proposals
\cite{VKS:SFPC,DLH:geom:bayes,EPT:PPCF}, except that our real numbers
are exact reals, not true reals, and that our transition function is
continuous, not just measurable; it is given in
Section~\ref{sec:ideal-oper-semant}.  It is not too hard to give
another operational semantics which works with \emph{true} reals, and
which is therefore even closer to the operational semantics given in
\cite{VKS:SFPC,DLH:geom:bayes,EPT:PPCF}; this is the \emph{precise}
operational semantics of Section~\ref{sec:meas-oper-semant}.  Finally,
we give a sampling-based operational semantics in
Section~\ref{sec:effectivity}.  The name ``sampling-based'' is by
analogy with early work by \cite{PPT:sampling}, and with one of the
semantics of \cite{DLH:geom:bayes}.  That semantics is meant to be one
step closer to an implementation: instead of drawing real numbers at
random with respect to arbitrary measures on $\real$, the
sampling-based semantics draws bits independently at random, and
derives certain continuous distributions from those bits.  In each
case, we show that our operational semantics are sound and adequate
with respect to our denotational semantics.

We finish by reviewing related work in Section~\ref{sec:related-work},
and by concluding in Section~\ref{sec:conc}.

\section{Preliminaries, Challenges}
\label{sec:preliminaries}

We refer to \cite{billingsley86} for basics of measure theory, and to
\cite{abramsky94,gierz03,goubault13a} for basics of domain theory and
topology.

\subsection{Measure theory}

A \emph{$\sigma$-algebra} on a set $X$ is a collection of subsets
closed under countable unions and complements.  A \emph{measurable
  space} $X$ is a set with a $\sigma$-algebra $\Sigma_X$.  The
elements of $\Sigma_X$ are usually called the \emph{measurable
  subsets} of $X$.

A \emph{measure} $\mu$ on $X$ is a $\sigma$-additive map from
$\Sigma_X$ to $\creal$, where $\creal$ is the set of extended
non-negative real numbers $\realp \cup \{+\infty\}$.  We will agree
that $0 . (+\infty) = 0$.  (This makes multiplication on $\creal$
Scott-continuous, see below.)  The property of
\emph{$\sigma$-additivity} means that, for every countable family of
pairwise disjoint sets $E_n$,
$\mu (\bigcup_n E_n) = \sum_n \mu (E_n)$.  (Here $n$ ranges over any
subset of $\nat$, possibly empty.)

A \emph{measurable map} $f \colon X \to Y$ between measurable spaces
is a map such that $f^{-1} (E) \in \Sigma_X$ for every
$E \in \Sigma_Y$.  The \emph{image measure} $f [\mu]$ of a measure
$\mu$ on $X$ is defined by $f [\mu] (E) \eqdef \mu (f^{-1} (E))$.

The $\sigma$-algebra $\Sigma (A)$ \emph{generated by} a family $A$ of
subsets of $X$ is the the smallest $\sigma$-algebra containing $A$.
The \emph{Borel $\sigma$-algebra} on a topological space is the
$\sigma$-algebra generated by its topology.  The \emph{standard
  topology} on $\creal$ is generated by the intervals $[0, b[$,
$]a, b[$ and $]a, +\infty]$, with $0 < a < b < +\infty$.  Its Borel
$\sigma$-algebra is also generated by just the intervals
$]a, +\infty]$ (the Scott-open subsets, see below).  Hence a
measurable map $h \colon X \to \creal$ is a map such that
$h^{-1} (]t, +\infty]) \in \Sigma_X$ for every $t \in \real$.  Its
\emph{Lebesgue integral} can be defined elegantly through
\emph{Choquet's formula}:
$\int_X h d\mu \eqdef \int_0^{+\infty} \mu (h^{-1} (]t, +\infty]))
dt$, where the right-hand integral is an ordinary Riemann integral.

This formula makes the following \emph{change-of-variables formula} an
easy observation: for every measurable map $f \colon X \to Y$, for
every measurable map $h \colon Y \to \creal$,
$\int_Y h df[\mu] = \int_X (h \circ f) d\mu$.



There is a unique measure $\lambda$ on $\real$ such that
$\lambda (]a, b[) = b-a$ for every open bounded interval $]a, b[$.
This measure is called \emph{Lebesgue measure}.

A measure $\mu$ on $X$ is \emph{bounded} if and only if
$\mu (X) < +\infty$.  A measure $\mu$ is \emph{$\sigma$-finite} if
there is a sequence
$E_0 \subseteq E_1 \subseteq \cdots \subseteq E_n \subseteq \cdots$ of
measurable subsets of $X$ whose union is $X$ and such that
$\mu (E_n) < +\infty$ for every $n \in \nat$.  A \emph{$\pi$-system}
$\Pi$ on a set $X$ is a family of sets closed under finite
intersections.  If $X$ is a measurable space such that
$\Sigma_X = \Sigma (\Pi)$, any two $\sigma$-finite measures that agree
on $\Pi$ also agree on $\Sigma_X$.  In particular, Lebesgue measure on
$\real$ is uniquely defined by the specification
$\lambda (]a, b[) = b-a$.

\subsection{Domain theory and topology}

A \emph{dcpo} is a poset in which every directed family $D$ has a
supremum $\sup D$.  A prime example is $\creal$, with the usual
ordering.  Another example is $\IR$, the poset of closed intervals
$[a, b]$ with $a, b \in \real$ and $a \leq b$, ordered by reverse
inclusion $\supseteq$.  Every directed family
${([a_i, b_i])}_{i \in I}$ in $\IR$ has a supremum
$\bigcap_{i \in I} [a_i, b_i] = [\sup_{i \in I} a_i, \inf_{i \in I}
b_i]$.  $\IR_\bot$ is the \emph{lift} of $\IR$, namely the dcpo
obtained by adding a fresh element $\bot$ below all others.  In
general, we define the lift $X_\bot$ of a dcpo $X$ similarly.  In the
case of $\IR_\bot$, we may equate $\bot$ with the whole set $\IR$
itself, so that $\IR_\bot$ is still ordered by reverse inclusion.
$\IR_\bot$ will be the domain of interpretation of the type $\realT$
of exact real numbers.  Among them, we find the \emph{total} numbers
$a \in \real$, which we may equate with the maximal elements $[a, a]$
of $\IR$.

$\IR_\bot$ is an example of a \emph{pointed} dcpo, namely one that has
a least element, which we will always write as $\bot$, and which we read
as \emph{bottom}.

The \emph{standard topology} on $\real$ is generated by the open
intervals $]a, b[$, with $a < b$.  The map $i \colon a \mapsto [a, a]$
is then a topological embedding of $\real$, with its standard
topology, into $\IR$ (or $\IR_\bot$) with its Scott topology.  In
other words, $i$ is continuous, and every open subset $U$ of $\real$
is the inverse image of some Scott-open subset of $\IR$ (resp.,
$\IR_\bot$) by $i$.  Explicitly, $]a, b[$ is the inverse image of the
Scott-open subset of intervals $[c, d]$ such that $a < c \leq d < b$.
This allows us to consider $\real$ as a subspace of $\IR$, resp.\
$\IR_\bot$.

We will also write $\leq$ for the ordering on any poset.  In the
example of $\IR$ or $\IR_\bot$, $\leq$ is $\supseteq$.  The
\emph{upward closure} $\upc A$ of a subset $A$ of a poset $X$ is
$\{y \in X \mid \exists x \in A, x \leq y\}$.  The \emph{downward
  closure} $\dc A$ is defined similarly.  A set $A$ is \emph{upwards
  closed} if and only if $A = \upc A$, and \emph{downwards closed} if
and only if $A = \dc A$.  A subset $U$ of a poset $X$ is
\emph{Scott-open} if and only if it is upwards closed and, for every
directed family $D$ such that $\sup D$ exists and is in $U$, some
element of $D$ is in $U$ already.  The Scott-open subsets of a poset
$X$ form its \emph{Scott topology}.

The \emph{way-below} relation $\ll$ on a poset $X$ is defined by
$x \ll y$ if and only if, for every directed family $D$ with a
supremum $z$, if $y \leq z$, then $x$ is less than or equal to some
element of $D$ already.  We write $\uuarrow x$ for
$\{y \in X \mid x \ll y\}$, and $\ddarrow y$ for
$\{x \in X \mid x \ll y\}$.  A poset $X$ is \emph{continuous} if and
only if $\ddarrow x$ is directed and has $x$ as supremum for every
$x \in X$.  A \emph{basis} $B$ of a poset $X$ is a subset of $X$ such
that $\ddarrow x \cap B$ is directed and has $x$ as supremum for every
$x \in X$.  A poset $X$ is continuous if and only if it has a basis
(namely, $X$ itself).  A poset is \emph{$\omega$-continuous} if and
only if it has a countable basis.  Examples include $\creal$, with any
countable dense subset (with respect to its standard topology), such
as the rational numbers, or such as the \emph{dyadic numbers} $k/2^n$
($k, n \in \nat$); or $\IR$ and $\IR_\bot$, with the basis of
intervals $[a, b]$ where $a$ and $b$ are both dyadic or both rational.

We write $\Open X$ for the lattice of open subsets of a topological
space.  This applies to dcpos $X$ as well, which will always be
considered with their Scott topology.  The continuous maps
$f \colon X \to Y$ between two dcpos coincide with the
\emph{Scott-continuous} maps, namely the monotonic (order-preserving)
maps that preserve all directed suprema.  We write $\Lform X$ for the
space of continuous maps from a topological space $X$ to $\creal$, the
latter with its Scott topology, as usual.  Such maps are usually
called \emph{lower semicontinuous}, or \emph{lsc}, in the mathematical
literature.  Note that $\Lform X$, with the pointwise ordering, is a
dcpo.

There are several ways in which one can model probabilistic choice.
The most classical one is through measures.  A popular alternative
used in domain theory is given by \emph{continuous valuations}
\cite{jones89,jones90}.  A continuous valuation is a Scott-continuous
map $\nu \colon \Open X \to \creal$ such that $\nu (\emptyset)=0$
(\emph{strictness}) and, for all $U, V \in \Open X$,
$\nu (U \cup V) + \nu (U \cap V) = \nu (U) + \nu (V)$
(\emph{modularity}). 
Canonical examples of continuous valuations on $X$ are \emph{Dirac
valuations} $\delta_{x}$ for $x\in X$, where for each open subset $U$
of $X$, $\delta_{x}(U) = 1$ if $x\in U$ and $\delta_{x}(U) = 0$, otherwise. 
 The set of all continuous valuations on $X$
is denoted by $\Val X$. We order $\Val X$ by the \emph{stochastic 
order} defined as $\nu_{1} \leq \nu_{2}$ if and only if 
$\nu_{1}(U) \leq \nu_{2}(U)$ for all opens of $X$. The set 
$\Val X$ is a dcpo in the stochastic order. 

There is a notion of integral
$\int_{x \in X} h (x) d\nu$, or briefly $\int h d\nu$, for every
$h \in \Lform X$, which can again be defined by a Choquet formula.
Tix \cite[Satz~4.4]{tix95} showed that the integral is a
Scott-continuous bilinear form, namely:
\begin{itemize}
\item for every $\nu \in \Val X$, the map $h \in \Lform X \mapsto \int
  h d\nu$ is Scott-continuous and linear, in the sense that
  $\int \alpha h d\nu = \alpha \int h d\nu$ for every $\alpha \in
  \real$ and $\int (h+h') d\nu = \int h d\nu + \int h' d\nu$, for all
  $h, h' \in \Lform X$;
\item for every $h \in \Lform X$, the map $\nu \in \Val X \mapsto \int
  h d\nu$ is Scott-continuous and linear, in a similar sense.
\end{itemize}
More generally, a \emph{linear} map $G \colon \Lform X \to \creal$
satisfies $G (h+h') = G (h)+G(h')$ and $G (\alpha.h) = \alpha.G(h)$
for all $\alpha \in \realp$, $h, h' \in \Lform X$.  Conversely, any
Scott-continuous linear map $G \colon \Lform X \to \creal$ is of the
form $h \mapsto \int h d\nu$ for a unique continuous valuation $\nu$,
given by $\nu (U) \eqdef G (\chi_U)$, where $\chi_U$ is the
characteristic map of $U$ ($\chi_U (x) \eqdef 1$ if $x \in U$, $0$
otherwise).

Continuous valuations and measures are pretty much the same thing on
$\omega$-continuous dcpos, namely on continuous dcpos with a countable
basis.  This holds more generally on de Brecht's quasi-Polish spaces
\cite{deBrecht:qPolish}, a class of spaces that contains not only the
$\omega$-continuous dcpos from domain theory but also the Polish
spaces from topological measure theory.

One can see this as follows.  In one direction, Adamski's theorem
states that every measure $\mu$ on a hereditarily Lindel\"of space $X$
is $\tau$-smooth \cite[Theorem~3.1]{Adamski:measures}, meaning that
its restriction to the lattice of open subsets of $X$ is a continuous
valuation.  A hereditarily Lindel\"of space is a space whose subspaces
are all Lindel\"of, or equivalently a space in which every family of
open sets contains a countable subfamily with the same union.  Every
second-countable space is hereditarily Lindel\"of, and that includes
all quasi-Polish spaces.  In the other direction, every continuous
valuation on an LCS-complete space extends to a Borel measure
\cite[Theorem~1.1]{DGJL-isdt19}.  An \emph{LCS-complete} space is a
space that is homeomorphic to a $G_\delta$ subset of a locally compact
sober space.  Every quasi-Polish space is LCS-complete; in fact, the
quasi-Polish spaces are exactly the second-countable LCS-complete
spaces \cite[Theorem~9.5]{DGJL-isdt19}.

\subsection{Is there any trouble with the probabilistic powerdomain?}

The probabilistic powerdomain, namely the dcpo $\Val X$ of all
continuous valuations on a space $X$, ordered pointwise, is known to
have its problems \cite{jung98}.  Precisely, there is no known
Cartesian-closed category of continuous dcpos that is closed under the
$\Val$ functor.  It is sometimes believed that this means that domain
theory cannot be used to give semantics to higher-order probabilistic
languages.  This would be wrong: the category $\Dcpo$ of \emph{all}
dcpos, not just the continuous dcpos, \emph{is} Cartesian-closed and
closed under the $\Val$ functor \cite{jones89,jones90}.

Continuity is not required to prove, say, soundness and adequacy
theorems using logical relations, as one realizes by reading the
relevant parts of \cite{streicher02}, and as we will do in
Section~\ref{sec:adequacy}.  But it is required to obtain a form of
the Fubini-Tonelli theorem, or, in categorical terms, to turn $\Val$
into a \emph{commutative} monad.  (See
\cite[Theorem~9.2]{Kock:monad:comm} for the relation between the two
notions.)  Commutativity is important in applications, as we will
briefly discuss in Remark~\ref{rem:comm}, and is a key ingredient of
the semantics of \cite{VKS:SFPC}.  We \emph{do} obtain a commutative
monad, without any need for continuity, by a simple trick based on
inductive closures (Definition~\ref{defn:simple},
Section~\ref{sec:minim-valu-fubini}).  This trick was found
independently by at least one other group of researchers
\cite{JLM:prob:quant}, but on different, and incomparable, monads.

\subsection{Scoring, and density functions}
Compared to ordinary probabilistic languages, statistical programming
languages aim to offer the possibility of computing conditional
distributions.  This runs into questions on non-computability
\cite{AFR:noncomp}.  At this point, we note that there are
well-established theories of computable probability distributions,
notably on computable metric spaces \cite{GHR:symb:dyn}, on more
general computable topological spaces \cite{Roy:PhD}, and also based
on Weihrauch's type two theory of effectivity, also known as
\emph{TTE} \cite{Weihrauch:computability}, see
\cite{Weihrauch:prob,SS:processes}.  In the absence of probabilities,
Schulz \cite{Schulz:TTE=RealPCF} shows that the functions from
$[0, 1]^n$ to $[0, 1]$ that are computable in the sense of TTE and of
RealPCF coincide.

In practice, a number of algorithms are implemented to compute certain
special conditional distributions, with a fallback strategy based on
one form or another of rejection sampling, or with a so-called
\emph{scoring} mechanism, which allows one to give more or less
importance to specific outcomes.

This scoring mechanism is typically implemented through a primitive
called $\score$
\cite{SYHKW:soft:constr,VKS:SFPC,DLH:geom:bayes,EPT:PPCF}.  Roughly,
the effect of $\score\; M$, where $M \colon \realT$ evaluates to a
non-negative real number $\alpha$, is to multiply the `probability' of
the current computation branch by $\alpha$---making it a measure
rather than a probability, whence the quotes.

As a consequence, $\score$ can be used to build new measures
$g \cdot \mu$ from a measure $\mu$ and a density function $g$.  The
measure $g \cdot \mu$ is defined by
$(g \cdot \mu) (E) \eqdef \int_x \chi_E (x) g (x) d\mu$, and is
sometimes written as $g d\mu$.  Indeed, in case $x$ is drawn at random
with respect to some measure $\mu$, writing $\score\;g(x); M(x)$ will
have the effect of executing $M (x)$ as though $x$ had been drawn with
probability multiplied by $g (x)$, namely as though it had been drawn
at random with respect to $g \cdot \mu$.  This interpretation conforms
to intuition if $g$ really is a density function, namely if
$\int_x g (x) d\mu=1$, in which case $g \cdot \mu$ is a probability
distribution.  In general, however, $g \cdot \mu$ will be a measure.
As a simple, but extreme example, $\score\; 0$ annihilates the effect
of the current computation.  If $x$ is drawn with a measure whose
total mass is, say, $\pi/4$, then $\score (4/\pi)$ will renormalize
the measure to a probability distribution.  Other uses of $\score$
include soft conditioning and Bayesian fitting, as illustrated in
\cite{VKS:SFPC}.


\section{Monads of continuous valuations}
\label{sec:monads-continuous}

We will describe probabilistic effects by following Moggi's seminal
work on monads \cite{moggi89,moggi91}.  We use Manes' presentation of
monads \cite{Manes:algth}: a monad $(T, \eta, \_^\dagger)$ on a
category $\catc$ is a function $T$ mapping objects of $\catc$ to
objects of $\catc$, a collection of morphisms
$\eta_X \colon X \to TX$, one for each object $X$ of $\catc$, and
called the \emph{unit}, and for every morphism $f \colon X \to TY$, a
morphism $f^\dagger \colon TX \to TY$ called the \emph{extension} of
$f$; those are required to satisfy the axioms:
\begin{enumerate}
\item $f^\dagger \circ \eta_X = f$;
\item $\eta_X^\dagger = \identity {TX}$;
\item $(g^\dagger \circ f)^\dagger = g^\dagger \circ f^\dagger$.
\end{enumerate}
Then $T$ extends to an endofunctor, acting on morphisms through
$T f = (\eta_Y \circ f)^\dagger$.  Proposition~\ref{prop:VX:monad}
below is due to Jones \cite[Theorem~4.5]{jones90}.  Her definition of
the integral was different, and she implicitly restricted valuations
to subprobability valuations.  A similar statement is due to Kirch
\cite[Satz~6.1]{kirch93}, for continuous dcpos instead of general
dcpos.  Tix \cite{tix95} was probably the first to use the Choquet
formula in this context.  Her study was also restricted to continuous
dcpos.

We define $\eta_X \colon X \to \Val X$ as mapping $x \in X$ to
$\delta_x$.  For every Scott-continuous map $f \colon X \to \Val Y$,
for every $\mu \in \Val X$, for every Scott-open subset $V$ of $Y$, we
define:
\begin{align}
  \label{eq:dagger:def}
  f^\dagger (\mu) (V) & \eqdef \int_{x \in X} f (x) (V) d\mu.
\end{align}

The following lemma is proved exactly as in most of the references we
have just cited.
\begin{lemma}
  \label{lemma:eta:mu}
  For all dcpos $X$ and $Y$, for every Scott-continuous map
  $f \colon X \to \Val Y$,
  \begin{enumerate}[label=(\roman*)]
  \item the map $\eta_X$ is Scott-continuous;
  \item the map $f^\dagger$ is Scott-continuous from $\Val X$ to $\Val
    Y$;
  \item for every $\mu \in \Val X$, for every $g \in \Lform Y$,
    \begin{align}
      \label{eq:dagger}
      \int_{y \in Y} g (y) df^\dagger (\mu)
      & = \int_{x \in X}
        \left(\int_{y \in Y} g (y) df (x)\right) d\mu.
    \end{align}
    \item  For every
$h \in \Lform X$, $\int_{x' \in X} h(x') d\delta_x = h (x)$. 
  \end{enumerate}
\end{lemma}
\begin{proof}
  (i) If $x \leq y$, then for every $U \in \Open X$, if
  $\delta_x (U)=1$ then $x$ is in $U$, so $y$ is in $U$ as well, and
  therefore $\delta_y (U)=1$.  It follows that
  $\delta_x \leq \delta_y$.  Let $D$ be any directed family in $X$,
  and $x \eqdef \sup D$.  For every $U \in \Open X$, $\delta_x (U) = 1$ if
  and only if $x \in U$, if and only if some element $y \in D$ is in
  $U$, by definition of Scott-open sets; and that is equivalent to
  $\sup_{y \in D} \delta_y (U)=1$.

  (ii) We first verify that $f^\dagger (\mu)$ is a continuous
  valuation, for every $\mu \in \Val X$.  Strictness and modularity
  follow easily from the fact that $f (x) \in \Val Y$ for every $x \in
  X$.  Scott-continuity follows from the fact that the integral is a
  (bilinear) form that is Scott-continuous in its function argument.
  The integral is also Scott-continuous in its valuation argument, so
  $f^\dagger$ is itself Scott-continuous.

  (iii) We write $g$ as the directed supremum of the maps $g_K$,
  defined by $g_K (y) \eqdef \sum_{k=1}^{K2^K} \frac 1 {2^K}
  \chi_{g^{-1} (]\frac k {2^K}, +\infty])}$, $K \in \nat$.  Then:
  \begin{align*}
    \int_{y \in Y} g_K (y) df^\dagger (\mu)
    & = \sum_{k=1}^{K2^K} \frac 1 {2^K} \int_{y \in Y} \chi_{g^{-1}
      (]\frac k {2^K}, +\infty])} (y) df^\dagger (\mu) \\
    & = \sum_{k=1}^{K2^K} \frac 1 {2^K} f^\dagger (\mu)
      (g^{-1} (]\frac k {2^K}, +\infty]) \\
    & = \sum_{k=1}^{K2^K} \frac 1 {2^K} \int_{x \in X} f (x)
      (g^{-1} (]\frac k {2^K}, +\infty]) d \mu \\
    & = \int_{x \in X} \sum_{k=1}^{K2^K} \frac 1 {2^K} f (x)
      (g^{-1} (]\frac k {2^K}, +\infty]) d \mu \\
    & = \int_{x \in X} \left(\int_{y \in Y} g_K (y) df (x)\right) d\mu,
  \end{align*}
  and the result follows by  Scott-continuity of the integral in its
  function argument and taking suprema, as $K$ tends to $+\infty$.
  
  (iv)  The Choquet formula for the integral yields\\
 $\int_{x' \in X} h(x') d\delta_x = \int_0^{+\infty} \delta_x (h^{-1} (]t, +\infty])) dt = \int_0^{h (x)} 1 dt = h (x)$.
\end{proof}

\begin{proposition}
  \label{prop:VX:monad}
  The triple $(\Val, \eta, \_^\dagger)$ is a monad on the category
  $\Dcpo$ of dcpos and Scott-continuous maps.  For every
  Scott-continuous map $f \colon X \to Y$, for every $\mu \in \Val X$,
  $\Val f (\mu)$ is the image valuation $f [\mu]$, defined by
  $f [\mu] (V) \eqdef \mu (f^{-1} (V))$, for every $V \in \Open Y$.
\end{proposition}
\begin{proof}
  In light of Lemma~\ref{lemma:eta:mu}, the first part will be proved
  once we have verified the three axioms given by Manes:
  \begin{enumerate}
  \item $f^\dagger \circ \eta_X$ maps $x$ to the continuous valuation
    $\nu$ defined by $\nu (V) \eqdef \int_{x \in X} f (x) (V)
    d\delta_x = f (x) (V)$; so $\nu = f (x)$, and therefore $f^\dagger
    \circ \eta_X = f$.
  \item $\eta_X^\dagger$ satisfies $\eta_X^\dagger (\mu) (U) = \int_{x
      \in X} \eta_X (x) (U) d\mu = \int_{x \in X} \chi_U (x) d\mu =
    \mu (U)$, so $\eta_X^\dagger = \identity {\Val X}$.
  \item For all $f \colon X \to \Val Y$, $g \colon Y \to \Val Z$, and
    $W \in \Open Z$,
    \begin{align*}
      (g^\dagger \circ f)^\dagger (\mu) (W)
      & = \int_{x \in X} g^\dagger (f (x)) (W) d\mu \\
      & = \int_{x \in X} \left(\int_{y \in Y} g (y) (W) df(x)\right) d\mu,
    \end{align*}
    while:
    \begin{align*}
      g^\dagger (f^\dagger (\mu)) (W)
      & = \int_{y \in Y} g (y) (W) df^\dagger (\mu) \\
      & = \int_{x \in X}
        \left(\int_{y \in Y} g (y) (W) df (x)\right) d\mu,
    \end{align*}
    by (\ref{eq:dagger}).
  \end{enumerate}
\end{proof}

\section{The Fubini-Tonelli theorem, and what goes wrong with $\Dcpo$}
\label{sec:fubini-theorem-what}



Jones proved a form of Fubini's (more accurately, Tonelli's) theorem
for (subprobability) continuous valuations, on continuous dcpos
\cite{jones90}.  This actually generalizes to continuous valuations on
arbitrary topological spaces, and the proof is not that complicated,
as we demonstrate.  (Another purpose we have in giving that proof is in
order to fix a gap in Jones' proof, who actually does not show the
existence of the product valuation.)

As we will see in more detail later, Tonelli's theorem turns $\Val$
into a commutative monad, and that is a basic requirement for being
able to say that drawing two objects at random independently can be
done in any order.
\begin{proposition}[Fubini-Tonelli for continuous valuations on $\TOP$]
  \label{prop:fubini:top}
  Let $X$ and $Y$ be two spaces, $\mu \in \Val X$, $\nu \in \Val Y$.
  There is a unique so-called product valuation $\mu \times \nu$ on $X
  \times Y$ such that, for every $U \in \Open X$ and for every $V \in
  \Open Y$, $(\mu \times \nu) (U \times V) = \mu (U) . \nu (V)$.  For
  every $f \in \Lform (X \times Y)$,
  \begin{align*}
    \int_{(x, y) \in X \times Y} f (x, y) d(\mu \times \nu)
    & = \int_{x \in X} \left(\int_{y \in Y} f (x, y) d\nu\right)d\mu \\
    & = \int_{y \in Y} \left(\int_{x \in X} f (x, y) d\mu\right)d\nu.
  \end{align*}
\end{proposition}
\begin{proof}(Sketch.)  We first deal with uniqueness.  We assume any
  continuous valuation $\xi$ on $X \times Y$ such that
  $\xi (U \times V) = \mu (U) . \nu (V)$ for all $U \in \Open X$ and
  $V \in \Open Y$.  We call any such product $U \times V$ an
  \emph{open rectangle}.  The value of $\xi$ on finite unions
  $W \eqdef \bigcup_{i \in J} U_i \times V_i$ of open rectangles is
  determined uniquely: either $\mu (U_i) . \nu (U_i) = +\infty$ for
  some $i \in J$, and then $\xi (W)$ must be equal to $+\infty$, by
  monotonicity; or $\xi (W)$ must be equal to
  $\sum_{K \neq \emptyset, K \subseteq J} (-1)^{|K|+1} \xi (\bigcap_{i
    \in K} U_i \times V_i)$ by the so-called inclusion-exclusion
  formula (an easy consequence of modularity), showing that $\xi (W)$
  is again determined uniquely.  Finally, by definition of the product
  topology, every open subset $W$ of $X \times Y$ is a union
  $\bigcup_{i \in I} U_i \times V_i$ of open rectangles, hence a
  directed union of finite unions $\bigcup_{i \in J} U_i \times V_i$
  (where $J$ ranges over the finite subsets of $I$); $\xi (W)$ is then
  determined uniquely, since $\xi$ is Scott-continuous.

  For the existence part, the easiest route is to consider the two
  maps $G, G' \colon \Lform (X \times Y) \to \creal$ defined by:
  \begin{align*}
    G (f) & \eqdef \int_{x \in X} \left(\int_{y \in Y} f (x, y)
            d\nu\right)d\mu \\
    G' (f) & \eqdef \int_{y \in Y} \left(\int_{x \in X} f (x, y) d\mu\right)d\nu.
  \end{align*}
  We check that $G$ and $G'$ are Scott-continuous, linear maps, and
  are therefore integral functionals for unique continuous valuations
  $\xi$ and $\xi'$ on $X \times Y$, respectively.  For every open
  rectangle $U \times V$, $\xi (U \times V) = G (\chi_{U \times V})$
  is equal to $\mu (U).\nu (V)$, and similarly for $\xi' (U \times
  V)$.  By the uniqueness part, $\xi=\xi'$, so $G=G'$.  We write $\mu
  \times \nu$ for $\xi$, and the theorem is proved.
\end{proof}
This proof is the core of several similar proofs.  One of the closest
is due to Vickers \cite{Vickers:val:loc}, who proves a similar
theorem on the category of locales instead of $\TOP$.  Although
localic theorems usually generalize purely topological theorems, this
one does not, because locale products do not coincide with products in
$\TOP$ in general \cite[Theorem~2]{isbell81}.

Since every dcpo can be seen as a topological space with its Scott
topology, it would seem that we would obtain a Fubini-Tonelli theorem
for Scott-continuous maps and continuous valuations on arbitrary
dcpos, as a special case of Proposition~\ref{prop:fubini:top}.

This is not the case, but the reason is subtle.  As with locales,
products in $\Dcpo$ do not usually coincide with products in $\TOP$.
Explicitly, let us write $X_\sigma$ for the topological space obtained
by equipping a dcpo $X$ with its Scott topology.  Then the topology on
$(X \times Y)_\sigma$ (where $\times$ is dcpo product) is finer, and
in general strictly finer, than the product topology on
$X_\sigma \times Y_\sigma$.  (See Exercise~5.2.16 of
\cite{goubault13a} for an example where it is strictly finer.)  As a
consequence, there are more, and generally, strictly more
Scott-continuous maps from $X \times Y$ to $\creal$ than continuous
maps from the topological product $X_\sigma \times Y_\sigma$ to
$\creal$.  The Fubini-Tonelli formula holds for functions of the
second, smaller class, but it is unknown whether it holds for the
first kind of functions.  One exception is when either $X$ or $Y$ is
\emph{core-compact} (the Scott-opens form a continuous lattice), 
in which case the Scott and product topologies
coincide on $X \times Y$ \cite[Theorem II-4.13]{gierz03}; then
we retrieve Jones' version of the Fubini-Tonelli theorem.

One may blame $\Dcpo$ for this state of affairs.  Another possibility
is to consider that we are considering too general a notion of
continuous valuation.  We will explore this avenue in the next
section, by restricting to so-called minimal valuations.

\section{Minimal valuations, and Fubini-Tonelli again}
\label{sec:minim-valu-fubini}

%

A \emph{simple valuation} is any finite linear combination
$\sum_{i=1}^n r_i \delta_{x_i}$ of Dirac masses, with coefficients
$r_i$ in $\realp$.

\begin{definition}[Minimal valuations]
  \label{defn:simple}
  Let $\Val_\fin X$ be the poset of simple valuations on $X$, and
  $\Val_m X$ be the inductive closure of $\Val_\fin X$ in $\Val X$.
  The elements of $\Val_m X$ are called the \emph{minimal valuations}
  on $X$.
\end{definition}

The \emph{inductive closure} of a subset $A$ of a dcpo $Z$ is the
smallest subset of $Z$ that contains $A$ and is closed under directed
suprema.  It is obtained by taking all directed suprema of elements of
$A$, all directed suprema of elements obtained in this fashion, and
proceeding this way transfinitely.

\begin{remark}
  \label{rem:minval:cont}
  On a continuous dcpo $X$, every continuous valuation is a directed
  supremum of simple valuations.  In fact, $\Val X$ is a continuous
  dcpo with a basis of simple valuations, as showed by Jones
  \cite{jones90}, at least in the case of subprobability valuations.
  (The general theorem can be found as Theorem~IV-9.16 of
  \cite{gierz03}.)  Hence, in particular, continuous valuations and
  minimal valuations agree on every continuous dcpo.  One may wonder
  whether all continuous valuations on a dcpo are minimal.  We will
  give a counterexample to this claim in a forthcoming paper.
\end{remark}

\subsection{The monad $\Val_m$ of minimal valuations}
\label{sec:monad-val_m-minimal}


We will now show that $\Val_m$ defines a submonad of $\Val$.  To
this end, we need to know more about inductive closures.  A
\emph{d-closed subset} of a dcpo $Z$ is a subset $C$ such that the
supremum of every directed family of elements of $C$, taken in $Z$, is
in $C$.  The d-closed subsets form the closed subsets of a topology
called the \emph{d-topology} \cite[Section~5]{keimel08}, and the
inductive closure of a subset $A$ coincides with its \emph{d-closure}
$cl_d (A)$, namely its closure in the d-topology.

We note that every Scott-continuous map is continuous with respect to
the underlying d-topologies.  This is easily checked, or see
\cite[Lemma~5.3]{keimel08}.  In particular:
\begin{fact}
  \label{fact:cld}
  For every Scott-continuous map $f \colon \Val X \to \Val Y$, for
  every $A \subseteq \Val X$, $f (cl_d (A)) \subseteq cl_d (f (A))$.
\end{fact}

\begin{lemma}
  \label{lemma:VXm:drag}
  For every space $X$, $\Val_m X$ is closed under addition and
  multiplication by elements of $\creal$, as computed in the larger
  space $\Val X$.
\end{lemma}
\begin{proof}
  Let us deal with addition.  Multiplication is similar.
  
  For every simple valuation $\mu$, the map
  $f_\mu \colon \nu \in \Val X \mapsto \mu + \nu$ is Scott-continuous,
  and maps simple valuations to simple valuations.  By
  Fact~\ref{fact:cld} with $A \eqdef \Val_\fin X$, $f_\mu$ maps all
  elements of $cl_d (A) = \Val_m X$ to
  $cl_d (f_\mu (A)) \subseteq cl_d (\Val_\fin X) = \Val_m X$.

  It follows that for every minimal valuation $\nu$, the map
  $g \colon \mu \in \Val X \mapsto \mu+\nu = f_\mu (\nu)$ maps simple
  valuations to minimal valuations.  We observe that $g$ is also
  Scott-continuous.  By Fact~\ref{fact:cld} with the same $A$ as
  above, $g$ maps all elements of $cl_d (A) = \Val_m X$ to
  $cl_d (g (A)) \subseteq cl_d (\Val_m X) = \Val_m X$.  Hence, for
  every $\nu \in \Val_m X$, for every $\mu \in \Val_m X$, $\mu + \nu$
  is in $\Val_m X$.
\end{proof}

\begin{lemma}
  \label{lemma:VXm:dagger}
  For any Scott-continuous map $f \colon X \to \Val_m Y$, $f^\dagger$
  is a Scott-continuous map from $\Val_m X$ to $\Val_m Y$.
  Similarly, for every Scott-continuous map $f \colon X \to Y$, $\Val
  f$ is a Scott-continuous map from $\Val_m X$ to $\Val_m Y$.
\end{lemma}
\begin{proof}
  For the first part, the only challenge is to show that, for every
  $\nu \in \Val_m X$, $f^\dagger (\nu)$ is in $\Val_m Y$.
  Scott-continuity follows from the fact that $f^\dagger$ is
  Scott-continuous from $\Val X$ to $\Val Y$.

  For every $\nu \eqdef \sum_{i=1}^n r_i \delta_{x_i}$ in
  $\Val_\fin X$, $f^\dagger (\nu)$ is the continuous valuation
  $\sum_{i=1}^n r_i f (x_i)$: for every $V \in \Open Y$, $f^\dagger
  (\nu) (V) = \int_{x \in X} f (x) (V) d\nu = \sum_{i=1}^n r_i f (x_i) (V)$.
  By Lemma~\ref{lemma:VXm:drag}, and since $f (x_i)$ is in $\Val_m
  Y$ for each $i$, $f^\dagger (\nu)$ is in $\Val_m Y$ as well.

  Hence $f^\dagger$ maps $\Val_\fin X$ to $\Val_m Y$.  Using
  Fact~\ref{fact:cld} with $A \eqdef \Val_\fin X$,
  $f^\dagger (cl_d (A)) = f^\dagger (\Val_m X)$ is included in
  $cl_d (f^\dagger (A)) \subseteq cl_d (\Val_m Y) = \Val_m Y$.

  The second part follows from the first part and the equation $\Val f
  = (\eta_Y \circ f)^\dagger$.
\end{proof}

We observe that $\eta_X (x) = \delta_x$ is in
$\Val_\fin X \subseteq \Val_m X$ for every dcpo $X$, and every
$x \in X$, whence the following.
\begin{proposition}
  \label{prop:VmX:monad}
  The triple $(\Val_m, \eta, \_^\dagger)$ is a monad on the category
  of dcpos and Scott-continuous maps.
\end{proposition}

\subsection{Tensorial strengths}
\label{sec:tensorial-strengths}

A tensorial strength for a monad $(T, \eta, \_^\dagger)$ is a
collection $t$ of morphisms
$t_{X,Y} \colon X \times TY \to T (X \times Y)$, natural in $X$ and
$Y$, satisfying certain coherence conditions (which we omit, see
\cite{moggi91}.)  We then say that $(T, \eta, \_^\dagger, t)$ is a
\emph{strong} monad.  We will satisfy ourselves with the following
result.  By \cite[Proposition~3.4]{moggi91}, in a category with finite
products and enough points, if one can find morphisms
$t_{X,Y} \colon X \times TY \to T (X \times Y)$ for all objects $X$
and $Y$ such that
$t_{X,Y} \circ \langle x, \nu\rangle = T (\langle x \circ !, \identity
{Y} \rangle) \circ \nu$, then the collection of those morphisms is the
unique tensorial strength.  A category with a terminal object $1$ has
\emph{enough points} if and only if, for any two morphisms
$f, g \colon X \to Y$, $f=g$ if and only if for every
$x \colon 1 \to X$, $f \circ x = g \circ x$.

The category $\Dcpo$ of dcpos has finite products, and has enough
points.  Specializing the above to $T \eqdef \Val$, the formula for
$t_{X,Y}$ reads: for every $x \in X$, for every $\nu \in \Val Y$, for
every $W \in \Open (X \times Y)$,
$t_{X, Y} (x, \nu) (W) = \nu (\{y \in Y \mid (x, y) \in W\}$.
Rewriting this as
$t_{X,Y} (x, \nu) (W) = \int_{y \in Y} \chi_W (x, y) d\nu$, hence
$t_{X, Y} (x, \nu) = \Val (\lambda y \in Y . (x, y)) (\nu)$, we
retrieve formulae already given by Jones \cite[Section~4.3]{jones90},
and which show immediately that the map $t_{X, Y}$ is well-defined and
Scott-continuous.

It also follows from Lemma~\ref{lemma:VXm:dagger} that, if $\nu$ is a
minimal valuation, namely an element of $\Val_m X$, then
$t_{X, Y} (x, \nu)$ is an element of $\Val_m (X \times Y)$.  It
follows:
\begin{proposition}
  \label{prop:VmX:t}
  $(\Val, \eta, \_^\dagger, t)$ and
  $(\Val_m, \eta, \_^\dagger, t)$ are strong monads on $\Dcpo$.
\end{proposition}

\subsection{Minimal valuations form a commutative monad on $\Dcpo$, or
  Fubini-Tonelli again}
\label{sec:min-fubini}

We now show that $\Val_m$ is a \emph{commutative} monad on $\Dcpo$.
The corresponding result is unknown for $\Val$.  Equivalently, the
Fubini-Tonelli theorem holds for minimal valuations.  In order to
prove it, we use the following simple lemma.
\begin{lemma}
  \label{lemma:f=g:cld}
  Two morphisms $f, g \colon X \to Y$ in $\Dcpo$ that coincide on
  $A \subseteq X$ also coincide on $cl_d (A)$.
\end{lemma}
\begin{proof}
  Let $B \eqdef \{x \in X \mid f (x)=g (x)\}$.  Since $f$ and $g$
  preserve directed suprema, $B$ is d-closed.  By assumption, $A$ is
  included in $B$, so $B$ also contains $cl_d (A)$.
\end{proof}

\begin{theorem}[Fubini-Tonelli for minimal valuations on $\Dcpo$]
  \label{thm:fubini:min}
  Let $X$, $Y$ be two dcpos, $f \in \Lform (X \times Y)$, $\mu \in
  \Val X$, and $\nu \in \Val Y$.  If $\mu$ or $\nu$ is minimal, then:
  \begin{align}
    \label{eq:fubini}
    \int_{x \in X} \left(\int_{y \in Y} f (x, y) d\nu\right)d\mu
    & = \int_{y \in Y} \left(\int_{x \in X} f (x, y) d\mu\right)d\nu.
  \end{align}
\end{theorem}
\begin{proof}
  When $\mu$ is a simple valuation $\sum_{i=1}^n a_i \delta_{x_i}$, it
  is clear that the two sides of the equation coincide.  Hence the two
  Scott-continuous functions:
  \begin{align*}
    \mu & \mapsto \int_{x \in X} \left(\int_{y \in Y} f (x, y)
          d\nu\right)d\mu \\
    \mu & \mapsto \int_{y \in Y} \left(\int_{x \in X} f (x, y) d\mu\right)d\nu
  \end{align*}
  coincide on $\Val_\fin X$.  By Lemma~\ref{lemma:f=g:cld}, they coincide
  on $cl_d (\Val_\fin X) = \Val_m X$.  This finishes the case where
  $\mu$ is minimal.  The case where $\nu$ is minimal is symmetric.
\end{proof}

Given a tensorial strength $t$ for a monad $T$ on a category with
finite products, there is a dual tensorial strength $t'$, where
$t'_{X,Y} \colon TX \times Y \to T (X \times Y)$, obtained by swapping
the two arguments, applying $t_{X, Y}$, and then swapping back the
roles of $X$ and $Y$.  Here
$t'_{X,Y} (\mu, y) (W) = \int_{x \in X} \chi_W (x, y) d\mu$.  We can
then define two morphisms from $TX \times TY$ to $T (X \times Y)$,
namely ${t'}^\dagger_{X,Y} \circ t_{TX,Y}$ and
$t^\dagger_{X,Y} \circ t'_{X,TY}$.  The monad $T$ is
\emph{commutative} when they coincide.  The connection between
commutative monads and Fubini-Tonelli-like theorems is made explicit
by Kock \cite{Kock:monad:comm}.  For completeness, we prove the
following explicitly.

\begin{proposition}
  \label{prop:VmX:comm}
  Let $X$, $Y$ be two dcpos.  The maps
  $t^\dagger_{X,Y} \circ t'_{X,\Val Y}$ and
  ${t'}^\dagger_{X,Y} \circ t_{\Val X,Y}$ coincide on those pairs
  $(\mu, \nu) \in \Val X \times \Val Y$ such that
  $\mu \in \Val_m X$ or $\nu \in \Val_m Y$.
\end{proposition}
\begin{proof}
  We have already seen that
  $t_{X,Y} (x, \nu) = \Val (\lambda y \in Y . (x, y)) (\nu)$, for all
  $x$ and $\nu$.  In other words,
  $t_{X, Y} (x, \nu) = (\lambda y \in Y . \delta_{(x,y)})^\dagger
  (\nu)$.  Similarly,
  $t'_{X,Z} (\mu, z) = (\lambda x \in X . \delta_{(x,z)})^\dagger
  (\mu)$.
  
  For every $W \in \Open (X \times Y)$,
  \begin{align*}
    t^\dagger_{X, Y} (t'_{X, \Val Y} (\mu, \nu)) (W)
    & = \int_{(x', \nu) \in X \times \Val Y} t_{X, Y} (x', \nu) (W) d
      t'_{X, \Val Y} (\mu, \nu) \\
    & = \int_{(x', \nu) \in X \times \Val Y} t_{X, Y} (x', \nu) (W) d
      (\lambda x \in X . \delta_{(x,\nu)})^\dagger (\mu) \\
    & = \int_{x \in X} \left(\int_{(x', \nu) \in X \times \Val Y}t_{X,
      Y} (x', \nu) (W) d \delta_{(x,\nu)} \right) d \mu
    & \text{by (\ref{eq:dagger})} \\
    & = \int_{x \in X} t_{X,Y} (x, \nu) (W) d\mu \\
    & = \int_{x \in X} \left(\int_{y \in Y} \chi_W (x, y) d\nu\right)d\mu.
  \end{align*}
  Symmetrically,
  \begin{align*}
    {t'}^\dagger_{X,Y} (t_{\Val X,Y} (\mu, \nu)) (W)
    & = \int_{y \in Y} \left(\int_{x \in X} \chi_W (x, y) d\mu\right)d\nu.
  \end{align*}
  The result then follows from Theorem~\ref{thm:fubini:min}.
\end{proof}

\begin{corollary}
  \label{corl:VmX:comm}
  $(\Val_m, \eta, \_^\dagger, t)$ is a commutative monad on $\Dcpo$.
\end{corollary}


\begin{remark}
  \label{rem:otimes}
  Strictly speaking, the Fubini-Tonelli theorem is more general, and
  states the existence of a product measure.  Here we obtain a minimal
  product valuation, as follows.  We write $\otimes$ for the morphism
  ${t'}^\dagger_{X,Y} \circ t_{\Val X,Y}$
  ($=t^\dagger_{X,Y} \circ t'_{X,\Val Y}$) from
  $\Val_m X \times \Val_m Y$ to $\Val_m (X \times Y)$, as with any
  commutative monad \cite[Section~5]{Kock:monad:comm}.  Then, for all
  $\mu \in \Val_m X$ and $\nu \in \Val_m Y$, $\otimes (\mu, \nu)$,
  which we prefer to write as $\mu \otimes \nu$, is in
  $\Val_m (X \times Y)$.  Looking back at the computations we have
  done during the proof of Proposition~\ref{prop:VmX:comm},
  \begin{align*}
    (\mu \otimes \nu) (W)
    & = \int_{x \in X} \left(\int_{y \in Y} \chi_W (x, y)
      d\nu\right)d\mu \\
    & = \int_{y \in Y} \left(\int_{x \in X} \chi_W (x, y) d\mu\right)d\nu
  \end{align*}
  for every $W \in \Open (X \times Y)$.  An easy computation, based on
  (\ref{eq:dagger}), yields that, for every
  $f \in \Lform (X \times Y)$,
  $\int_{(x, y) \in X \times Y} f (x, y) d (\mu \otimes \nu)$ is equal
  to any of the double integrals of (\ref{eq:fubini}).

  As an additional benefit of the categorical approach, we obtain that
  the map $\otimes \colon (\mu, \nu) \mapsto \mu \otimes \nu$ is
  Scott-continuous from $\Val_m X \times \Val_m Y$ to
  $\Val_m (X \times Y)$.
\end{remark}

\section{Measures on $\real$, minimal valuations on $\IR$}
\label{sec:lebesgue-r-valuation}

All this is good, but aren't we restricting continuous valuations too
much by only considering minimal valuations?  And indeed, Lebesgue
measure is not minimal on $\real$, as we will soon see.  But it
\emph{is} minimal on $\IR$ and on $\IR_\bot$.  In fact, we will see
that \emph{every} measure $\mu$ on $\real$ is represented by a minimal
valuation on $\IR$ and on $\IR_\bot$.

To be more precise, since $\real$ is second-countable hence
hereditarily Lindel\"of, every measure $\mu$ on $\real$ restricts to a
continuous valuation on $\real$.  Its image valuation $i [\mu]$ by the
embedding $i \colon x \mapsto [x, x]$ of $\real$ into $\IR$ (or
$\IR_\bot$) is then a continuous valuation on $\IR$ (resp.,
$\IR_\bot$).  While the former is rarely minimal, we will see that the
latter always is.

Let us call \emph{valuation} on $X$ any strict, modular, monotonic map
from $\Open X$ to $\creal$; namely, we forgo the continuity
requirement.  A valuation $\nu$ on $X$ is \emph{point-continuous} if
and only if for every open subset $U$ of $X$, for every $r \in \realp$
such that $r < \nu (U)$, there is a finite subset $A$ of $U$ such
that, for every open neighborhood $V$ of $A$, $\nu (V) > r$.  The
notion is due to Heckmann \cite{heckmann96}.  His main achievement was
to show that the space $\Val_p X$ of point-continuous valuations over
$X$, with the so-called weak topology, is a sobrification of the space
of simple valuations, also with the weak topology.  This can be used
to show that there is also a commutative monad
$(\Val_p, \eta, \_^\dagger, t)$ on $\Dcpo$, but we will not show this
here.  We will use the following results by Heckmann: every
point-continuous valuation is continuous; every simple valuation is
point-continuous; $\Val_p X$ is sober, in particular it is a dcpo
under the pointwise ordering; in particular, every minimal valuation
is point-continuous.

Let $\lambda$ be the Lebesgue measure on $\real$, or ambiguously, its
restriction to the open subsets of $\real$, in which case we call it
the \emph{Lebesgue valuation}.
\begin{lemma}
  \label{lemma:lambda:notmin}
  The Lebesgue valuation $\lambda$ on $\real$ is not point-continuous,
  hence not minimal.
\end{lemma}
\begin{proof}
  Let $U$ be any non-empty open subset of $\real$.  Then
  $\lambda (U) > 0$.  Let us pick any $r \in \realp$ such that
  $r < \lambda (U)$.  For every finite subset $A$ of $U$, say of
  cardinality $n$, the open set
  $V \eqdef \bigcup_{x \in A} ]x-\epsilon, x+\epsilon[$, where
  $\epsilon > 0$ is chosen so that $2n\epsilon < r$, is an open
  neighborhood of $A$ but
  $\lambda (V) \leq \sum_{x \in A} \lambda (]x-\epsilon, x+\epsilon[)
  = 2n\epsilon < r$.
\end{proof}

\begin{proposition}
  \label{prop:mu:min}
  For every measure $\mu$ on $\real$, and writing again $\mu$ for the
  valuation it induces by restriction to $\Open \real$, the image
  valuation $i [\mu]$ is a minimal valuation on $\IR$ (resp.,
  $\IR_\bot$).
\end{proposition}
\begin{proof}
  $\IR$ (resp., $\IR_\bot$) is a continuous dcpo.  Hence $\Val (\IR)$
  (resp., $\Val (\IR_\bot)$) is also a continuous dcpo, with a basis
  of simple valuations \cite[Theorem~IV-9.16]{gierz03}.  In
  particular, $i [\mu]$ is a directed supremum of simple valuations,
  hence a minimal valuation.
\end{proof}

\begin{remark}
  \label{rem:mu:min}
  The proof of Proposition~\ref{prop:mu:min} shows, more generally,
  that, given any hereditarily Lindel\"of space $X$ with a topological
  embedding $i$ into some continuous dcpo $P$, for every Borel measure
  $\mu$ on $X$, $i [\mu]$ is a minimal valuation on $P$.  Lawson
  showed that every Polish space $X$, not just $\real$, has this
  property \cite{lawson95}.
\end{remark}

$\IR$ and $\IR_\bot$ are even $\omega$-continuous, namely, they have a
countable basis.  One can then show that $i [\mu]$ is the supremum of
a countable chain of simple valuations (Exercise~IV-9.29 of
\cite{gierz03} helps).  It is interesting to see that we can obtain
such an explicit countable chain by elementary means.  The following
construction, which has some common points with the so-called
Riemann-Stieltjes integral, is uniform, and only requires the
knowledge of $\mu (]a, b])$ for intervals $]a, b]$ with rational or
dyadic endpoints.

\begin{definition}
  \label{defn:part}
  A \emph{partition} $P$ (of $\real$) is a finite non-empty subset of
  $\real$.  We write $P$ as $\{a_1 < \cdots < a_n\}$ in order to make
  both the elements and their order manifest, with $n \geq 1$.  Every
  such partition induces a finite subset $E_P$ of $\IR$ defined by
  $E_P \eqdef \{[a_1, a_2], \cdots, [a_{n-1}, a_n]\}$.  If $n=1$, then
  $E_P$ is empty.

  Let $\mathcal P$ be the set of all partitions of $\real$.  A
  partition $Q$ \emph{refines} $P$ if and only if $P \subseteq Q$.

  Given any measure $\mu$ on $\real$, and any partition
  $P \eqdef \{a_1 < \cdots < a_n\}$ of $\real$, we let $\mu_P$ be the
  simple valuation
  $\sum_{i=2}^n \mu (]a_{i-1}, a_i]) . \delta_{[a_{i-1}, a_i]}$ on
  $\IR$ (resp., $\IR_\bot$).
\end{definition}

\begin{remark}
  \label{rem:part}
  The coefficient of $\delta_{[a_{i-1}, a_i]}$ in $\mu_P$ is
  $\mu (]a_{i-1}, a_i])$, not $\mu ([a_{i-1}, a_i])$.  The point is
  that $\mu (]a_{i-1}, a_i]) + \mu (]a_i, a_{i+1}]) = \mu (]a_{i-1},
  a_{i+1}])$.  Choosing the coefficient to be $\mu ([a_{i-1}, a_i[)$
  would work equally well.
\end{remark}

\begin{remark}
  \label{rem:cdf}
  If $\mu$ is a probability valuation on $\real$ given by its
  cumulative distribution function $F$ (namely,
  $F (t) \eqdef \mu (]-\infty, t])$), then $\mu_P$ is equal to
  $\sum_{i=2}^n (F (a_i) - F (a_{i-1})) . \delta_{[a_{i-1}, a_i]}$,
  and represents the process of picking the interval $[a_{i-1}, a_i]$
  with probability $F (a_i) - F (a_{i-1})$.
\end{remark}

\begin{theorem}
  \label{thm:part}
  Let $\mu$ be any measure on $\real$.  Then:
  \begin{enumerate}
  \item for all partitions $P$ and $Q$ of $\real$, if $P \subseteq Q$
    then $\mu_P \leq \mu_Q$;
  \item $i [\mu]$ is the directed supremum of ${(\mu_P)}_{P \in
      \mathcal P}$, and is therefore a minimal valuation;
  \item for every dense subset $D$ of $\real$, $i [\mu]$ is also the
    directed supremum of ${(\mu_P)}_{P \in \mathcal P_D}$, where
    $\mathcal P_D$ is the collection of partitions $P \subseteq D$.
  \end{enumerate}
\end{theorem}
\begin{proof}
  1. Since $Q$ refines $P$, $Q$ is obtained by adding some points to
  $P$, and it suffices to show the claim when we add just one point.
  Let us write $P$ as $\{a_1 < \cdots < a_n\}$, and let
  $Q \eqdef P \cup \{a\}$, where $a \not\in P$.  If $a < a_1$, then
  $\mu_Q$ is equal to
  $\mu_P + \mu (]a, a_1]) . \delta_{[a, a_1]} \geq \mu_P$; similarly
  if $a > a_n$.  Otherwise, let $j$ be the unique index such that
  $a_{j-1} < a < a_j$, $2\leq j \leq n$.  Then:
  \begin{align*}
    \mu_P
    & = \sum_{\substack{i=2\\i\neq j}}^n \mu (]a_{i-1}, a_i])
    . \delta_{[a_{i-1}, a_i]} + \mu (]a_{j-1}, a_j])
    . \delta_{[a_{j-1}, a_j]} \\
    & = \sum_{\substack{i=2\\i\neq j}}^n \mu (]a_{i-1}, a_i])
    . \delta_{[a_{i-1}, a_i]} + \mu (]a_{j-1}, a])
    . \delta_{[a_{j-1}, a_j]} + \mu (]a, a_j])
    . \delta_{[a_{j-1}, a_j]} \\
    & \leq \sum_{\substack{i=2\\i\neq j}}^n \mu (]a_{i-1}, a_i])
    . \delta_{[a_{i-1}, a_i]} + \mu (]a_{j-1}, a])
    . \delta_{[a_{j-1}, a]} + \mu (]a, a_j])
    . \delta_{[a, a_j]} = \mu_Q,
  \end{align*}
  where the latter inequality is justified by the fact that
  $[a_{j-1}, a]$ and $[a, a_j]$ are both larger than or equal to
  $[a_{j-1}, a_j]$, and that $\mathbf x \mapsto \delta_{\mathbf x}$ is
  monotonic.

  Item~2 is a special case of item~3, with $D \eqdef \real$.  Hence we
  prove item~3 directly.  Let $D$ be any dense subset of $\real$.
  
  The family $\mathcal P_D$ is directed under inclusion, since
  $\{d\} \in \mathcal P$ for any given $d \in D$, and since for all
  $P, Q \in \mathcal P_D$, $P \cup Q$ is in $\mathcal P_D$.  By
  item~1, ${(\mu_P)}_{P \in \mathcal P_D}$ is therefore directed.
  
  Let us fix an arbitrary Scott-open subset $\mathcal U$ of $\IR$
  (resp., $\IR_\bot$), and let $U \eqdef i^{-1} (\mathcal U)$.

  For every $P \eqdef \{a_1 < \cdots < a_n\} \in \mathcal P_D$,
  $\mu_P (\mathcal U) \leq i [\mu] (\mathcal U)$.  Indeed,
  \begin{align*}
    \mu_P (\mathcal U)
    & = \sum_{\substack{2 \leq i \leq n\\ [a_{i-1}, a_i] \in \mathcal
    U}} \mu (]a_{i-1}, a_i]) \\
    & = \mu \left(\bigcup_{\substack{2 \leq i \leq n\\ [a_{i-1}, a_i] \in \mathcal
    U}} ]a_{i-1}, a_i]\right)
    & \text{since the sets }]a_{i-1}, a_i]\text{ are pairwise
      disjoint} \\
    & \leq \mu (U) = i [\mu] (\mathcal U).
  \end{align*}
  The last inequality is justified by the fact that for every $i$ with
  $2 \leq i\leq n$ such that $[a_{i-1}, a_i] \in \mathcal U$,
  $[a_{i-1}, a_i]$ is included in $U$: for every
  $x \in [a_{i-1}, a_i]$, $[a_{i-1}, a_i]$ is below $[x, x]$ in $\IR$
  (resp., $\IR_\bot$), so $[x, x] \in \mathcal U$, meaning that
  $x \in U$.

  In order to show that
  $\sup_{P \in \mathcal P_D} \mu_P (\mathcal U) \geq i [\mu] (\mathcal
  U)$, we consider any $r \in \realp$ such that
  $r < i [\mu] (\mathcal U) = \mu (U)$, and we will show that there is
  a partition $P \in \mathcal P_D$ such that $\mu_P (\mathcal U) > r$.
  For every $x \in U$, $[x, x]$ is in $\mathcal U$ and is the supremum
  of the chain of elements $[x-\epsilon, x+\epsilon]$, $\epsilon > 0$.
  Therefore there is a positive real number $\epsilon_x$ such that
  $[x-\epsilon_x, x+\epsilon_x]$ is in $\mathcal U$.  Since $D$ is
  dense in $\real$, we can find a subinterval $[a_x, b_x]$ such that
  $a_x$ and $b_x$ are in $D$, and $a_x < x < b_x$.  We note that
  $[a_x, b_x]$ is also in $\mathcal U$.  In particular, $[a_x, b_x]$,
  hence also $]a_x, b_x[$, is included in $U$ for every $x \in U$, so
  $U = \bigcup_{x \in U} ]a_x, b_x[$.  We write the latter as the
  directed union of the open sets
  $U_E \eqdef \bigcup_{x \in E} ]a_x, b_x[$, where $E$ ranges over the
  finite subsets of $U$.  Since $\mu$ restricted to the open sets of
  $\real$ is a continuous valuation, there is a finite subset $E$ of
  $U$ such that $r < \mu (U_E)$.
  
  We let $P$ be the collection of elements $a_x$ and $b_x$, $x \in E$.
  This is an element of $\mathcal P_D$.  Let us write $P$ as
  $\{c_1 < \cdots < c_n\}$.  For every $x \in E$, $[a_x, b_x]$ is an
  element of $\mathcal U$, and is equal to $[c_{i_x}, c_{j_x}]$ for
  some indices $i_x$ and $j_x$ such that $1 \leq i_x < j_x \leq n$.
  Then $[c_{i-1}, c_i]$ is larger than or equal to $[a_x, b_x]$ for
  every $i$ with $i_x+1 \leq i \leq j_x$, hence is also in
  $\mathcal U$.  It follows that
  $\mu_P (\mathcal U) = \sum_{\substack{2 \leq i \leq n\\{} [c_{i-1},
      c_i] \in \mathcal U}} \mu (]c_{i-1}, c_i]) \geq
  \sum_{\substack{x \in E\\i_x+1 \leq i \leq j_x}} \mu (]c_{i-1},
  c_i])$.  The latter is equal to
  $\mu (\bigcup_{\substack{x \in E\\i_x+1 \leq i \leq j_x}} ]c_{i-1},
  c_i])$, since the intervals $]c_{i-1}, c_i]$ are pairwise disjoint.
  By definition of $i_x$ and $j_x$, this is larger than or equal to
  $\mu (\bigcup_{x \in E} ]a_x, b_x])$, hence to
  $\mu (\bigcup_{x \in E} ]a_x, b_x[) = \mu (U_E)$.  Therefore
  $r < \mu_P (\mathcal U)$, as desired.
\end{proof}

\begin{corollary}
  \label{corl:RF:Rval:countable}
  Let $\mu$ be any measure on $\real$.  For every $n \in \nat$, let
  the partition $P_n$ consist of all integer multiples of $1/2^n$
  between $-n$ and $n$.  The minimal valuation $i [\mu]$ is the
  supremum of the countable chain of simple valuations $\mu_{P_n}$,
  $n \in \nat$.
\end{corollary}
\begin{proof}
  The family ${(P_n)}_{n\in \nat}$ is a subfamily of $\mathcal P_D$,
  where $D$ is the set of dyadic numbers.  It is in fact cofinal:
  every element $P$ of $\mathcal P_D$ is refined by $P_n$ for some
  $n \in \nat$, namely for any $n$ larger than every element of $P$,
  every opposite of an element of $P$, and such that every element of
  $P$ is an integer multiple of $1/2^n$.  It follows that
  ${(\mu_P)}_{p \in \mathcal P_D}$ and ${(\mu_{P_n})}_{n \in \nat}$
  have the same supremum, and that is $i [\mu]$ by
  Theorem~\ref{thm:part}, item~3.
\end{proof}

\section{The ISPCF Calculus}
\label{sec:simple-sfpc-calculus}

We now come to the description of our calculus.  This is a variant of
SFPC, a statistical variant of Fiore and Plotkin's Fixed Point
Calculus \cite{fiore94a} due to V{\'a}k{\'a}r, Kammar, and Staton
\cite{VKS:SFPC}.  It is also very close to the calculus PCFSS of
\cite{DLH:geom:bayes} and to PPCF \cite{EPT:PPCF}.  Just like the
latter, but contrarily to the other languages we have mentioned, our
calculus ISPCF (for \emph{Interval} statistical PCF) is a call-by-name
language.

Its algebra of types is as follows.  One may naturally consider
additional types, such as more complex, recursively defined data
types.
\begin{align*}
  \sigma, \tau, \ldots
  & ::= \unitT  \mid \voidT \mid \intT \mid
    \realT
    \mid \sigma + \tau
    \mid \sigma \times \tau
    \mid \sigma \to \tau \mid D\tau
\end{align*}
The types of the form $D \tau$ are called \emph{distribution types}.
The function arrow $\to$ associates to the right, so that
$\sigma \to \tau \to \lambda$ abbreviates
$\sigma \to (\tau \to \lambda)$.  Compared to the algebra of types of
SFPC, first, we keep a type $\realT$ of real numbers---exact reals
instead of true reals, though.
Second, we do not require the presence of recursive types as SFPC, so
as to simplify the presentation.  Accommodating such type
constructions is not central to our work, and can be handled through
bilimit constructions \cite[Section~5]{abramsky94}.  (Showing
adequacy, as we will do in Section~\ref{sec:adequacy}, in the presence
of recursive types, is much more challenging, though.  See
\cite{JLM:prob:quant} for adequacy in the presence of recursive types
and discrete probabilistic choice, and \cite{pitts93a} for a general
technique for proving adequacy in the presence of recursive types.)
Third, and finally, we choose to have an explicit type of
distributions former $D$, in the style of Moggi's Simple Metalanguage
\cite{moggi91}: arrow types $\sigma \to \tau$ in SFPC are typically
encoded as $\sigma \to D\tau$ in ISPCF.

\begin{figure}
  \centering
  \[
    \begin{array}{cc}
      \begin{prooftree}
        \strut
        \justifies
        \vdash x_\sigma \colon \sigma
      \end{prooftree}
      &
      \begin{prooftree}
        \underline f \in \Sigma
        \justifies
        \underline f \colon typ (\underline f)
      \end{prooftree}
        \qquad
        \begin{prooftree}
          \vdash M \colon \sigma \to \sigma
          \justifies
          \vdash \reckw\; M \colon \sigma
        \end{prooftree}
      \\ \\
      \begin{prooftree}
        \justifies
        \vdash \sample [\mu] \colon D \realT
      \end{prooftree}
      &
        \begin{prooftree}
          \vdash M \colon \realT
          \justifies
          \vdash \score M \colon D \unitT
        \end{prooftree}
      \\ \\
      \begin{prooftree}
        \vdash M \colon \tau
        \justifies
        \vdash \lambda x_\sigma . M \colon \sigma \to \tau
      \end{prooftree}
      &
        \begin{prooftree}
          \vdash M \colon \sigma \to \tau \quad
          \vdash N \colon \sigma
          \justifies
          \vdash MN \colon \tau
        \end{prooftree}
      %
      \\ \\
      \begin{prooftree}
        \vdash M \colon \tau
        \justifies
        \vdash \retkw M \colon D \tau
      \end{prooftree}
      &
        \begin{prooftree}
          \vdash M \colon D \sigma \quad
          \vdash N \colon D \tau
          \justifies
          \vdash \dokw {x_\sigma \leftarrow M}; N \colon D \tau
        \end{prooftree}
      \\ \\
      \begin{prooftree}
        \vdash M \colon \sigma \quad
        \vdash N \colon \tau
        \justifies
        \langle M, N \rangle \colon \sigma \times \tau
      \end{prooftree}
      &
      \begin{prooftree}
        \vdash M \colon \sigma \times \tau
        \justifies
        \vdash \pi_1 M \colon \sigma
      \end{prooftree}
        \quad
      \begin{prooftree}
        \vdash M \colon \sigma \times \tau
        \justifies
        \vdash \pi_2 M \colon \tau
      \end{prooftree}
      \\ \\
      \begin{prooftree}
        \begin{array}{c}
          \strut\\
          \vdash M \colon \sigma
        \end{array}
        \justifies
        \vdash \iota_1 M \colon \sigma+\tau
      \end{prooftree}
      \quad
      \begin{prooftree}
        \begin{array}{c}
          \strut\\
          \vdash M \colon \tau
        \end{array}
        \justifies
        \vdash \iota_2 M \colon \sigma+\tau
      \end{prooftree}
      &
        \begin{prooftree}
          \begin{array}{c}
            \vdash M \colon \sigma + \tau
            \\
            \vdash N \colon \sigma \to \upsilon \quad
            \vdash P \colon \tau \to \upsilon
          \end{array}
          \justifies
          \casekw M N P \colon \upsilon
        \end{prooftree}
    \end{array}
  \]
  \caption{The typing rules of ISPCF}
  \label{fig:typing}
\end{figure}

Just like SFPC, ISPCF is a typed higher-order functional language.
Its terms are given by the following grammar:
\begin{align*}
  M, N, P, \ldots
  & ::= x_\sigma \mid \underline f \mid \sample [\mu] \mid \score M \\
  & \quad \mid \lambda x . M \mid MN \mid \reckw\; M 
  \mid \retkw M \mid \dokw {x \leftarrow M}; N \\
  & \quad \mid \langle M, N \rangle \mid \pi_1 M \mid \pi_2 M
    \mid \iota_1 M \mid \iota_2 M \mid \casekw M N P
\end{align*}
where $x_\sigma$ (and also $y_\tau$, $z_\upsilon$, etc.) range over
variables, whose type appears explicitly as a subscript (which we will
often omit for clarity), $\underline f$ ranges over a fixed set
$\Sigma$ of constants, and $\mu$ ranges over a fixed set $\mathcal M$
of measures on $\real$.  We assume a countably infinite supply of
variables of each type.

We assume a function $typ$ from $\Sigma$ to the set of ISPCF types:
$typ (\underline f)$ is the type of $\underline f$, by definition.

Amongst those we require a constant $\underline *$ of type $\unitT$,
and constants $\underline n$, one for each integer $n \in \Z$, with
$typ (\underline n) = \intT$.

We may include constants $\underline q$, one for each rational number
$q$, with $typ (\underline q) = \realT$.  We may also include
constants for non-rational numbers, such as $\underline\pi$; and
standard numerical primitives such as $\underline\sin$,
$\underline{\mathrm{exp}}$ or $\underline{\log}$, with
$typ (\underline\sin) = typ(\underline{\mathrm{exp}}) =
typ(\underline{\log}) = \realT \to \realT$, meant to implement the
sine, exponential and logarithm functions respectively, as in SFPC.

We will see later that we can define a type $\boolT$ of Booleans,
namely $\voidT + \voidT$, and two terms $\true$ and $\false$ of type
$\boolT$.  Given that, we may include a sign test
$\poskw \colon \realT \to \boolT$, with the intention that $\poskw\; M$
evaluates to the value $\trueop$ of $\true$ if $M$ evaluates to a
positive real number, and to the value $\falseop$ of $\false$ if $M$
evaluates to a negative real number (we will discuss semantics soon).

The constructs $\retkw$ and $\dokw$ are standard monadic constructs;
$\reckw$ is a fixed-point combinator, and we assume the usual
$\alpha$-renaming conventions.  Crucially, we keep the two iconic
primitives of SFPC: $\sample$ (or rather, $\sample [\mu]$) and
$\score$.  The typing rules are given in Figure~\ref{fig:typing}.

We also write $M; N$ for
$\dokw {x \leftarrow M}; N$, where $x$ is a fresh variable---one that
is not free in $N$.

While SFPC is call-by-value, ISPCF is a call-by-name language, just
like Hakaru \cite{SR:disin} or PPCF \cite{EPT:PPCF}.  (Hakaru,
however, lacks higher-order functions and recursion.)  This slightly
simplifies the denotational semantics, and does not preclude us from
using a call-by-value style of programming if we so desire.  Namely,
call-by-value application $M @ N$ of $M \colon \sigma \to D \tau$ to
$N \colon D \sigma$ is defined as $\dokw {x \leftarrow N}; Mx$, where
$x$ is a fresh variable.  We will return to this point in
Section~\ref{sec:call-value-question}.

\section{The Denotational Semantics of ISPCF}
\label{sec:semantics}

The semantics $\Eval \tau$ of types $\tau$ is given by induction on $\tau$:
\[
  \begin{array}{c}
    \Eval {\unitT} \eqdef \{*, \bot\} \quad
    \Eval {\voidT} \eqdef \{\bot\} \quad
    \quad \Eval {\intT} \eqdef \Z_\bot \quad
    \Eval {\realT} \eqdef \IR_\bot \\
    \Eval {\sigma \to \tau} \eqdef [\Eval \sigma \to \Eval \tau]
    \quad    
    \Eval {D \tau} \eqdef \Val_m (\Eval \tau)
    \\
    \Eval {\sigma \times \tau} \eqdef \Eval \sigma \times \Eval \tau
    \quad
    \Eval {\sigma + \tau} \eqdef (\Eval \sigma + \Eval \tau)_\bot.
  \end{array}
\]
In $\Eval {\unitT}$, we have $\bot < *$.  $\Z_\bot$ is $\Z$ plus a
fresh element $\bot$, ordered so that $\bot$ is least and all integers
are pairwise incomparable.  $[X \to Y]$ denotes the dcpo of
Scott-continuous maps from the dcpo $X$ to the dcpo $Y$, ordered
pointwise.  Given any two pointed dcpos $X$ and $Y$, $X \times Y$ is
their ordinary product, ordered pointwise.  $X + Y$ is their
coproduct, and consists of elements $(1, x)$ with $x \in X$, $(2, y)$
with $y \in Y$, so that $(1, x) \leq (1, x')$ if and only if
$x \leq x'$, $(2, y) \leq (2, y')$ if and only if $y \leq y'$, and
$(1, x)$ and $(2, y)$ are incomparable.

\begin{figure}
  \centering
  \begin{align*}
    \Eval {x_\sigma} \rho & \eqdef \rho (x_\sigma)
    & \Eval {\underline f} \rho & \eqdef f \\
    \Eval {\sample [\mu]} \rho & \eqdef i [\mu]
    & \Eval {\score M} \rho & \eqdef 
                              |\Eval M \rho| . \delta_* \\
    \Eval {\lambda x_\sigma . M} \rho & \eqdef \lambda V \in \Eval \sigma . \Eval M (\rho [x_\sigma
                                        \mapsto V]) \mskip-180mu \\
    \Eval {MN} \rho & \eqdef \Eval M \rho (\Eval N \rho) \mskip-20mu
    & \Eval {\reckw\; M} \rho & \eqdef \lfp (\Eval M \rho) \\
    \Eval {\retkw M} \rho & \eqdef \delta_{\Eval M \rho}
    & \Eval {\dokw {x_\sigma \leftarrow M}; N} \rho
                                & \eqdef
                                  (\Eval {\lambda x_\sigma . N}
                                  \rho)^\dagger (\Eval M \rho)
    \\
    \Eval {\langle M, N \rangle} \rho & \eqdef
                                          (\Eval M \rho, \Eval N \rho)
    & \Eval {\pi_i M} \rho & \eqdef \lambda (V_1, V_2) . V_i
    \\
    \Eval {\casekw M N P} \rho
                          & \eqdef \left\{
                            \begin{array}{ll}
                              \Eval N \rho (a) & \text{if }\Eval M
                                             \rho=(1, a) \\
                              \Eval P \rho (b) & \text{if }\Eval M
                                             \rho=(2, b) \\
                              \bot & \text{otherwise}
                            \end{array}
                                     \right. \mskip-80mu
    & \Eval {\iota_i M} \rho & \eqdef (i, \Eval M \rho)
  \end{align*}
  \caption{The semantics of ISPCF}
  \label{fig:sem}
\end{figure}

The denotational semantics $\Eval M \rho$ of each term $M$, in the
environment $\rho$, is given in Figure~\ref{fig:sem}.  Environments
map variables $x_\sigma$ to elements of the corresponding dcpo
$\Eval \sigma$.  The notation $\rho [x \mapsto V]$ stands for the
environment that maps $x$ to $V$ and every other variable $y$ to
$\rho (y)$.  The least fixed point operator $\lfp$ on every pointed
dcpo $X$ is such that
$\lfp (f) \eqdef \bigcup_{n \in \nat} f^n (\bot)$ for every
$f \in [X \to X]$.  In the denotation of $\casekw M N P$,
$\Eval N \rho (a)$ is the application of $\Eval N \rho$, which is a
function, to the value $a$.

Let us make a few comments.  We assume an element $f \in \Eval \tau$
for each constant $\underline f \in \Sigma$ of type $\tau$.  Hence,
for example, we have elements
$\pi \in \Eval {\realT}$,
$\sin, \mathrm{exp}, \log \in \Eval {\realT \to \realT}$, etc.
The case of sign tests $\poskw$ is interesting: by an easy argument
based on the fact that the real line is connected, there is no
Scott-continuous function from $\IR_\bot$ to $\Eval {\boolT}$, where
$\boolT \eqdef \unitT+\unitT$, which would map the elements $[a, a]$
with $a \geq 0$ to $\trueop \eqdef (1, *)$, and those with $a < 0$ to
$\falseop \eqdef (2, *)$ (see \cite[Proposition~2.4]{EE:int:PCF} for a
generalization of this observation).  And indeed, such tests are not
computable.  One reasonable choice for the semantics of $\poskw$, if
we decide to include it, is to define $\posop ([a, b])$ as $\trueop$
if $a > 0$, $\falseop$ if $b < 0$, and $\bot$ otherwise.

The semantics of $\sample [\mu]$, namely $i [\mu]$, is the minimal
valuation associated with the measure $\mu$ on $\real$, see
Theorem~\ref{thm:part}.

For every $\mathbf a \in \IRbb$, the notation $|\mathbf a|$ denotes
$0$ if $\mathbf a = \bot$ or if $\mathbf a = [a, b]$ with
$a \leq 0 \leq b$, $a$ if $\mathbf a = [a, b]$ with $a \geq 0$, and
$-b$ if $\mathbf a = [a,b]$ with $b\leq 0$.

This, as well as the semantics of $\poskw$, is an instance of a more
general, well-known domain-theoretic construction.  A \emph{bc-domain}
is a continuous dcpo in which every subset with an upper bound has a
least upper bound.  The bc-domains with their Scott topology are
exactly the \emph{densely injective} $T_0$ topological spaces, namely
the $T_0$ spaces $Z$ such that, for every dense subset $X$ of any
space $Y$, every continuous map $f \colon X \to Z$ extends to a
continuous map from the whole of $Y$ to $Z$
\cite[Proposition~II-3.11]{gierz03}.  There is even a pointwise
largest such extension $f^*$, defined by
$f^* (y) \eqdef \sup_{U \in \Open Y, y \in U} \inf_{x \in U \cap X} f
(x)$ \cite[Proposition~II-3.9]{gierz03}.  The latter formula makes
sense and defines a continuous map $f^*$ even when $f$ is not
continuous, but in that case $f^*$ will fail to extend $f$.

One can check that $\Eval {\unitT}$, $\Eval {\voidT}$,
$\Eval {\intT}$, $\Eval {\realT}$, and in fact $\Eval \tau$ for any
type $\tau$ that does not contain the $D$ operator, are bc-domains.
In particular, every continuous map $g \colon \real \to \real$ induces
a continuous map
$\hat g \eqdef (i \circ g)^* \colon \IR_\bot \to \IR_\bot$ that
extends $i \circ g$, namely: $\hat g ([x, x]) = [g (x), g (x)]$ for
every $x \in \real$.  In doing so, we silently equate $\real$ with a
subspace of $\IR_\bot$ through $i$, and it is easy to realize that
$\real$ is indeed dense in $\IR_\bot$.

This applies to $\sin$, $\cos$, $\exp$, notably. The case of $\log$ is
slightly different, since it is only defined on $\realp \diff \{0\}$.
In that case, the map $f \colon \real \to \IR_\bot$ that sends every
positive real number $x$ to $i (\log x) = [\log x, \log x]$, and all
other numbers to $\bot$, is continuous, and it is natural to define
$\log \colon \IR_\bot \to \IR_\bot$ as its largest continuous
extension. The notation $|\mathbf a|$, as used in the semantics of
$\score\; M$ (where $\mathbf a$ is $\Eval M \rho$), is equal to $f^*$
where $f$ is the absolute value map from $\real$ to $\real$, and
therefore extends it to $\IR_\bot$. The semantics $\posop$ of $\poskw$
is $f^*$, where $f$ is the non-continuous map that sends every
positive number to $\trueop$ and all other numbers to $\falseop$.
Hence $f^*$ does not extend $f$, but we still have
$f^* ([a,a]) = f (a)$ if $a \neq 0$.

It is clear that for each derivable judgment $\vdash M \colon \tau$,
for every environment $\rho$, $\Eval M \rho$ is a well-defined element
of $\Eval \tau$. In particular, the least fixed point
$\lfp (\Eval M \rho)$ is well-defined in the definition of
$\Eval {\reckw\; M} \rho$, because $\Eval \tau$ is a pointed dcpo; this
is shown by induction on $\tau$, and the least element of
$\Eval {D \tau}$ is the constant zero valuation. We also mention the
following case, which shows where tensorial strengths are
required. There is a morphism:
\[ \xymatrix{ T X \times [X \to TY] \ar[r]^{t'_{X,Y}} & T (X \times [X
    \to TY]) \ar[r]^(0.7){\App^\dagger} & TY }
\]
where $X \eqdef \Eval \sigma$, $Y \eqdef \Eval \tau$, $T$ is the
$\Val_m$ monad, $t'$ is its dual tensorial strength, and $\App$ is
application.  Applying this to
$(\Eval M \rho, \Eval {\lambda x_\sigma . N} \rho)$ yields the
semantic value $\Eval {\dokw {x_\sigma \leftarrow M}; N} \rho$.


\begin{remark}
  \label{rem:comm}
  By standard category-theoretic arguments, since $\Val_m$ is
  commutative, the following equation holds when $x$ is not free in
  $N$ and $y$ is not free in $M$:
  \[
    \mskip-40mu
    \Eval {\dokw {x \leftarrow M}; \dokw {y \leftarrow N}; P} \rho =
    \Eval {\dokw {y \leftarrow N}; \dokw {x \leftarrow M}; P} \rho.
  \]
  It was already argued \cite[Comments after Corollary~6.11]{VKS:SFPC}
  that this semantic equation is crucial in program transformation
  techniques such as disintegration-based Bayesian inference
  \cite{SR:disin} as implemented in the Hakaru system
  \cite{NCRSZ:hakaru}.  As in \cite{VKS:SFPC}, this equation is valid
  at all types in our semantics.
\end{remark}

\section{Examples}
\label{sec:examples}

\subsection{The call-by-value versus call-by-name question}
\label{sec:call-value-question}

A recurring paradox in the theory of higher-order probabilistic
programming languages is the following: what should the value of
$(\lambda x . x=x) (\mathtt{rand} ())$ be, where $\mathtt{rand} ()$
returns a random number?  Since $x=x$ is universally true, one expects
this program to always return true.  In call-by-name, one also expects
$\beta$-reduction to be a sound reduction rule, so that this value
should also be equal to $\mathtt{rand} () = \mathtt{rand} ()$.
Unfortunately, the latter draws \emph{two} independent random numbers,
and may return false.

This problem is avoided in \emph{call-by-value} languages, which force
$\mathtt{rand} ()$ to be evaluated first, and only once.

This need not be a problem in call-by-name languages such as ISPCF.
This is well-known, but we would like this point to be clear.  In the
call-by-name language PPCF, the issue is solved by using a
$\mathtt{let} (x, M, N)$ construct that evaluates $M$, binds the
resulting value to $x$ and then evaluates $N$
\cite[Example~3.10]{EPT:PPCF}.  In ISPCF, $\dokw$ may serve a similar
purpose, as we now demonstrate.

Let us build a simple uniform random number generator
$\mathtt{randbool}$ on $\boolT$.  We recall that
$\boolT \eqdef \unitT + \unitT$.  We define $\true$ as
$\iota_1 \underline *$; its semantics is $\trueop = (1, *)$; and
$\false$ as $\iota_2 \underline *$, whose semantics is
$\falseop = (2, *)$.  If $M$ has type $\boolT$, we also define
$\ifkw M N P$ as an abbreviation for
$\casekw M (\lambda x_1 .\allowbreak N) (\lambda x_2 . P)$, where
$x_1$ and $x_2$ are fresh variables.

In order to implement $\mathtt{randbool}$, we assume we have a term
$\sample [\lambda_1]$, where $\lambda_1$ is the uniform measure on
$[0, 1]$, namely: for every measurable subset $E$ of $\real$,
$\lambda_1 (E) \eqdef \lambda (E \cap [0, 1])$.  We also assume the
$\poskw$ primitive, and subtraction $-$.  We define
$\mathtt{randbool}$ as follows.
\begin{align*}
  \mathtt{randbool} & \eqdef \dokw {x_{\realT} \leftarrow \sample
                      [\lambda_1]};
                      \retkw (\poskw (x - \underline {0.5}))
\end{align*}
For every environment $\rho$, $\Eval {\mathtt{randbool}} \rho$ is the
continuous valuation
$\frac 1 2 \delta_{\trueop} + \frac 1 2 \delta_{\falseop}$.

Let
$\mathtt{eqbool} \eqdef \lambda x . \lambda y . \ifkw \; x \; {(\ifkw
  \;y\; \true\; \false)}\; {(\ifkw \;y \; \false\; \true)}$, and we
abbreviate $\mathtt{eqbool}\; M N$ as $M=N$, for all
$M, N \colon \boolT$.  One cannot write
$(\lambda x_{\boolT} . x=x) (\mathtt{randbool})$ directly, because
$\mathtt{randbool}$ has type $D \boolT$, not $\boolT$.  Sampling from
a distribution must be done through $\dokw$, and we have (at least)
two choices.  The term:
\begin{align*}
  \dokw {b_{\boolT} \leftarrow \mathtt{randbool}};\; {\retkw ((\lambda
  x_{\boolT} . x=x) b)}
\end{align*}
first samples, obtaining a random Boolean value $b$, which is then
passed to $\lambda x_{\boolT} . x=x$.  This term has semantics
$\delta_{\trueop}$; in other words, it is a probabilistic process that
returns $\trueop$ with probability one.  Or we may defer sampling
as follows:
\begin{align*}
  (\lambda x_{D \boolT} . \dokw {b_1 \leftarrow x}; \dokw {b_2
  \leftarrow x}; \retkw (b_1=b_2)) \mathtt{randbool}.
\end{align*}
The semantics of that term is $\frac 1 2 \delta_{\trueop} + \frac 1 2
\delta_{\falseop}$: the probability that the two independent Boolean
values $b_1$ and $b_2$ are equal is only $1/2$.

\subsection{Rejection sampling}
\label{sec:rejection-sampling}

Let us imagine that you have a probability measure $\mu$ on a space
$X$, let $A$ be a measurable subset of $X$ such that $\mu (A) \neq 0$.
Then there is a conditional probability measure $\mu_A$, defined by
$\mu_A (E) \eqdef \mu (E) / \mu (A)$ for every measurable subset $E$
of $A$.  Rejection sampling allows one to draw an element at random in
$A$ with probability $\mu_A$, by drawing an element in $X$ with
probability $\mu$ until we find one in $A$.  This terminates with
probability $1$, and in fact in $1/\mu (A)$ trials on average.

Let $\tau$ be any type.  
We consider the following ISPCF term:
\begin{align*}
  \mathtt{rej}
  & \eqdef \lambda p_{D\tau} . \lambda sel_{\tau \to D\boolT} . \\
  & \qquad\reckw (\lambda r_{D\tau} .
    \dokw {x_\tau \leftarrow p 
    }; \\
  & \qquad\qquad\qquad\quad \dokw {b_{\boolT} \leftarrow sel\; 
    x 
    }; \\
  & \qquad\qquad\qquad\quad \ifkw \;{b 
    } \;{(\retkw
    x
    )} \; {r
    }).
\end{align*}
Intuitively, $\mathtt{rej}$ takes a distribution $p$ as input, a
selection function $sel$, and draws $x$ at random with probability $p$
until $sel\; x$ returns $\trueop$.  We allow $sel$ to return an object
of type $D\boolT$, not $\boolT$, in order to allow $sel$ to do some
random computations as well.

For every environment $\rho$, $\Eval {\mathtt{rej}} \rho$ is a
function that takes a minimal valuation $\mu$ (bound to the variable
$p$) and a selection function $s \in \Eval {\tau \to D\boolT}$ (bound
to $sel$) to a minimal valuation $\nu_{|sel}$, which we now elucidate.
By definition, $\nu_{|sel}$ is the least minimal valuation such that
$\nu_{|sel} = \Phi (\nu_{|sel})$, where:
\begin{align*}
  \Phi (\nu) & = \Eval {\dokw {x \leftarrow p}; \dokw {b \leftarrow
             sel\;x};
             \ifkw\; b \; {(\retkw x)} \;r}
             \rho [p \mapsto \mu, sel \mapsto s, r \mapsto \nu]
\end{align*}
For every selection function $s \in \Eval {\tau \to D\boolT}$, for
every $x \in \Eval \tau$, we may write $s (x)$ as
$s_{\trueop} (x) . \delta_{\trueop} + s_{\falseop} (x) .
\delta_{\falseop} + s_{(1, \bot)} (x) . \delta_{(1, \bot)} + s_{(2,
  \bot)} . \delta_{(2, \bot)} + s_\bot (x) . \delta_\bot$.  For
simplicity, let us assume that
$s_\bot (x)=s_{(1, \bot)}(x)=s_{(2, \bot)}(x)=0$.  In other words, the
probability of $s (x)$ returning a non-terminating Boolean is zero.
Expanding the semantics shows that $\Phi (\nu)$ is the continuous
valuation such that, for every $U \in \Open {\Eval \tau}$,
\begin{align*}
  \Phi (\nu) (U) & =
                   \int_{x \in \Eval \tau} (s_{\trueop} (x) . \chi_U (x)
                   + s_{\falseop} (x) . \nu (U)) d\mu.
\end{align*}


Given any space $X$, any continuous valuation $\nu$ on $X$, and any
continuous map $g \colon X \to \creal$, there is a continuous
valuation $g \cdot \nu$ on $X$ defined by
$(g \cdot \nu) (U) \eqdef \int_{x \in X} g (x) \chi_U (x) d\nu$, for
every $U \in \Open X$.  This formula characterizes $g \cdot \nu$ as
the continuous valuation obtained from $\nu$ by applying the
\emph{density function} $g$; a notation that is sometimes used for
$g \cdot \nu$ is $g d\nu$.


This abbreviation allows us to rewrite $\Phi (\nu) (U)$ as follows.
The value $(s_{\falseop} \cdot \mu) (\Eval \tau)$ is the total mass
$\int_{x \in \Eval \tau} s_{\falseop} (x) d \mu$ of
$s_{\falseop} \cdot \mu$ on $\Eval \mu$.
\begin{align*}
  \Phi (\nu) & = s_{\trueop} \cdot \mu + (s_{\falseop} \cdot \mu)
               (\Eval \tau) . \nu.
\end{align*}
Let us solve the equation $\nu = \Phi (\nu)$ for $\nu$.  One may
compute the least solution as $\sup_{n \in \nat} \Phi^n (0)$, where
$0$ denotes the zero valuation, but it is easier to solve the equation
directly.

If $(s_{\falseop} \cdot \mu) (\Eval \tau) < 1$, then there are at most
two solutions to the equation $\nu = \Phi (\nu)$.  The smaller one is:
\begin{align*}
  \frac 1 {1 - (s_{\falseop} \cdot \mu) (\Eval \tau)} s_{\trueop} \cdot \mu,
\end{align*}
and this is the value of $\Eval {\mathtt{rej}} \rho (\mu) (s)$.  When
$0 < (s_{\falseop} \cdot \mu) (\Eval \tau) < 1$, there is another
solution, and that is the continuous valuation that maps every
non-empty open subset of $\Eval \tau$ to $+\infty$, the \emph{largest}
continuous valuation on $\Eval \tau$.  (Yes, that is a minimal
valuation, since it is simply the directed supremum of all simple
valuations whatsoever.)

If $(s_{\falseop} \cdot \mu) (\Eval \tau) = 1$, then we are typically
in the case of sampling with respect to a set of measure zero.  The
sitation is more complicated in general, since
$s_{\falseop} + s_{\trueop}$ may not be equal to $1$: the condition
$(s_{\falseop} \cdot \mu) (\Eval \tau) = 1$ imposes no constraint on
$s_{\trueop} \cdot \mu$.

The general case where $(s_{\falseop} \cdot \mu) (\Eval \tau) \geq 1$
is probably even less intuitive.  Let
$A \eqdef (s_{\falseop} \cdot \mu) (\Eval \tau)$.  The equation
$\nu = \Phi (\nu)$ reduces to $\nu = s_{\trueop} \cdot \mu + A.\nu$.
The largest continuous valuation is always a solution, but there is
smaller one.  For every open subset $U$ of $\Eval \tau$, the possible
values $x$ for $\nu (U)$ satisfy
$x = (s_{\trueop} \cdot \mu) (U) + A.x$: if
$(s_{\trueop} \cdot \mu) (U) = 0$, then the least value for $x$ is $0$
(this is the only possible value if $A > 1$; if $A=1$, every element
of $\creal$ is a possible value); otherwise, the only possible value
for $x$ is $+\infty$.
This suggests that the least solution $\nu$ to
$\nu = \Phi (\nu)$ is defined by:
\begin{align*}
  \nu (U) & = \left\{
            \begin{array}{ll}
              0 & \text{if }(s_{\trueop} \cdot \mu) (U)=0 \\
              +\infty & \text{otherwise.}
            \end{array}
            \right.
\end{align*}
This is just the continuous valuation
$(+\infty).(s_{\trueop} \cdot \mu)$.  It is a minimal valuation by
Lemma~\ref{lemma:VXm:drag}.  Hence the desired semantics
$\Eval {\mathtt{rej}} \rho$ is $(+\infty).(s_{\trueop} \cdot \mu)$ if
$(s_{\falseop} \cdot \mu) (\Eval \tau) \geq 1$.


In the special case where $s$ is deterministic, namely that
$s_{\trueop} = \chi_U$ and $s_{\falseop} = \chi_V$ for two disjoint
open subsets $U$ and $V$ of $\Eval \tau$, then $s_{\trueop} \cdot \mu$
maps every open set $W$ to $\mu (U \cap W)$, and is therefore equal to
the restriction $\mu_{|U}$ of $\mu$ to $U$, as introduced by Heckmann
\cite{heckmann96}.  In that case, $\Eval {\mathtt{rej}} \rho$ is equal
to $\frac 1 {1 - \mu (V)} \mu_{|U}$ if $\mu (V) < 1$, and to
$(+\infty).\mu_{|U}$ otherwise.

\subsection{Using $\score$}
\label{sec:using-score}


Using $\score$ instead of recursion, we can implement rejection
sampling by the following ISPCF term:
\begin{align*}
  \mathtt{rej}'
  & \eqdef \lambda p_{D\tau} . \lambda sel_{\tau \to D\boolT} . \\
  & \qquad\dokw {x_\tau \leftarrow p 
    }; \\
  & \qquad \dokw {b_{\boolT} \leftarrow sel\; 
    x 
    }; \\
  & \qquad\quad \score (\ifkw \; b \;\underline {1.0}\;\underline {0.0}); \\
  & \qquad\quad \retkw \; x.
\end{align*}
Let $\mu \in \Eval {D \tau}$, $s \in \Eval {\tau \to D \boolT}$.  We
again assume that $s (x)$ as
$s_{\trueop} (x) . \delta_{\trueop} + s_{\falseop} (x) .
\delta_{\falseop}$ for every $x \in \Eval \tau$.  Then, for every
$W \in \Open {\Eval \tau}$,
\begin{align*}
  \Eval {\mathtt{rej}'} \rho (\mu) (s) (W)
  & = \int_{x \in X} s_{\trueop} (x) \chi_W (x)
    d\mu,
\end{align*}
so that $\Eval {\mathtt{rej}'} \rho (\mu) (s)$ is simply
$s_{\trueop} \cdot \mu$.

In other words, $\Eval {\mathtt{rej}'} \rho (\mu) (s)$ computes the
valuation obtained from $\mu$ by applying the density function
$s_{\trueop}$.  When $s_{\trueop} = \chi_U$, this is $\mu_{|U}$.
Hence we retrieve the same result as for $\Eval {\mathtt{rej}} \rho$,
up to a renormalization factor, which may be infinite.

It is practical to abbreviate the form
$\score (\ifkw \; b \;\underline {1.0}\;\underline {0.0})$, where
$b \colon \boolT$, as $\mathtt{observe}\; b$.  This is a familiar
operation, implemented in Anglican \cite{WvdMM:Anglican}, as
a derived operator in Hakaru \cite{SR:disin}, or
as the third argument of $\mathtt{lex\text{-}query}$
clauses in Church \cite{GMRBT:Church}.

\subsection{Generating normal distributions}
\label{sec:box-muller-algorithm}

Let us imagine that we do not have a term $\sample [\mathcal N (0,
1)]$ for the normal (Gaussian) distribution $\mathcal N (0, 1)$ with
mean $0$ and standard deviation $1$.  We will give several different
possible implementations in ISPCF of well-known algorithms.

For every measure map $f \colon \real \to \creal$, $\mathcal N (0, 1)$
is such that:
\begin{align*}
  \int_{x \in \real} f (x) d\mathcal N (0, 1)
  & = \int_{x \in \real} f (x) \sqrt {\frac 1 {2\pi}} e^{-\frac 1 2
    x^2} dx.
\end{align*}
This is a measure defined from Lebesgue measure $\lambda$ on $\real$
by applying the density map
$x \mapsto \sqrt {\frac 1 {2\pi}} e^{-\frac 1 2 x^2}$.  As such, we
can implement it as follows, assuming a term $\sample [\lambda]$, as
well as primitives for multiplication, division, exponential $\expkw$,
square root $\sqrtkw$, and various numerical constants, all with their
intended semantics.  We agree that product is written with an infix
$\times$, and we remember that $M; N$ abbreviates $\dokw {x
  \leftarrow M}; N$, where $x$ is a fresh variable.
\begin{align*}
  \mathtt{normal}
  & \eqdef \dokw {x_{\realT} \leftarrow \sample [\lambda]}; \\
  & \qquad \score (\sqrtkw (\underline {0.5} / \underline\pi)
    \times \expkw (\underline{-0.5} \times x \times x)); \retkw x.
\end{align*}
This is the original purpose of $\score$: to define a distribution
from another one through a density function.  We check that this gives
us the intended result.  For every environment $\rho$, for every open
subset $W$ of $\IR_\bot$,
\begin{align*}
  \Eval {\mathtt{normal}} \rho (W)
  & = \int_{\mathbf x \in \IR_\bot}
    |\overline g (\mathbf x)| . \delta_{\mathbf x} (W)
    d i [\lambda]
\end{align*}
where
$\overline g (\mathbf x) \eqdef \Eval {\sqrtkw (\underline {0.5} /
  \underline\pi) \times \expkw (\underline{-0.5} \times x \times x)}
\rho [x \mapsto \mathbf x]$; we check that
$\overline g ([x,x]) = [g (x), g (x)]$ where
$g (x) \eqdef \sqrt {\frac 1 {2\pi}} e^{-\frac 1 2 x^2}$.  By the
change-of-variables formula, and the fact that $\delta_{[x,x]} (W) =
\chi_{i^{-1} (W)} (x)$,
\begin{align*}
  \Eval {\mathtt{normal}} \rho (W)
  & = \int_{x \in \real}
    g (x) . \chi_{i^{-1} (W)} (x)
    d \lambda
    = \mathcal N (0, 1) (i^{-1} (W)).
\end{align*}
Therefore $\Eval {\mathtt{normal}} \rho = i [\mathcal N (0, 1)]$.

Another way of implementing $\mathcal N (0, 1)$ is the
\emph{Box-Muller algorithm} \cite{boxmuller58}, which produces pairs
of two independent normally distributed variables, from two
independent uniformly distributed variables in $[0, 1]$.

In order to implement this in ISPCF, we assume a term
$\sample [\lambda_1]$, where $\lambda_1$ is the uniform measure on
$[0, 1]$.  We also assume we have primitives $\underline\sin$,
$\underline\cos$, $\underline\log$, and $\sqrtkw$, as well as
addition, subtraction, opposite, and multiplication, with the expected
semantics.  We consider the following ISPCF term of type
$D (\realT \times \realT)$.  We use the shorthand $\letbe {x=u} v$ for
$(\lambda x . v) u$.
\begin{align*}
  \mathtt{box\_muller}
  & \eqdef \dokw {x_{\realT} \leftarrow \sample [\lambda_1]}; \\
  & \qquad \dokw {y_{\realT} \leftarrow \sample [\lambda_1]}; \\
  & \qquad \letbe {c_{\realT} = \underline\cos (\underline{2.0} \times
    \underline\pi \times y)} {} \\
  & \qquad \letbe {s_{\realT} = \underline\sin (\underline{2.0} \times
    \underline\pi \times y)} {} \\
  & \qquad \letbe {m_{\realT} = \sqrtkw (\underline{-2.0} \times
    \underline\log\; x)} {} \\
  & \qquad \qquad \retkw (\langle m \times c, m \times s\rangle).
\end{align*}
This is a non-recursive definition, and as a consequence,
$\Eval {\mathtt{box\_muller}} \rho$ is a probability distribution on
$\IR_\bot \times \IR_\bot$ that is actually concentrated on the
subspace $\real \times \real$.  Explicitly, for every open subset $W$
of $\IR_\bot \times \IR_\bot$,
\begin{align*}
  \Eval {\mathtt{box\_muller}} \rho (W)
  & = \int_{\mathbf x \in \IR_\bot} \left(
    \int_{\mathbf y \in \IR_\bot}
    \delta_{f (\mathbf x, \mathbf y)} (W)
    d i [\lambda_1]
    \right)d i[\lambda_1],
\end{align*}
where $f (\mathbf x, \mathbf y)$ is
$\Eval {\letbe {c_{\realT} = \underline\cos (\underline{2.0} \times
    \underline\pi \times y)}
  { \cdots \inkw\;\langle m \times c, m \times s\rangle}}
\rho [x \mapsto \mathbf x, y \mapsto \mathbf y]$.  By the
change-of-variables formula, we have:
\begin{align*}
  \Eval {\mathtt{box\_muller}} \rho (W)
  & = \int_{x \in \real} \left(
    \int_{y \in \real}
    \delta_{f ([x,x], [y,y])} (W)
    d \lambda_1
    \right)d \lambda_1,
\end{align*}
and $f([x,x], [y,y]) = [g(x,y), g(x,y)]$, where
$g (x,y) = (\sqrt {-2 \log x} \cos (2\pi y), \allowbreak \sqrt {-2
  \log x} \sin (2\pi y))$.  It follows that $f ([x,x], [y,y])
\in W$ if and only if $g (x, y) \in i^{-1} (W)$, so that:
\begin{align*}
  \Eval {\mathtt{box\_muller}} \rho (W)
  & = \int_{x \in \real} \left(
    \int_{y \in \real}
    \chi_{i^{-1} (W)} (g (x, y))
    d \lambda_1
    \right)d \lambda_1,
\end{align*}
We now use the familiar argument subtending the Box-Muller algorithm,
which we recapitulate for completeness.  For every measurable map
$h \colon \real^2 \to \creal$ (the integrals are ordinary Lebesgue
integrals),
\begin{align*}
  & \int_{(s, t) \in \real^2} h (s, t) d \mathcal N (0, 1) \otimes
  \mathcal N (0, 1) \\
  & = \int_{(s, t) \in \real^2} \frac 1 {2\pi} h (s, t)
    e^{-\frac 1 2 (s^2+t^2)} ds\; dt \\
  & = \int_{r \in \realp, \theta \in [0, 2\pi]} \frac 1 {2\pi} h (r
    \cos \theta, r \sin \theta) e^{-\frac 1 2 r^2} r \;dr\;d\theta
  \\
  & \text{using the change of variables $s=r\cos\theta$,
    $t=r\sin\theta$} \\
  & = \int_{x, y \in [0, 1]} h (\sqrt {-2 \log x} \cos (2\pi y),
    \sqrt {-2 \log x} \sin (2 \pi y))
    \; dx\; dy
  \\
  & \text{by letting $x \eqdef e^{-\frac 1 2 r^2}$, $y \eqdef
    \theta/(2\pi)$}.
\end{align*}
By taking $h \eqdef \chi_{i^{-1} (W)}$, it follows that
$(\mathcal N (0, 1) \otimes \mathcal N (0, 1)) (i^{-1} (W)) = \Eval
{\mathtt{box\_muller}} \rho (W)$.  Therefore $\mathtt{box\_muller}$
implements the product $\mathcal N (0, 1) \otimes \mathcal N (0, 1)$,
in the sense that:
\begin{align*}
  \Eval {\mathtt{box\_muller}} \rho & = i [\mathcal N (0, 1) \otimes \mathcal N (0, 1)].
\end{align*}

\begin{remark}
  \label{rem:boundedsupp}
  The valuation $\lambda_1$ has compact support $[0, 1]$.  Since only
  continuous maps can be defined in ISPCF, it might seem intuitive
  that any continuous valuation defined by ISPCF terms that only use
  $\sample [\lambda_1]$ as sampler should also have compact support,
  enforcing a necessary limitation on the kind of distributions that
  one can define.  The example of the term $\mathtt{box\_muller}$ shows
  that this is not the case.  We will see another example in
  Section~\ref{sec:gener-lebesg-meas}.
\end{remark}

A more efficient variant of this algorithm combines it with rejection
sampling, and is due to Marsaglia and Bray \cite{marsaglia64}:
\begin{align*}
  \mathtt{box\_muller'}
  & \eqdef \dokw {\langle x_{\realT}, y_{\realT}\rangle \leftarrow
    \mathtt{rej}\;(\mathtt{prod}_{\realT,\realT}\;\mathtt{L}\;\mathtt{L})\;\mathtt{discp}};
  \\
  & \qquad \letbe {u_{\realT} = x\times x + y \times y} {} \\
  & \qquad \letbe{m_{\realT} = \sqrtkw (\underline{-2.0} \times
    \underline\log\;u / u)} {} \\
  & \qquad \quad \retkw \langle m \times x, m \times y \rangle,
\end{align*}
where $\mathtt{prod}_{\sigma,\tau}$ implements the product of two
minimal valuations on $\Eval \sigma$ and $\Eval \tau$, $\mathtt{L}$
implements the uniform measure on $[-1, 1]$ (namely, half of the
restriction of Lebesgue measure to $[-1, 1]$), and $\mathtt{discp}$
tests whether a point lies on the unit disc:
\begin{align*}
  \mathtt{prod}_{\sigma,\tau}
  & \eqdef \lambda p_{D \sigma} . \lambda q_{D \tau} . \\
  & \qquad \dokw {x_\sigma \leftarrow p};  \dokw {y_\tau \leftarrow
    q}; \retkw \langle x, y\rangle \\
  \mathtt{L}
  & \eqdef \dokw {x_{\realT} \leftarrow \sample [\lambda_1]};
    \retkw (\underline{2.0} \times x - \underline{1.0}) \\
  \mathtt{discp}
  & \eqdef \lambda \langle x_{\realT}, y_{\realT}\rangle .
    \retkw (\poskw (\underline {1.0} - x \times
    x - y \times y).
\end{align*}
We also use the convenient abbreviation
$\lambda \langle x_\sigma, y_\tau\rangle. M$ for
$\lambda p_{\sigma \times \tau} . \letbe {x_\sigma = \pi_1 p} {}
\allowbreak \letbe {y_\tau = \pi_2 p} {M}$, where $p$ is not free in
$M$, and similarly for
$\dokw {\langle x_\sigma, y_\tau\rangle \leftarrow M}; N$.

\begin{remark}
  \label{rem:prod}
  We could have defined $\mathtt{prod}_{\sigma,\tau}$ as
  $\lambda p_{D \sigma} . \lambda q_{D \tau} .  \dokw {y_\tau
    \leftarrow q}; \dokw {x_\sigma \leftarrow p}; \retkw \langle x,
  y\rangle$ instead.  Since the monad of minimal valuations is
  commutative, this would not make any semantical difference.
\end{remark}


We let the reader check that $\Eval {\mathtt{box\_muller'}} \rho$ is
also equal to $i [\mathcal N (0, 1), \mathcal N (0, 1)]$.  We only
give a sketch of an argument.  The semantics of
$\mathtt{rej}\;(\mathtt{prod}_{\realT,\realT}\;\mathtt{L}\;\mathtt{L})\;\mathtt{discp}$
is $\frac 1 {1-i [\mu] (\mathcal V)} i [\mu]_{|\mathcal U}$, where:
\begin{itemize}
\item $\mu$ is the uniform measure on the square
  $[-1, 1] \times [-1, 1]$;
\item $\mathcal U$ is the Scott-open set of pairs
  $([a,b], [c,d]) \in \IR^2$ such that
  $[x,y] \eqdef [a,b] \times [a,b] + [c, d] \times [c, d]$ satisfies
  $y < 1$; in particular, $U \eqdef i^{-1} (\mathcal U)$ is the open
  unit disc $\{(x, y) \in \real^2 \mid x^2+y^2 < 1\}$;
\item $\mathcal V$ is a similar Scott-open set, with the property that
  $V \eqdef i^{-1} (\mathcal V)$ is the complement of the closed unit
  disc, namely $V = \{(x, y) \in \real^2 \mid x^2+y^2 > 1\}$.
\end{itemize}
Then $i [\mu] (\mathcal V) = \mu (V) = 1 - \pi/4 = 1 - \mu (U)$, so
that the semantics of
$\mathtt{rej}\;(\mathtt{prod}_{\realT,\realT}\;\mathtt{L}\;\mathtt{L})\;\mathtt{discp}$
is $i [\mu']$, where $\mu'$ is the uniform measure on the (open) unit
disc.  The rest of the argument follows similar lines as with the
original Box-Muller algorithm.

Here is a fourth and final implementation of $\mathcal N (0, 1)$.  Let
$\mathtt{box\_muller'}'$ be obtained from $\mathtt{box\_muller'}$ by
replacing $\mathtt{rej}$ by $\mathtt{rej}'$.  The semantics of
$\mathtt{rej}'\;(\mathtt{prod}_{\realT,\realT}\;\mathtt{L}\;\mathtt{L})\;\mathtt{discp}$
is equal to $1-i [\mu] (\mathcal V) = \pi/4$ times the semantics of
$\mathtt{rej}\;(\mathtt{prod}_{\realT,\realT}\;\mathtt{L}\;\mathtt{L})\;\mathtt{discp}$.
It follows that
$\Eval {\mathtt{box\_muller'}'} \rho = \frac \pi 4 \; i [\mathcal N (0,
1), \mathcal N (0, 1)] = i \left[\frac \pi 4 \mathcal N (0, 1), \frac
  \pi 4 \mathcal N (0, 1)\right]$.  This yields the same result as
before, up to a renormalization factor.

\subsection{Generating Lebesgue measure on $\real$}
\label{sec:gener-lebesg-meas}

We have noticed that one can define measures with unbounded support
from just $\sample [\lambda_1]$, although $\lambda_1$ has bounded
support (see Remark~\ref{rem:boundedsupp}).  We now show that we can
define measures with unbounded support and whose total mass is
infinite as well.  This requires the use of $\score$.  What may be
intriguing is that $\score$ multiplies the current `probability' by a
real number different from $+\infty$.  Still, this can be used to
increase the total mass from $1$ to $+\infty$.

Hence, we imagine that the term $\sample [\lambda_1]$ is available,
where $\lambda_1$ is Lebesgue measure on $[0, 1]$, but that
$\sample [\lambda]$ is not, where $\lambda$ is Lebesgue measure on
$\real$.  Assuming constants $\underline\pi$, $\underline{1.0}$,
$\underline{2.0}$, $\underline\tan$, $\underline\arctan$, and some
constants for addition, subtraction, multiplication and division, we
define:
\begin{align*}
  \mathtt{uniform}
  & \eqdef \lambda a_{\realT} . \lambda b_{\realT} .
    \dokw {x_{\realT} \leftarrow \sample [\lambda_1]}; \retkw
    {(a+(b-a)\times x)} \\
  \mathtt{lebesgue}
  & \eqdef \dokw {z_{\realT} \leftarrow \mathtt{uniform}
    (-\underline\pi/\underline{2.0}) (\underline\pi/\underline{2.0})};
  \\
  & \qquad \score (\underline{1.0} / (\underline {1.0} + z \times z));
  \\
  & \qquad \retkw (\underline\arctan\; z).
\end{align*}
We claim that $\mathtt{lebesgue}$ fits the bill.  Let
$\upsilon_{[a, b]}$ denote the uniform probability measure on the
interval $[a, b]$; up to the embedding $i \colon \real \to \IR_\bot$,
this is what $\mathtt{uniform}$ applied to $a$ and $b$ computes.  Then
we use a monotonic, differentiable homeomorphism $g$ from $\real$ to
the open interval $]a, b[$ in order to transport that to the whole of
$\real$.  (One may observe that $]a, b[$ differs from $[a, b]$, but
that will not matter, since the complement of $]a, b[$ in $[a, b]$ has
$\upsilon_{[a, b]}$-measure $0$.)  We have chosen
$g (u) \eqdef \arctan u$, $a = -\pi/2$, $b = \pi/2$ here.

Let us check our claim formally; this is a paper on semantics, and we
need to check that our semantics produces the intended result.  First,
for every environment $\rho$, for all $a, b \in \real$ such that
$a < b$, $\Eval {\mathtt{uniform}} \rho (a) (b)$ is the minimal
valuation $(\lambda x . \delta_{a+(b-a)x)})^\dagger (i [\lambda_1])$.
For every open subset $U$ of $\IR_\bot$, therefore,
\begin{align*}
  \Eval {\mathtt{uniform}} \rho (a) (b) (U)
  & = \int_{\mathbf u \in \IR_\bot} \delta_{i(a) + (i (b) - i (a)) \times
    \mathbf u)} (U) d i [\lambda_1] \\
  & = \int_{u \in \real} \delta_{(i (a) + (i (b) - i (a)) \times i
    (u))} (U) d
    \lambda_1 \\
  & \text{by the change-of-variables formula} \\
  & = \int_{u \in \real} \delta_{i (a+(b-a)u)} (U) d \lambda_1 \\
  & = \int_0^1 \chi_{i^{-1} (U)} (a+(b-a)u) du \\
  & \text{using the identity }\delta_c (U) = \chi_U (c) \\
  & = \frac 1 {b-a} \int_a^b \chi_{i^{-1} (U)} (v) dv \\
  & \text{by letting }v \eqdef a+(b-a)u \\
  & = \upsilon_{[a,b]} (i^{-1} (U)).
\end{align*}
Therefore $\Eval {\mathtt{uniform}} \rho (a) (b) = i [\upsilon_{[a,b]}]$.

Next, we evaluate $\Eval {\mathtt{lebesgue}} \rho (U)$ for every open
subset $U$ of $\IR_\bot$ as follows.  Let us abbreviate the term
$\score (\underline{1.0} / (\underline {1.0} + z \times z)); \retkw
(\underline\arctan\; z)$ as $M (z)$.  For every $u \in \real$, for
every open subset $U$ of $\IR_\bot$, letting
$\nu_u \eqdef \Eval {M(z)} \rho [z \mapsto i (u)]$, we have the
following, where $g (u) \eqdef \arctan u$, and therefore $g' (u) = 1 / (1+u^2)$:
\begin{align*}
  \nu_u (U)
  & = (\lambda \_ . \delta_{g (u)})^\dagger (g' (u) \delta_*) (U) \\
  & = \int_{\_ \in \Eval {\unitT}} \delta_{g (u)} (U) d g' (u) \delta_* \\
  & = g' (u) \chi_{i^{-1} (U)} (g (u)).
\end{align*}
Therefore,  
\begin{align*}
  \Eval {\mathtt{lebesgue}} \rho (U)
  & = \int_{\mathbf u \in \IR_\bot} \Eval {M (z)} \rho [z \mapsto
    \mathbf u] (U) di [\upsilon_{[-\pi/2, \pi/2]}] \\
  & = \int_{u \in \real} \Eval {M (z)} \rho [z \mapsto i (u)] (U) d
    \upsilon_{[-\pi/2, \pi/2]} \\
  & \text{by the change-of-variables formula} \\
  & = \int_{u \in \real} \nu_u (U) d \upsilon_{[-\pi/2, \pi/2]} \\
  & = \int_{u \in \real} g' (u) \chi_{i^{-1} (U)} (g (u))
    d \upsilon_{[-\pi/2, \pi/2]} \\
  & = \int_{-\pi/2}^{\pi/2} g' (u) \chi_{i^{-1} (U)} (g (
    u)) du \\
  & = \int_{-\infty}^{+\infty} \chi_{i^{-1} (U)} (t) dt \\
  & \text{by letting }t \eqdef g (u) \\
  & = \lambda (i^{-1} (U)).
\end{align*}
Therefore, as promised, $\Eval {\mathtt{lebesgue}} \rho$ is the image
measure $i [\lambda]$ of Lebesgue measure $\lambda$ on $\real$.

\subsection{Generating exponential distributions}
\label{sec:gener-expon-distr}

Generating an exponential distribution is another classical example.
The novelty is that this is defined from Lebesgue measure $\lambda$ by
a density function that is \emph{not} continuous.  Although ISPCF can
only express continuous functions, we show that this will not prevent
us from defining the exponential distribution.

For simplicity, we consider the exponential distribution with
parameter $1$.  The density map is:
\begin{align*}
  g (x) & \eqdef \left\{
          \begin{array}{ll}
            e^{-x} & \text{if } x \geq 0 \\
            0 & \text{otherwise}
          \end{array}
          \right.
\end{align*}
We note that $g$ is not continuous at $0$.  However, we have the
following candidate ISPCF term, assuming the relevant primitives:
\begin{align*}
  \mathtt{exp\_density}
  & \eqdef \lambda x_{\realT} . \ifkw\; {(\poskw \;
    x)}\;(\underline{\exp} (-x)) \; \underline{0.0}.
\end{align*}
For every environment $\rho$, we see that
$\overline g \eqdef \Eval {\mathtt{exp\_density}}$ maps $\bot$ to
$\bot$, and every interval $[a, b]$ to $[e^{-b}, e^{-a}]$ if $a > 0$,
to $[0,0]$ if $b < 0$, and to $\bot$ otherwise.  In particular,
$\Eval {\mathtt{exp\_density}} (i (a))$ coincides with $i (g (a))$ for
all points of continuity of $g$, namely for every $a \neq 0$.

The most obvious way to implement the exponential distribution is by
using $\score$.  We assume a term $\sample [\lambda]$, where $\lambda$
is Lebesgue measure on $\real$.
\begin{align*}
  \mathtt{expo}
  & \eqdef \dokw {x_{\realT} \leftarrow \sample [\lambda]};
    \score (\mathtt{exp\_density}\;x); \retkw x.
\end{align*}
For every open subset $W$ of $\IR_\bot$, we compute:
\begin{align*}
  \Eval {\mathtt{expo}} \rho (W)
  & = \int_{\mathbf x \in \IR_\bot}
    |\overline g (\mathbf x)| . \delta_{\mathbf x} (W)
    d i [\lambda] \\
  & = \int_{x \in \real}
    |\overline g (i (x))| . \delta_{i (x)} (W) d \lambda
  & \text{by the change-of-variables formula} \\
  & = \int_{x \in \real} g (x) . \chi_{i^{-1} (W)} (x) d\lambda.
\end{align*}
The last line is justified by the fact that
$\delta_{i (x)} (W) = \chi_{i^{-1} (W)} (x)$, and more importantly, by
the fact that $|\overline g (i (x))|$ and $g (x)$ coincide for every
$x \in \real$ except on a set of Lebesgue measure $0$.  (Namely, when
$x=0$, in which case $g (x) = 1$ and $|\overline g (i (x))| = 0$.)

It follows that $\Eval {\mathtt{expo}} \rho$ is exactly
$i [g \cdot \lambda]$, the image by $i$ of the exponential
distribution $g \cdot \lambda$.

Another way of implementing the exponential distribution is as:
\begin{align*}
  \mathtt{expo}'
  & \eqdef \dokw {x_{\realT} \leftarrow \sample [\lambda_1]};
    \retkw (-\underline{\log}\; x),
\end{align*}
assuming a term $\sample [\lambda_1]$ for sampling along the Lebesgue
measure $\lambda_1$ on $[0, 1]$.  We check that this also gives the
correct result.  For every open subset $W$ of $\IR_\bot$,
\begin{align*}
  \Eval {\mathtt{expo}'} \rho (W)
  & = \int_{\mathbf x \in \IR_\bot} \delta_{-\Eval {\underline{\log}} \rho
    (\mathbf x)} (W) di [\lambda_1] \\
  & = \int_0^1 \chi_W (-\log x) dx,
\end{align*}
by the change-of-variables formula, the fact that integrating along
$\lambda_1$ means integrating on $[0, 1]$, and the equality
$\delta_a (W) = \chi_W (a)$.  We now let $x \eqdef e^{-t}$.
The latter integral is then equal to $\int_0^{+\infty} \chi_W (t)
e^{-t} dt = (g \cdot \lambda) (W)$.  Therefore $\Eval {\mathtt{expo}'}
\rho = \Eval {\mathtt{expo}} \rho = i [g \cdot \lambda]$.

A final implementation of the exponential distribution, due to von
Neumann, consists in simulating a distribution with density
$(1-x) + (\frac {x^2} {2!} - \frac {x^3} {3!}) + \cdots + (\frac
{x^{2n}} {(2n)!} - \frac {x^{2n+1}} {(2n+1)!}) + \cdots = e^{-x}$.  We
will leave the verification that the term $\mathtt{von\_neumann}$
below again has the same semantics $i [g \cdot \lambda]$ as
$\mathtt{expo}$ and $\mathtt{expo}'$, as an exercise.  We first define:
\begin{align*}
  \mathtt{longest\_decreasing\_run}
  & \eqdef \reckw (\lambda f_{\realT \to \intT \to D \intT} . \lambda x_{\realT} . \lambda n_{\intT} . \\
  & \qquad\qquad \dokw {u \leftarrow \sample [\lambda_1]}; \\
  & \qquad\qquad \ifkw \; {(\poskw (u-x))}\; \\
  & \qquad\qquad\qquad \retkw n \\
  & \qquad\qquad\qquad (f\;u\; (n+\underline 1)))
\end{align*}
Given any $x \in [0, 1]$,
$\Eval {\mathtt{longest\_decreasing\_run}} \rho (x) (0)$ returns the
largest $n$ such that $x > u_1 > u_2 > \cdots > u_n$ ($\leq u_{n+1}$)
for randomly uniformly distributed real numbers $u_1$, $u_2$, \ldots,
in $[0, 1]$.  It is a standard exercice to show that
$\Eval {\mathtt{longest\_decreasing\_run}} \rho (x) (0)$ is the
distribution on $\Z_\bot$ giving probability
$\frac {x^n} {n!} - \frac {x^{n+1}} {(n+1)!}$ to each $n \in \nat$,
and probability $0$ to all other elements.  Drawing a value $n$ at
random with respect to that distribution, the probability that $n$ is
even is
$\sum_{n \text{ even}} (\frac {x^n} {n!} - \frac {x^{n+1}} {(n+1)!}) =
e^{-x}$.  Assuming a term
$\underline{\text{odd}} \colon \intT \to \boolT$ testing whether its
argument is odd, one then defines:
\begin{align*}
  \mathtt{von\_neumann}
  & \eqdef \dokw {x_{\realT} \leftarrow \sample [\lambda_1]}; \\
  & \qquad \reckw (\lambda f_{\realT \to D \realT} . \lambda
    \ell_{\realT} . \\
  & \qquad\qquad \ifkw\; (\underline{\text{odd}}
    (\mathtt{longest\_decreasing\_run}\;x\;0)) \\
  & \qquad\qquad\qquad f\;(\ell+\underline{1.0}) \\
  & \qquad\qquad\qquad \retkw \ell
    ) x
\end{align*}
The $\reckw$ expression in the middle maps every real number $\ell$
to $\ell+m$, where $m$ is the average number of calls to
$\mathtt{longest\_decreasing\_run}\;x\;0$ it takes it to return an even
number.

We direct the reader to \cite{FDWM:expo} for an explanation and some
faster algorithms based on the same principle.  See also
\cite{Forsythe:expo} for a generalization to densities of the form
$e^{-G (x)}$, including normal distributions.

\subsection{Distributions on higher-order objects}
\label{sec:distr-high-order}

All our previous examples were about building distributions on simple
types such as $\realT$ or $\realT \times \realT$.  One may be
interested in generating distributions on functions, or even on types
of distributions themselves.

Let us start with the following problem.  We are given a countably
infinite family of distributions on some type $\tau$, as the value of
a parameter $F \colon \intT \to D \tau$, and we wish to find a random
function $f$ of type $\intT \to \tau$, such that $f (0)$, $f (1)$,
\ldots, are independently distributed according to the distributions
$F (0)$, $F (1)$, \ldots, respectively.  Formally, we look for a
distribution of type $D (\intT \to \tau)$ over those functions.

At first sight, this looks like an infinite-dimensional generalization
of the product distribution term $\mathtt{prod}_{\sigma,\tau}$
introduced earlier.  For example, we may think of writing:
\begin{align*}
  \mathtt{wrong\_infinite\_prod}_\tau
  & \eqdef \reckw (\lambda P_{(\intT \to D\tau) \to D (\intT \to \tau)}
    .
    \lambda F_{\intT \to D \tau} . \\
  &\qquad \dokw {x_\tau \leftarrow \headkw\;F}; \\
  &\qquad \dokw {rest_{\intT \to \tau} \leftarrow P (\tailkw\; F)}; \\
  &\qquad\qquad \retkw (x :: rest)),
\end{align*}
where $\headkw F \eqdef F (\underline 0)$ extracts the first element
of the sequence encoded by $F$,
$\tailkw F \eqdef \lambda n_{\intT} . F (n + \underline 1)$ extracts
the remaining elements, and
$x :: rest \eqdef \lambda n_{\intT} . \ifkw {(n=\underline 0)} \; x \;
(rest\;(n-\underline 1))$ adds $x$ to the front of $rest$.  However,
and as the lack of base case for the recursion may hint of, this is
wrong: the semantics of $\mathtt{wrong\_infinite\_prod}_\tau$ is
merely the zero valuation; that program never terminates.

In fact, our problem has no solution.  The reason is that the total
mass of what we want to compute is undefined (and this undefinedness
does not just apply to our domain-theoretic semantics, but to all the
semantics we know of).  We are given infinitely many distributions
$\mu_0$, $\mu_1$, \ldots, $\mu_n$, \ldots{} on the same space
$X \eqdef \Eval \tau$, as input.  Their infinite product,
if it exists, has total mass $\prod_{n=0}^{+\infty} a_n$, where
$a_n \eqdef \mu_n (X)$ is the total mass of $\mu_n$.  That only makes
sense if the infinite product $\prod_{n=0}^{+\infty} a_n$ converges.
For example, it makes non sense if $a_n=1/2$ for $n$ odd and $a_n=2$
for $n$ even.  Such cases are easy to build using $\score$.

Let us consider a slightly easier problem, where all the distributions
$\mu_n$ are constrained to be the same distribution $\mu$; we wish to
compute the infinite product $\mu^\infty$ of $\mu$ with itself.  Now
this simplified problem still does not have any solution (unless
$\Eval \tau$ is empty), and this is now particular to our
domain-theoretic semantics.  Indeed, let us imagine that the map
$f \colon \mu \mapsto \mu^\infty$ were computable, hence
Scott-continuous, and let $x$ be any fixed element of $\Eval \tau$.
Then the map $a \in [0, 1] \mapsto (a \delta_x)^\infty (\Eval \tau)$
would be Scott-continuous.  But that function maps $1$ to $1$, and all
other elements of $[0, 1]$ to $0$, and is therefore not
Scott-continuous.

Instead, we will show how one can answer our problem in special cases.
The first special case we consider is to build the product
$\lambda_1^\infty$ of a countably infinite number of copies of
Lebesgue measure $\lambda_1$ on $[0, 1]$.  By definition,
$\lambda_1^\infty$ is the unique continuous valuation on the countable
topological product $[0, 1]^\nat$ whose image valuation
$\pi_S [\lambda_1^\infty]$ onto $[0, 1]^S$ is the (finite) product
valuation $\prod_{n \in S} \lambda_1$, for every finite subset $S$ of
$\nat$.  We write $\pi_S \colon [0, 1]^\nat \to [0, 1]^S$ for
projection onto the coordinates in $S$.  This exists and is unique by
general theorems \cite[Theorem~5.3]{JGL:kolmogorov}.

We embed $[0, 1]^\nat$, and more generally $\real^\nat$, into
$\Eval {\intT \to \realT} = [\Z_\bot \to \IR_\bot]$ as follows.  There
is a map $j \colon \real^\nat \to \Eval {\intT \to \realT}$ defined by
$j (\vec x) (n) \eqdef i (x_n)$ for every
$\vec x \eqdef {(x_n)}_{n \in \nat}$ in $[0, 1]^\nat$ and every
$n \in \nat$, where $i$ is the usual embedding of $\real$ into
$\IR_\bot$, restricted to $[0, 1]$; and $j (\vec x) (n) \eqdef \bot$
for every $n \in \{-1, -2, \cdots\} \cup \{\bot\}$.
\begin{lemma}
  \label{lemma:j}
  The Scott topology of $\Eval {\intT \to \realT}$ has a subbase of
  open sets of the form
  $[n \in V] \eqdef \{f \in [\Z_\bot \to \IR_\bot] \mid f (n) \in
  V\}$, where $n$ ranges over $\Z_\bot$ and $V$ ranges over the open
  subsets of $\IR_\bot$.

  The map $j$ is a topological embedding of $\real^\nat$ into
  $\Eval {\intT \to \realT}$.
\end{lemma}
\begin{proof}
  $\Z_\bot$ and $\IR_\bot$ are bc-domains, and therefore the Scott
  topology on $[\Z_\bot \to \IR_\bot]$ coincides with the topology of
  pointwise convergence \cite[Proposition~11.2]{jgl-jlap14}.  The
  latter is the topology generated by the indicated subbasic open sets
  $[n \in V]$, by definition.  This can also be deduced from
  Proposition~II-4.6 of \cite{gierz03}.

  The inverse image $j^{-1} ([n \in V])$ is equal to the subbasic open
  set $\pi_n^{-1} (U) \eqdef \{\vec x \in \real^\nat \mid x_n \in U\}$
  (where $U \eqdef i^{-1} (V)$) if $n \in \nat$, is empty if
  $n \not\in \nat$ and $\bot \not\in V$, and is the whole of
  $\real^\nat$ otherwise.  Therefore $j$ is continuous.

  Finally, every subbasic open set $\pi_n^{-1} (U)$ ($n \in \nat$, $U$
  open in $\real$) is equal to $j^{-1} ([n \in V])$, where $V$ is any
  open subset of $\IR_\bot$ such that $i^{-1} (V) = U$ (which exists
  since $i$ is an embedding).  Hence $j$ is a topological embedding.
\end{proof}
We will build a term of type $D (\intT \to \realT)$ whose semantics is
$j [\lambda_1^\infty]$---a continuous valuation that now makes sense;
this will be stated in Theorem~\ref{thm:lambda1:inf}.

The idea is simple.  We draw a real number at random using
$\lambda_1$, we write it in binary, cut (or dice) its sequence of bits
into a countably infinite partition of subsequences of bits, and we
reassemble (or splice) each subsequence into a new real number in
$[0, 1]$.  The dicing operation is far from being continuous, but this
will not matter; the reason why will be given in
Remark~\ref{rem:bin:unfold}.

Let us introduce the mathematical objects we will need.  Let
$\upsilon$ be the uniform measure on infinite sequences of bits; this
is the unique Borel measure on $\{0, 1\}^\nat$ (where $\{0, 1\}$ has
the discrete topology) such that
$\upsilon (\{s \in \{0, 1\}^\nat \mid s_0=b_0, \cdots,
s_{k-1}=b_{k-1}\}) = 1/2^k$ for all $k \in \nat$ and elements
$b_0, \ldots, b_{k-1} \in \{0, 1\}$.  The easiest way to show its
existence is as the image measure of Lebesgue measure by a suitable
map.
\begin{lemma}
  \label{lemma:upsilon}
  The function $\bin \colon \real \to \{0, 1\}^\nat$ that maps
  every number in $[0, 1[$ to the string of bits in its binary
  expansion, every negative number to the all zero string $0^\omega$,
  and every number larger than or equal to $1$ to the all one string
  $1^\omega$ is measurable.  The image measure $\bin [\lambda_1]$ is
  the unique measure $\upsilon$ on $[0, 1]$ such that
  $\upsilon (\{s \in \{0, 1\}^\nat \mid s_0=b_0, \cdots,
  s_{k-1}=b_{k-1}\}) = 1/2^k$ for all $k \in \nat$ and elements
  $b_0, \ldots, b_{k-1} \in \{0, 1\}$.
\end{lemma}
\begin{proof}
  For every finite string $b \eqdef (b_0, b_1, \cdots, b_{k-1})$ of
  bits, let $U_b$ be the set of infinite strings of bits having $s$ as
  a prefix.  The sets $U_b$ form a countable base of the topology on
  $\{0, 1\}^\nat$, and $\bin^{-1} (U_b)$ is the interval
  $[\sum_{i=0}^{k-1} b_i/2^{i+1}, \sum_{i=0}^{k-1} b_i/2^{i+1} +
  1/2^k[$ if not all bits $b_i$ are equal to $0$ or to $1$,
  $]-\infty, 1/2^k[$ if all the bits $b_i$ are equal to $0$, and to
  $[1-1/2^k, +\infty[$ otherwise.  Every interval is in the Borel
  $\sigma$-algebra, and every open subset of $\{0, 1\}^\nat$ is a
  countable union of sets $U_s$, so the inverse image of any open
  subset of $\{0, 1\}^\nat$ by $\bin$ is Borel.  It follows that
  $\bin$ is measurable.  Then
  $\bin [\lambda_1] (U_b) = \lambda_1 (\bin^{-1} (U_b)) = \lambda
  (\bin^{-1} (U_s) \cap [0, 1])$ is equal to $1/2^k$.  Finally, the
  sets $U_b$ form a $\pi$-system, whence the claim of uniqueness.
\end{proof}

The function $\bin$ is not continuous on the whole of $\real$, hence
not implementable.  Its restriction to the set of non-dyadic numbers
in $[0, 1]$, is, though, and we make this clearer by defining a
continuous map $\bin' \colon \IR_\bot \to \{0, 1\}_\bot^\nat$ that
implements $\bin$ on the non-dyadic numbers in $[0, 1]$.  Roughly,
$\bin'$ is implemented as follows.  Given an argument value of type
$\realT$, if the value is smaller than $1/2$ (which one can test by
using $\posop$), then emit a zero bit and multiply the value by $2$;
if that value is larger than $1/2$, then emit a one bit, subtract
$1/2$ from the value and multiply the result by $2$; then collect all
the emitted bits into a sequence, which may be infinite, or finite (if
computation ever gets stuck).  Formally, we define
$\bin' (\underline a)_k$ by induction on $k$ as follows.  If
$\underline a = \bot$, then $\bin' (\underline a)_k$ is $\bot$.  If
$\underline a = [a, b]$ and $k=0$, then $\bin' ([a, b])_0$ is defined
as $0$ if $b < 1/2$, $1$ if $a > 1/2$, and as $\bot$ otherwise.  If
$\underline a = [a, b]$ and $k \geq 1$, then $\bin' ([a, b])_k$ is
defined as $\bin' (\underline c)_{k-1}$ where $\underline c$ is equal
to $[2a, 2b]$ if $b < 1/2$, $[2a-1, 2b-1]$ if $a > 1/2$, and $\bot$
otherwise.  One checks easily that $\bin'$ is continuous.

\begin{remark}
  \label{rem:bin:unfold}
  Let $E_0$ be the collection of all points that are either dyadic or
  outside $[0, 1]$.  This is a set of $\lambda_1$-measure zero,
  because there are only countably many dyadic numbers and $\lambda_1$
  is supported on $[0, 1]$.  Since $\bin$ and $\bin'$ coincide outside
  $E_0$, the image measures of $\lambda_1$ by each of those two
  functions is the same.  (Formally, $j' [\bin [\lambda_1]]$ and
  $\bin' [i [\lambda_1]]$ coincide, for any environment $\rho$, and
  where $j'$ is the obvious topological embedding
  $\{0, 1\}^\nat \to \{0, 1\}_\bot^\nat$.)
\end{remark}

\begin{lemma}
  \label{lemma:num}
  The map $\num \colon \{0, 1\}^\nat \to \real$ that sends $s$ to
  $\sum_{i \in \nat} s_i/2^{i+1}$ is continuous, and the image measure
  $\num [\upsilon]$ is equal to $\lambda_1$.
\end{lemma}
\begin{proof}
  Let $U$ be any open subset of $\real$, and $s \in \num^{-1} (U)$.
  For some $\epsilon > 0$, the interval
  $]\num (s)-\epsilon, \num (s)+\infty[$ is included in $U$.  We pick
  $k \in \nat$ so that $1/2^k < \epsilon$.  Then the set of elements
  $t \in \{0,1\}^\nat$ such that $t_0=s_0$, \ldots, $t_k=s_k$ is an
  open neighborhood of $s$ that is included in $\num^{-1} (U)$,
  showing that $\num^{-1} (U)$ is open.  Therefore $\num$ is
  continuous.

  It is easy to see that $\num (\bin (x)) = x$ for every
  $x \in [0, 1]$, so that $(\num \circ \bin)^{-1} (E)$ and $E$ have
  the same intersection with $[0, 1]$, for every subset $E$ of
  $\real$.  Hence, for every Borel measurable subset $E$ of $\real$,
  and using Lemma~\ref{lemma:upsilon},
  $\num [\upsilon] (E) = \num [\bin [\lambda_1]] (E) = \lambda ((\num
  \circ \bin)^{-1} (E) \cap [0, 1]) = \lambda (E \cap [0, 1]) =
  \lambda_1 (E)$.
\end{proof}

We assume an bijective computable pairing map
$\langle \_, \_\rangle \colon \nat \times \nat \to \nat$.  A practical
one is the map that interleaves the two binary representations of the
numbers in argument, namely
$\langle \sum_{i \in \nat} a_i 2^i, \sum_{j \in \nat} b_j 2^j \rangle
\eqdef \sum_{i \in \nat} a_i 2^{2i} + \sum_{j \in \nat} b_j 2^{2j+1}$.

Given a random string of bits $s$, the strings
$s [m] \eqdef {(s_{\langle m, n\rangle})}_{n \in \nat}$, $m \in \nat$,
are themselves random and independent.  We express this by the
following well-known lemma.  In order to state it, for
every $s \in \{0, 1\}^\nat$, let $\upc s$ (``shift $s$'') be the
string such that
$(\upc s)_{\langle m, n \rangle} = s_{\langle m+1,n\rangle}$ for all
$m, n \in \nat$.  In other words, $(\upc s) [m] = s [m+1]$ for every
$m \in \nat$.

Computing $s [m]$ from $s$ and $m$ is easy, considering that pairing
$\langle \_, \_, \rangle$ is computable.

The spaces $\{0, 1\}^\nat$ and $(\{0,1\}^\nat)^2$ are compact
Hausdorff, and therefore locally compact and sober, in particular
LCS-complete.  We recall that, in that case, every continuous
valuation on those spaces extends to a Borel measure
\cite[Theorem~1.1]{DGJL-isdt19}.  The extension is unique for bounded
continuous valuations such as $\upsilon$ or
$\upsilon \otimes \upsilon$, because the open sets form a $\pi$-system
that generates the Borel $\sigma$-algebra.  The spaces $\{0, 1\}^\nat$
and $(\{0,1\}^\nat)^2$ are also second-countable, and then by
Adamski's theorem \cite[Theorem~3.1]{Adamski:measures} every measure
restricts to a continuous valuation on the open sets.  In other words,
bounded measures and bounded continuous valuations are in bijective
correspondence on those spaces.  Hence we can trade freely between the
notions of (bounded) continuous valuations and (bounded) measures on
those spaces.
\begin{lemma}
  \label{lemma:split}
  The map $split \colon s \mapsto (s [0], \upc s)$ is a homeomorphism
  of $\{0,1\}^\nat$ onto $(\{0,1\}^\nat)^2$.  The image continuous
  valuation (resp., measure) $split [\upsilon]$ is the product
  valuation (resp., product measure) $\upsilon \otimes \upsilon$.
\end{lemma}
\begin{proof}
  The first part is a direct consequence of the fact that the pairing
  map is bijective.  For the second part, we reason on continuous
  valuations, since that will be slightly easier.  For every pair of
  disjoint finite subsets $I$ and $J$ of $\nat$, let $U_{I,J}$ be the
  set of elements $s$ of $\{0,1\}^\nat$ such that $s_m=0$ for every
  $m \in I$ and $s_n=1$ for every $n \in J$.  We see that those sets
  form a base of the topology on $\{0, 1\}^\nat$, and that
  $\upsilon (U_{I,J}) = 1/2^{|I|+|J|}$, where $|\_|$ denotes
  cardinality, using the last part of Lemma~\ref{lemma:upsilon}, A
  base of the product topology on $(\{0, 1\}^\nat)^2$ is given by the
  sets $U_{I,J} \times U_{I',J'}$, and:
  \begin{align*}
    split [\upsilon] (U_{I,J} \times U_{I',J'})
    & = \upsilon (split^{-1} (U_{I,J} \times U_{I',J'})) \\
    & = \upsilon (U_{I'', J''}),
  \end{align*}
  where
  $I'' \eqdef \{\langle 0, k\rangle \mid k \in I\} \cup \{\langle
  m'+1,k\rangle \mid m', k \in \nat, \langle m', k \rangle \in I'\}$ and
  $J'' \eqdef \{\langle 0, k\rangle \mid k \in J\} \cup \{\langle
  n'+1,k\rangle \mid n', k \in \nat, \langle n', k \rangle \in J'\}$.
  Hence
  $split [\upsilon] (U_{I,J} \times U_{I',J'}) = 1/2^{|I''|+|J''|} =
  1/2^{|I|+|I'| + |J|+|J'|}$.  We also have:
  \begin{align*}
    (\upsilon \otimes \upsilon) (U_{I,J} \times U_{I',J'})
    & = \upsilon (U_{I,J}) . \upsilon (U_{I',J'})
      = 1/2^{|I|+|J|} . 1/2^{|I'|+|J'|},
  \end{align*}
  so $split [\upsilon]$ and $\upsilon \otimes \upsilon$ coincide on
  the basic open sets $U_{I,J}$.  Using modularity, they must coincide
  on any finite disjoint union of basic open sets, and by
  Scott-continuity, they must coincide on any arbitrary disjoint union
  of basic open sets.  (This is an argument we have sketched in the
  proof of Theorem~\ref{prop:fubini:top} already.)  Since any two
  basic open sets $U_{I,J}$ are either comparable or disjoint, every
  open subset of $\{0, 1\}^\nat$ is such a disjoint union, and this
  concludes the proof.
\end{proof}

\begin{corollary}
  \label{corl:split}
  For every $N \in \nat$,
  \begin{enumerate}
  \item the map
    $split_N \colon s \mapsto (s [0], s[1], \cdots, s[N-1], \upc^N s)$
    is a homeomorphism of $\{0,1\}^\nat$ onto $(\{0,1\}^\nat)^{N+1}$;
  \item the image measure $split_N [\upsilon]$ is the product measure
    $\underbrace{\upsilon \otimes \upsilon \otimes \cdots \otimes
      \upsilon}_{N+1}$;
  \item the map
    $\varpi_N \colon s \mapsto (s [0], s [1], \cdots, s [N-1])$ is
    continuous from $\{0, 1\}^\nat$ to $(\{0,1\}^\nat)^N$, and
    $\varpi_N [\upsilon]$ is the product measure
    $\underbrace{\upsilon \otimes \upsilon \otimes \cdots \otimes
      \upsilon}_N$.
  \end{enumerate}
\end{corollary}
\begin{proof}
  We show 1 and 2 by induction on $N$.  This is obvious if $N=0$, and is by
  Lemma~\ref{lemma:split} if $N=1$.  If $N \geq 2$, we note that
  $split_N = (\identity {\real} \times split_{N-1}) \circ split$.
  Hence $split_N$ is a homeomorphism, using the induction hypothesis,
  and
  $split_N [\upsilon] = (\identity {\real} \times split_{N-1}) [split
  [\upsilon]] = (\identity {\real} \times split_{N-1}) [\upsilon
  \otimes \upsilon]$ (by Lemma~\ref{lemma:split})
  $= \upsilon \otimes split_{N-1} [\upsilon] = \upsilon \otimes
  \underbrace{\upsilon \otimes \upsilon \otimes \cdots \otimes
    \upsilon}_{N}$, by induction hypothesis.

  Item~3 follows from 1 and 2 by composing with the appropriate
  projection map.
\end{proof}

\begin{remark}
  \label{rem:split}
  It is also true that the image \emph{measure} $split [\upsilon]$ is
  the product \emph{measure} of $\upsilon$ by itself.  The proof is
  similar, and it suffices to observe that the sets $U_{I,J}$ form a
  $\pi$-system that generates the $\sigma$-algebra on $(\{0,1\}^\nat)^2$.
  A subtle point with that approach is that the latter is a product of
  $\{0,1\}^\nat$ with itself, certainly, but in which category?  In
  Lemma~\ref{lemma:split}, the product is taken in $\TOP$.  The
  measure-theoretic approach requires to work with the product taken
  in the category of measurable spaces and measurable maps.  In
  general, those two products differ.  The problem is similar to the
  issue discussed in Section~\ref{sec:fubini-theorem-what}.
  Fortunately, there is no such problem here, because the Borel
  $\sigma$-algebra of the topological product of two Polish spaces
  coincides with the product of the two $\sigma$-algebras.
\end{remark}

We assume an ISPCF constant
$\muxkw \colon \realT \to \intT \to \realT$ whose semantics $\muxop$
satisfies the following: for every $x \in \real$, for every
$m \in \nat$,
\begin{align*}
  \muxop (i (x)) (m) = i (\num' (\bin' (x) [m])),
\end{align*}
where $\bin'$ was introduced before Remark~\ref{rem:bin:unfold}, and
where $\num'$ is any continuous map from $\{0,1\}_\bot^\nat$ to
$\IR_\bot$ that extends $\num$, in the sense that
$\num' (s) = i (\num (s))$ for every $s \in \{0, 1\}^\nat$; for
example,
\begin{align}
  \label{eq:num'}
  \num' (s)
  & \eqdef
    \left\{
    \begin{array}{ll}
      i (\num (s)) & \text{if }s \in \{0, 1\}^\nat \\
      {} [x, x + 1/2^k]
                   & \text{where } x \eqdef \sum_{i=0}^{k-1}
                     s_i/2^{i+1}\\
                   & \text{if }s_0, \cdots, s_{k-1} \neq \bot,
                     s_k=\bot
    \end{array}
                     \right.
\end{align}

We are now ready to implement an ISPCF term computing
$\lambda_1^\infty$:
\begin{align*}
  \mathtt{rand\_uniform\_seq}
  & \eqdef \dokw {r_{\realT} \leftarrow \sample [\lambda_1]};
    \retkw (\muxkw\;r).
\end{align*}
This is a term of type $D (\intT \to \realT)$, and we claim that it
draws an infinite sequence of independent, uniformly distributed real
numbers in $[0, 1]$.

\begin{theorem}
  \label{thm:lambda1:inf}
  For every environment $\rho$, $\Eval {\mathtt{rand\_uniform\_seq}}
  \rho = j [\lambda_1^\infty]$.
\end{theorem}
\begin{proof}
  Let $\pi_N \colon \Eval {\intT \to \realT}$ map $f$ to
  $(f (0), f(1), \cdots, f (N-1))$.  We start by fixing $N$, and we
  compute $\pi_N [\Eval {\mathtt{rand\_uniform\_seq}} \rho]$.  This is
  the valuation that maps every open subset $U$ of $\IR_\bot^N$ to the
  probability that a function $f$ drawn at random according to
  $\Eval {\mathtt{rand\_uniform\_seq}} \rho$ produces a tuple
  $(f (0), f(1), \cdots, f (N-1))$ that falls in $U$.  Let us begin by
  noting that
  $\Eval {\mathtt{rand\_uniform\_seq}} \rho = \muxop [i [\lambda_1]]$.

  Let $g \colon \real \to \real^\nat$ map $x$ to
  ${(\num (\bin (x) [m]))}_{m \in \nat}$.  It is easy to see that $g$
  is measurable.  By definition, $g (x)_m = \muxop (i (x)) (m)$ for
  every non-dyadic number $x$ in $[0, 1]$, namely for every
  $x \in \real \diff E_0$ where the set $E_0$ was introduced in
  Remark~\ref{rem:bin:unfold}, and for every $m \in \nat$.

  We claim that $j [g [\lambda_1]] = \muxop [i [\lambda_1]$.  This
  needs a reminder and a comment.  We recall that $j$ is the canonical
  embedding of $\real^\nat$ into $\Eval {\intT \to \realT}$; and we
  have silently promoted $\lambda_1$ from a continuous valuation to a
  measure on $\real$, therefore working with image \emph{measures}
  rather than image \emph{valuations}.

  The claim is proved as follows.  For every Borel measurable subset
  $E$ of $\Eval {\intT \to \realT}$,
  $j [g [\lambda_1]] (E) = \lambda_1 ((j \circ g)^{-1} (E))$, while
  $\muxop [i [\lambda_1] = \lambda_1 ((\muxop \circ i)^{-1} (E))$.
  The symmetric difference of $(j \circ g)^{-1} (E)$ and of
  $(\muxop \circ i)^{-1} (E)$ is included in $E_0$, which has
  $\lambda_1$-measure zero, so
  $j [g [\lambda_1]] (E) = \lambda_1 ((j \circ g)^{-1} (E))$.

  Let us write $i^{(N)}$ for the $N$-fold product $i \times \cdots
  \times i$, and similarly with $\num^{(N)}$.  For every $x \in \real$,
  $\pi_N (\muxop (i (x))) = (i (g (x)_0), \cdots, i (g (x)_{N-1}))
  = (i (\num (\bin (x) [0])), \cdots, i (\num (\bin (x) [N-1])))
  = (i^{(N)} \circ \num^{(N)} \circ \varpi_N \circ \bin) (x)$.  It follows:
  \begin{align*}
    \pi_N [\muxop [i [\lambda_1]]
    & = i^{(N)} [\num^{(N)} [\varpi_N [\bin [\lambda_1]]]] \\
    & =  i^{(N)} [\num^{(N)} [\varpi_N [\upsilon]]]
    & \text{by Lemma~\ref{lemma:upsilon}} \\
    & =  i^{(N)} [\num^{(N)} [\upsilon \otimes \cdots \otimes \upsilon]]
    & \text{by Corollary~\ref{corl:split}, item~3} \\
    & = i [\num [\upsilon]] \otimes \cdots \otimes i [\num [\upsilon]] \\
    & = i [\lambda_1] \otimes \cdots \otimes i [\lambda_1]
    & \text{by Lemma~\ref{lemma:num}}.
  \end{align*}
  This is a statement about measures.  Restricting to open sets, we
  obtain the following statement about continuous valuations:
  \begin{quote}
    $(*)$ $\pi_N [\Eval {\mathtt{rand\_uniform\_seq}} \rho]$ is the $N$-fold
    valuation product of $i [\lambda_1]$.
  \end{quote}
  We use this to prove that $\Eval {\mathtt{rand\_uniform\_seq}} \rho
  j [\lambda_1^\infty]$, qua continuous valuations.  It suffices to
  show that the two sides of the equation coincide on basic open
  subsets of $\Eval {\intT \to \realT}$: as in the proof of
  Theorem~\ref{prop:fubini:top} or of Lemma~\ref{lemma:split},
  modularity and Scott-continuity will imply they they also coincide
  on every open set at all.

  As basic open subsets, Lemma~\ref{lemma:j} suggests that we take the
  finite intersections of sets of the form $[n \in V]$.  Without loss
  of generality, let
  $U \eqdef [0 \in V_0] \cap [1 \in V_1] \cap \cdots \cap [N-1 \in
  V_{N-1}]$.  (Any set of the form $[n \in V_n]$ where $V_n$ is the
  whose of $\IR_\bot$ can be inserted at will, since they will not
  modify the intersection.)  We note that
  $U = \pi_N^{-1} (V_0 \times V_1 \times \cdots \times V_{N-1})$, so:
  \begin{align*}
    \Eval {\mathtt{rand\_uniform\_seq}} \rho (U)
    & = \pi_N [\Eval {\mathtt{rand\_uniform\_seq}} \rho] (V_0 \times
      V_1 \times \cdots \times V_{N-1}) \\
    & = \prod_{n=0}^{N-1} i [\lambda_1] (V_n)
    \qquad \text{by }(*).
  \end{align*}
  Now, since $j (\vec x) (n) = i (x_n)$ for all
  $\vec x \in [0, 1]^\nat$ and $n \in \nat$,
  $j^{-1} (U) = \pi_N^{-1} (i^{-1} (V_0) \times i^{-1} (V_1) \times
  \cdots \times i^{-1} (V_{N-1}))$, so:
  \begin{align*}
    j [\lambda_1^\infty] (U)
    & = \pi_N [\lambda_1^\infty] (i^{-1} (V_0) \times i^{-1} (V_1) \times
      \cdots \times i^{-1} (V_{N-1})) \\
    & = \prod_{n=0}^{N-1} \lambda_1 (i^{-1} (V_n)),
  \end{align*}
  by definition of $\lambda_1^\infty$, and this concludes the proof.
\end{proof}

\subsection{More distributions on higher-order objects}
\label{sec:more-distr-high}

At this point, let us be less formal.  The construction of
$\lambda_1^\infty$ through $\mathtt{rand\_uniform\_seq}$ opens up
several avenues, which we mostly leave as programming exercises to the
reader.

For one, we can now implement the Marsaglia-Bray algorithm (see
Section~\ref{sec:box-muller-algorithm}) by drawing just \emph{one}
random real.  We remember that the Marsaglia-Bray algorithm uses
rejection sampling, which may require us to draw arbitrarily many
independent, random, $\lambda_1$-distributed real numbers.  Instead,
we can just call $\mathtt{rand\_uniform\_seq}$ once, obtaining a
random, $\lambda_1^\infty$-distributed function $f$ (modulo the
embedding $j$); each time we need another random real number, we
simply read the next value $f (n)$, for larger and larger values of
$n$.

Going further, instead of calling $\mathtt{rand\_uniform\_seq}$, we
may call $\muxkw$ on a random $\lambda_1$-distributed real number
given as argument.  We obtain an ISPCF term
$\mathtt{box\_muller\_transform}$ of type
$\realT \to \realT \times \realT$ (not
$\realT \to D (\realT \times \realT)$, since that term no longer draws
anything at random by itself) that maps a random
$\lambda_1$-distributed real number to a pair of two independent,
random $\mathcal N (0, 1)$-distributed real numbers.  We reuse
notations and conventions introduced earlier; $\mathtt{find\_pair}$
looks for the first pair of consecutive real numbers satisfying $prop$
in the infinite list $randnums$, $\mathtt{box\_muller\_engine}$
reimplements the $\mathtt{box\_muller'}$ procedure by using a given
infinite list of real numbers instead of drawing them at random.
\begin{align*}
  \mathtt{find\_pair}
  & \eqdef \lambda prop_{\realT \to \realT \to \boolT} .
    \lambda randnums_{\intT \to \realT} . \\
  & \qquad \reckw (\lambda find_{\intT \to \realT \to \realT} . \lambda n_{\intT} . \\
  & \qquad\qquad \letbe {x_{\realT} = randnums\;n} {} \\
  & \qquad\qquad \letbe {y_{\realT} = randnums\;(n+\underline 1)} {} \\
  & \qquad\qquad \ifkw (prop\;x\;y)\; \langle x, y
    \rangle\;find\;(n+\underline 2))\; \underline 0 \\
  \mathtt{box\_muller\_engine}
  & \eqdef \lambda randnums_{\intT \to \realT} . \\
  & \qquad \letbe {\langle x, y \rangle =
    \mathtt{find\_pair}\;\mathtt{discp}\;randnums} {} \\
  & \qquad \letbe {u_{\realT} = x\times x + y \times y} {} \\
  & \qquad \letbe{m_{\realT} = \sqrtkw (\underline{-2.0} \times
    \underline\log\;u / u)} {} \\
  & \qquad \langle m \times x, m \times y \rangle \\
  \mathtt{box\_muller\_transform}
  & \eqdef \lambda seed_{\realT} . \mathtt{box\_muller\_engine}\;(\muxkw\;seed).
\end{align*}

We let the reader check that the semantics of:
\[
  \dokw {seed \leftarrow \sample
    [\lambda_1]}; \retkw (\mathtt{box\_muller\_transform}\;seed)
\]
is $i [\mathcal N (0, 1) \otimes \mathcal N (0, 1)]$, just like
$\mathtt{box\_muller'}$ and $\mathtt{box\_muller}$.

The tricks described above suggest the following scheme for
implementing the infinite product $\mathcal N (0, 1)^\infty$.  We call
$\mathtt{rand\_uniform\_seq}$ in order to obtain an infinite,
independent sequence of $\lambda_1$-distributed real numbers.  We then
apply $\mathtt{box\_muller\_transform}$ to each.  Explicitly, we form
the ISPCF term:
\[
  \dokw {f_{\intT \to \realT} \leftarrow \mathtt{rand\_uniform\_seq}};
  \retkw {\lambda n_{\intT} . \mathtt{box\_muller\_transform} (f\;n)},
\]
of type $D (\intT \to \realT)$.

Again, this can be made into a transformer
\[
  \mathtt{box\_muller\_sequence\_transform} \colon \realT \to D (\intT
  \to \realT)
\]
that maps any random $\lambda_1$-distributed real number to a random
$\mathcal N (0, 1)^\infty$-distributed infinite sequence of real
numbers.  This can also be iterated.  We let the reader explore around
this idea.  For example, given two maps
$mean, sigma \colon \intT \to \realT$, build a term that maps any
random $\lambda_1$-distributed real number to a random infinite
sequence of independent real numbers, whose $n$th entry follows a
normal distribution with mean $mean (n)$ and with standard deviation
$sigma (n)$.

One can implement infinite independent sequences of random real
numbers following various other distributions in the same way.  For
example, do this for infinite products of exponential distributions.

This machinery can also be used to implement non-trivial distributions
over distributions.  A typical, and useful, application is the
so-called \emph{Dirichlet process} \cite{Ferguson:DP}, which is widely
used as a source of Bayesian priors in so-called nonparametric
estimation problems.  Given a base distribution $H$ over some space
$X$, and a parameter $\alpha \in ]0, 1[$, the Dirichlet process
$DP (H)$ is a distribution over the space of probability distributions
over $X$, which one may describe as follows.  We draw an infinite
sequences of elements $x_n$, $n \in \nat$, from $X$.  At step $n$,
with probability $\alpha/(\alpha+n)$, we draw $x_n$ from $X$,
independently from all previous values, using the distribution $H$;
otherwise, we draw $x_n$ at random among the previous values, each
having the same probability, in other words with respect to the
distribution $\frac 1 n \sum_{i=0}^{n-1} \delta_{x_i}$.  (This is
undefined if $n=0$, but then, if $n=0$, that second case happens with
zero probability.)  This process yields a random infinite sequence,
which we can use directly, or which we can convert to a distribution,
the \emph{directing measure}, mapping every measurable subset $E$ of
$X$ to the limit of $(\frac 1 n \sum_{i=0}^{n-1} \delta_{x_i}) (E)$ as
$n$ tends to $+\infty$, which exists $H$-almost surely by de Finetti's
theorem \cite[Chapter~1.1]{Kallenberg:proba}, based on the fact that
the sequence ${(x_n)}_{n \in \nat}$ is a so-called exchangeable
sequence.

In principle, one can compute the directing measure from any so-called
exchangeable sequence, as shown by Daniel Roy
\cite[Chapter~IV]{Roy:PhD}.  In the special case of the Dirichlet
process, there is a popular, and more efficient, implementation of
$DP (H)$, using the so-called \emph{stick-breaking algorithm}, see
Figure~2 of \cite{GMRBT:Church} for example.  (We also refer to that
paper for an application of the Dirichlet process to an interesting,
stochastic form of memoization.)  While the latter relies on memoizing
facilities (hence on global state, possibly hidden), one can implement
it in ISPCF as follows, given $H^\infty$ (not $H$!) as argument.

We assume that $X$ is the denotation $\Eval \tau$ of some type $\tau$,
and that the variable $proc_{D (\intT \to \tau)}$ holds a
representation of $H^\infty$.  We draw a random function
$atoms \colon \intT \to \tau$ with respect to $H^\infty$,
representing some infinite sequence ${(\theta_n)}_{n \in \nat}$ of
elements of $X$; this uses $\dokw$ and $proc$.  We also draw a random
function $sticks \colon \intT \to \realT$ with respect to the
distribution $\text{Beta} (1, \alpha)^\infty$ (the infinite product of
the so-called $\beta$ distribution with parameters $1$ and $\alpha$),
representing a random sequence ${(\beta'_n)}_{n \in \nat}$.  We define
a new random (but no longer independent) sequence
${(\beta_n)}_{n \in \nat}$, where
$\beta_n \eqdef \beta'_n . \prod_{i=0}^{n-1} (1-\beta'_i)$.  and we
return the (random) distribution
$\sum_{n \in \nat} \beta_n \delta_{\theta_n}$.

While that would be possible in ISPCF, we do not compute any $\beta_n$
explicitly.  Instead, the final distribution is obtained by the
following informal procedure: with probability $\beta'_0$, return
$\theta_0$; else, with probability $\beta'_1$, return $\theta_1$;
else, with probability $\beta'_2$, return $\theta_2$, and so on.  This
terminates with probability $1$.

This can be implemented in ISPCF by drawing a third function
$rand \colon \intT \to \realT$ at random with respect to
$\lambda_1^\infty$, then returning
$atoms \; (\mathtt{pick\_a\_stick}\;\allowbreak sticks\;\allowbreak
rand\;0)$, where:
\begin{align*}
  \mathtt{pick\_a\_stick}
  & \eqdef \lambda sticks_{\intT \to \realT} . \lambda rand_{\intT \to
    \realT} . \reckw (\lambda pick_{\intT \to \intT} . \\
  & \qquad \lambda j_{\intT} . \\
  & \qquad\qquad \ifkw (\poskw (sticks\;j - rand\;j))\\
  & \qquad\qquad\qquad j\\
  & \qquad\qquad (pick\; (j+1))).
\end{align*}
We have taken the same procedure names as in
\cite[Figure~2]{GMRBT:Church}, in the hope that this will ease a
comparison with the Church implementation of the Dirichlet process
given there.

By now, we hope to have provided enough examples in order to convince
the reader that one can write enough useful distributions in a
language with a domain-theoretic semantics such as ISPCF.  Let us
proceed with matters of operational semantics.

\section{Operational Semantics}
\label{sec:ideal-oper-semant}

The simplest possible operational semantics of ISPCF is one where
$\sample [\mu]$ draws actual real numbers at random, following the
distribution $\mu$.  We follow similar other proposals
\cite{VKS:SFPC,DLH:geom:bayes,EPT:PPCF} pretty closely.  The main
difference is that we will draw \emph{exact} reals, namely elements of
$\IR_\bot$, instead of true reals in $\real$, at random.  This is
really not much of a difference, since the probability that
$\sample [\mu]$ draws an exact real number that is not a true real
number is zero.

In order to describe our operational semantics formally, we consider
generalized ISPCF terms, with extra constants, one for each real
number, and we consider these generalized ISPCF terms as
configurations of an abstract machine.  We then define a probabilistic
transition relation on the space of configurations.  In order to do
so, we will need to topologize the space of configurations, much as
previous proposals \cite{VKS:SFPC,DLH:geom:bayes,EPT:PPCF} defined a
$\sigma$-algebra on similar spaces of configurations.  The
probabilistic transition relation will be a map from configurations to
measures (rather, continuous valuations) on the space of
configurations.  That map will be continuous, and not just measurable.

Notionally, a \emph{generalized ISPCF term} is a closed ISPCF term
built on a set of constants $\Sigma$ that contains exactly one
constant $\underline {\mathbf a}$ for each $\mathbf a \in \IR$.  We
will equate the already existing constants $\underline r$,
$r \in \real$, with $\underline {[r,r]}$.

The following alternate definition will be more formal, and will allow
us to give a simple description of the topology we will put on
generalized ISPCF terms.  We fix a countable enumeration $x_1$, $x_2$,
\ldots, $x_n$, \ldots{} of so-called \emph{template} variables of type
$\realT$ once and for all, in such a way that we still have an
infinite supply of variables of type $\realT$ not in that list.  A
\emph{template} is an ISPCF term $M$ whose free variables, read from
left to right, are $x_1$, \ldots, $x_k$ for some $k \in \nat$, and
which does not contain any constant of type $\realT$.  In particular,
the only free variables of $M$ are template variables, they occur only
once in $M$, and are numbered from $1$ to $k$ consecutively from left
to right.  We let $\FV (M)$ be the set of free variables of a term
$M$.  Then we equate a generalized ISPCF term with a
\emph{configuration} $(M, \theta)$ where $M$ is a template and
$\theta \in \IR_\bot^{\FV (M)}$.  Replacing each template variable
$x_i$ that is free in $M$ by $\underline {\theta (x_i)}$ yields a
generalized ISPCF term, which we write as $M \theta$.  Conversely,
each generalized ISPCF term is represented by a unique configuration
$(M, \theta)$: $M$ has to be the given generalized ISPCF term where
each constant $\underline r$ of type $\realT$ has been replaced by a
fresh template variable, which is then mapped to $r$ by $\theta$.  We
call $M$ the \emph{shape} of the generalized ISPCF term.
For example, the generalized ISPCF term
$\lambda y . \underline {2.0} \times y + \underline {1.0}$ is
represented by the configuration
$(\lambda y . x_1 \times y + x_2, [x_1 \mapsto 2.0, x_2 \mapsto
1.0])$.

\begin{definition}[Space of configurations $\Config$]
  \label{defn:Config}
  $\Config$ denotes the set of all configurations, ordered by
  $(M, \theta) \leq (N, \theta')$ if and only if $M=N$ and
  $\theta \leq \theta'$, namely, for every $x \in \FV (M)$,
  $\theta (x) \leq \theta' (x)$.
\end{definition}

\begin{remark}
  \label{rem:Config:cont}
  $\Config$ is a dcpo, and in fact a continuous dcpo.  To see the
  latter, $\Config$ is a coproduct over all possible shapes of finite
  products of copies of $\IR_\bot$.  Each copy of $\IR_\bot$ is a
  continuous dcpos, and continuous dcpos are closed under finite
  products and arbitrary coproducts.
\end{remark}

This structure transports to a structure of continuous dcpo on the set
of generalized ISPCF terms through the bijection
$(M, \theta) \mapsto M \theta$.


While our algebra of terms was rather open-ended until now, we will
need to restrict it slightly, and we will assume that $\Sigma$ has
\emph{strictly observable first-order constants}, see below.
\begin{definition}[Basic and observable types]
  \label{defn:obs:type}
  The algebra of \emph{observable types} is:
  \begin{align*}
    \beta & ::= \unitT \mid \voidT \mid \intT \mid
            \realT
            \mid \beta_1 + \beta_2
            \mid \beta_1 \times \beta_2,
  \end{align*}
  The \emph{basic types} are $\unitT$, $\voidT$, $\intT$, and $\realT$.
\end{definition}
In other words, the observable types are all types that do not contain
$\to$ or $D$.
\begin{definition}[Observable element]
  \label{defn:partial}
  The \emph{observable elements} of $\Eval \beta$, where $\beta$ is an
  observable type, are defined as follows.  The observable elements of
  $\Eval {\voidT}$, $\Eval {\unitT}$, $\Eval {\intT}$,
  $\Eval {\realT}$ are all their elements except $\bot$.  The
  observable elements of $\Eval {\beta_1 \times \beta_2}$ are the
  pairs $(b_1, b_2)$ where both $b_1$ and $b_2$ are observable.  The
  observable elements of $\Eval {\beta_1+\beta_2}$ are the elements
  $\iota_1 b$ or $\iota_2 b$, where $b$ is observable.
\end{definition}

For every observable value $a$ of an observable type $\beta$, there is
a generalized ISPCF term $\underline a$ of type $\beta$ such that
$\Eval {\underline a} \rho = a$ for every environment $\rho$.
This extends our preexisting underlining notations.
\begin{definition}
  \label{defn:underline}
  The constants $\underline n \colon \intT$ ($n \in \Z$), $\underline
  * \colon \unitT$, and $\underline {\mathbf a}$ ($\mathbf a \in \IR$)
  are already defined.  We let $\underline {(a, b)} \eqdef \langle
  \underline a, \underline b\rangle$, $\underline {(1, a)} \eqdef
  \iota_1 \underline a$, $\underline {(2, b)} \eqdef \iota_2
  \underline b$.
\end{definition}
In order to prevent any conflict of notations, we will require
that the only preexisting constants $\underline f$ of ISPCF of
observable types are actually of basic types.  For example, we do not
allow any constant $\underline f$ of type $\intT \times \realT$ in our syntax.

\begin{definition}[Strictly observable first-order constants]
  \label{defn:assum:strict}
  The set $\Sigma$ consists of \emph{strictly observable first-order constants} if
  and only if each constant $\underline f$ in $\Sigma$:
  \begin{itemize}
  \item has a \emph{first-order type}, namely one of the form
    $\sigma_1 \to \cdots \to \sigma_k\to \tau$, where each $\sigma_i$
    is a basic type and $\tau$ is an observable type; we call $k$ the
    \emph{arity} of $\underline {f}$, and we write it as
    $\alpha (\underline f)$;
  \item its semantics $f$ is \emph{strictly observable}: $f (v_1) \cdots (v_k)$ is
    either observable or equal to $\bot$, and $f (v_1) \cdots (v_k) =
    \bot$ if $v_1$, \ldots, or $v_k$ is equal to $\bot$;
  \item the only preexisting constants $\underline f$ of ISPCF of
    observable types are of basic types, and their semantics $f$ is
    observable.
  \end{itemize}
\end{definition}
All the constants we have used in our examples are first-order, for
example $\underline{\log} \colon \realT \to \realT$, but also
$\poskw \colon \realT \to \boolT$, where $\boolT = \unitT + \unitT$ is
observable but not basic.

\begin{figure}
  \centering
  \begin{align}
    \nonumber
    \text{\emph{Probabilistic rules:}} \\
    \label{rule:sample}
    E^0 [\sample [\mu]]
                 & \to E^0 [\retkw \underline a]
                   \quad (a \in \real)
    \\
    \label{rule:score}
    E^0 [\score\; \underline {\mathbf a}] & \to E^0 [\retkw \underline
                                *]
    \\
    \nonumber
    E^0 [\dokw {x \leftarrow \retkw M}; N]
                 & \to E^0 [N [x:=M]] \\
    \nonumber
    \text{\emph{Deterministic rules:}} \\
    \nonumber
    E [\reckw\; M] & \to E [M (\reckw\; M)] \\
    \nonumber
    E [(\lambda x . M) N] & \to E [M [x:=N]] \\
    \label{rule:f}
    E [\underline f\; \underline {a_1} \cdots \underline {a_k}]
                 & \to E [\underline {f (a_1) \cdots (a_k)}]
    \\
    \nonumber
                 & \text{if }k=\alpha (\underline f) \neq 0 \text{ and }
                   f (a_1) \cdots (a_k)\text{ is observable}\\
    \nonumber
    E [\pi_1 \langle M, N\rangle] & \to E [M] \\
    \nonumber
    E [\pi_2 \langle M, N\rangle] & \to E [N] \\
    \nonumber
    E [\casekw (\iota_1 M) N P] & \to E [NM] \\
    \nonumber
    E [\casekw (\iota_2 M) N P] & \to E [PM]
  \end{align}
  \caption{The raw operational semantics}
  \label{fig:ideal:opsem}
\end{figure}

Our operational semantics will require two kinds of evaluation
contexts.  The point is that, given a generalized ISPCF term $M$ of
type $D \tau$, the computation starting from $M$ will in general go
through two phases $M \to^* \retkw N \to^* \retkw V$.  The first
phase $M \to^* \retkw N$ will be allowed to use all rules, including
rules that operate probabilistic choices, and ends at some term of the
form $\retkw N$ (or loops forever); the second phase reduces $N$ to a
value $V$ deterministically.  The second phase applies (deterministic)
rules under arbitrary evaluation contexts, while the first phase
applies (arbitrary) rules under some so-called weak evaluation contexts.

\begin{definition}[Evaluation contexts, values]
  \label{defn:context}
  The \emph{evaluation contexts} $E$, the \emph{weak evaluation
    contexts} $E_0$, the \emph{values} $V$ and the \emph{lazy values}
  $V^0$ are given by the following (pseudo-)grammar:
  \begin{align*}
    E & ::= E^0 \mid \retkw E \mid \langle E, M \rangle \mid \langle
    V, E \rangle \mid \iota_1 E \mid \iota_2 E \\
    E^0 & ::= [] \mid \score E^0 \mid E^0 N \mid \dokw {x \leftarrow E^0}; N
          \mid \underline f \;\underline a_1\; \cdots\; \underline a_{i-1}\; E^0  \quad (i \leq \alpha
          (\underline f))\\
      & \quad\mid \casekw E^0 N P \mid \pi_1 E^0 \mid \pi_2 E^0 \\
    V & ::= \underline a
        \mid \lambda x . M
        \mid \retkw V \mid \langle V, V \rangle \mid \iota_1 V
        \mid \iota_2 V
        \mid \underline f \;\underline a_1\; \cdots\; \underline a_i \quad (i < \alpha
        (\underline f)) \\
    V^0 & ::= \underline a
          \mid \lambda x . M
          \mid \retkw M \mid \langle M, N \rangle \mid \iota_1 M
    \mid \iota_2 M
          \mid \underline f \;\underline a_1\; \cdots\; \underline a_i \quad (i < \alpha
        (\underline f))
  \end{align*}
  where $\underline a$ ranges over the zero-ary constants of basic
  types, $N$, $P$ range over generalized ISPCF terms, and
  $r \in \real$.  The evaluation context $[]$ is the \emph{hole}.
\end{definition}

Evaluation contexts, weak evaluation contexts, values and lazy values
are typed, although we will often omit mentioning it.  The typing
rules are as follows.  Every value of the form $\underline{\mathbf a}$
with $\mathbf a \in \IR$ is of type $\realT$, and the type of the
other values is their type as generalized ISCPF term.  The types of
evaluation contexts $E$ (resp., $E^0$) are of the form
$\sigma \vdash \tau$, in such a way that for every $M \colon \sigma$,
$E [M]$ is of type $\tau$.  We omit the precise typing rules.  The
notation $E [M]$ stands for $E$ where the hole $[]$ is replaced by
$M$.

Using this, we define the \emph{raw} operational semantics of ISPCF as
the smallest binary relation $\to$ between generalized ISPCF terms
satisfying the clauses of Figure~\ref{fig:ideal:opsem}.  The raw
semantics does not mention any probabilistic information.  This is dealt
with by the following definition.  We recall that
$i \colon \real \to \IR_\bot$ is the embedding $r \mapsto [r, r]$.
\begin{definition}
  \label{defn:opsem}
  Assume that all constants in $\Sigma$ are first-order constants.
  The operational semantics of ISPCF is the map $Next \colon \Config
  \to \Val \Config$ defined by:
  \begin{itemize}
  \item for every instance of a rule $L \to R$ of
    Figure~\ref{fig:ideal:opsem}
    except (\ref{rule:sample}), (\ref{rule:score}) and (\ref{rule:f}),
    $Next (L) = \delta_R$;
  \item (case of $\sample$) $Next (E^0 [\sample [\mu]])$ is the
    continuous valuation $f_{E^0} [\mu]$, where
    $f_{E^0} (a) \eqdef E^0 [\retkw \underline a]$ for every
    $a \in \real$;
  \item (case of $\score$)
    $Next (E^0 [\score\; \underline {\mathbf a}]) \eqdef |\mathbf a|
    . \delta_{E^0 [\retkw \underline *]}$;
  \item (case of $\underline f$) Given
    $k = \alpha (\underline f) \neq 0$,
    $Next (E [\underline f\; \underline a_1 \cdots \underline a_k])
    \eqdef \delta_{E [\underline {f (a_1) \cdots (a_k)}]}$ if
    $f (a_1) \cdots (a_k)$ is observable, the constant zero valuation
    otherwise.
  \end{itemize}
\end{definition}
We let the reader check that every generalized ISPCF term parses in at
most one way as the left-hand side of a rule of
Figure~\ref{fig:ideal:opsem}, so that Definition~\ref{defn:opsem} is
non-ambiguous.  To see this, one can write any generalized ISPCF term
as $E [M_1]$, where $E$ has maximal depth, in a unique way.  If
$M_1 = \sample [\mu]$, then only rule (\ref{rule:sample}) applies,
and only if $E$ is weak.  If $M_1 = \underline {\mathbf a}$, then
only rule (\ref{rule:score}) applies, and only if $E$ is of the form
$E^0 [\score []]$.  If $M_1 = \retkw M$ and
$E = E^0 [\dokw {x \leftarrow []}; N]$, then only the third rule
applies.  And so on, and we see that all the cases are mutually
exclusive.


\begin{proposition}
  \label{prop:next:cont}
  Let $\Sigma$ consist of strictly observable first-order constants.  The map
  $Next$ is Scott-continuous from $\Gamma$ to $\Val \Gamma$.
\end{proposition}
\begin{proof}
  Let $L$ be any generalized ISPCF term, $\mathcal V$ be any
  Scott-open subset of $\Val \Config$, and let us assume that
  $Next (L) \in \mathcal V$.  We will show that there is a Scott-open
  neighborhood $U$ of $L$ such that
  $U \subseteq Next^{-1} (\mathcal V)$.  $U$ will consist of terms of
  the same shape $P$ as $L$.

  If $L$ is of the form $E [\reckw\; M]$, then $Next (L) = \delta_R$
  where $R \eqdef E [M (\reckw\; M)]$ is computed as follows: writing
  $L$ as $P \theta$ for some unique configuration $(P, \theta)$, $R$
  is obtained as $Q \theta'$ for some unique configuration
  $(Q, \theta')$, where $Q$ is obtained from $P$ alone (independently
  of $\theta$), and $\theta' = \theta \circ \sigma$ for some function
  $\sigma \colon \FV (Q) \to \FV (P)$ that is determined from $P$ alone,
  again.  For example, if
  $L = \underline {1.0} + \reckw (\lambda x . x \times
  \underline{2.0}) (\underline {3.0})$, then
  $P = x_1 + \reckw (\lambda x . x \times x_2) (x_3)$,
  $Q = x_1 + (\lambda x . x \times x_2) (\reckw (\lambda x . x \times
  x_3)) (x_4)$, and $\sigma$ maps $x_1$ to $x_1$, $x_2$ and $x_3$ to $x_2$,
  and $x_4$ to $x_3$.

  Since $\delta_R$ is in $\mathcal V$, $R$ is in
  $V \eqdef \eta_\Config^{-1} (\mathcal V)$, and $\theta'$ is in the
  open subset $\{\theta' \in \IR^{\FV (Q)} \mid Q \theta' \in V\}$.
  The map $S \colon \theta \mapsto \theta \circ \sigma$ is clearly
  Scott-continuous.  It follows that
  $U \eqdef \{P \theta \mid \theta \in S^{-1} (\{\theta' \in \IR^{\FV
    (Q)} \mid Q \theta' \in V\})\}$ is Scott-open in $\Config$.
  Moreover, by definition, for every generalized ISPCF term $P \theta$
  in $U$, $Next (P)$ is in $\mathcal V$.

  We reason similarly for all other rules except (\ref{rule:sample}),
  (\ref{rule:score}), and (\ref{rule:f}).  The case of (\ref{rule:f})
  is not that different.  In the final step, $U$ is equal to the
  intersection of
  $\{P \theta \mid \theta \in S^{-1} (\{\theta' \in \IR^{\FV (Q)} \mid
  Q \theta' \in V\})\}$ with the set $U'$ of generalized ISPCF terms
  of shape $P$ on which rule (\ref{rule:f}) applies; then we will be
  able to conclude as above, once we prove that $U'$ is open.  In
  order to do so, and since we are in the case of (\ref{rule:f}), $P$
  is of the form $E [\underline f\;M_1 \cdots M_k]$, where each $M_i$
  is either a constant (if its type is not $\realT$) or a template
  variable.  Let $x_{j+1}$, \ldots, $x_{j+p}$ be the template
  variables occurring among the terms $M_i$, listed from left to
  right, and say that $x_{j+1} = M_{i_1}$, \ldots,
  $x_{j+p} = M_{i_p}$, where $1 \leq i_1 < \cdots i_p \leq k$.  Let
  also $M_i$ be the constant $\underline a_i$ for each $i$ not among
  $i_1$, \ldots, $i_p$.  Then
  $U' = \{P \theta \mid \theta \in \IR^{\FV (P)}, f (a_1) \cdots
  (a_{i_1-1}) (\theta (x_{j+1})) (a_{i_1+1}) \cdots (a_{i_2}-1)
  (\theta (x_{j+2})) (a_{i_2+1}) \cdots (a_{i_p-1}) (\theta (x_{k+p}))
  (a_{i_p+1}) \cdots (a_k) \in Obs_\beta\}$, where $\beta$ is the
  (observable) type of $f\;M_1 \cdots M_k$, and $Obs_\beta$ is the
  subset of observable elements of $\Eval \beta$.  It is easy to see
  that $Obs_\beta$ is Scott-open.  Since $f$ is Scott-continuous, $U'$
  is open.
  
  In the case of (\ref{rule:score}),
  $L = E^0 [\score\; \underline {\mathbf a}]$.  We write $L$ as
  $P \theta$ for some unique configuration $(P, \theta)$, once again.
  Then
  $Next (L) = |\mathbf a| . \delta_{E [\retkw \underline *]}$
  where $E^0 [\retkw \underline *]$ can be written as
  $Q \theta'$ for some configuration $(Q, \theta')$.  Once again, $Q$
  is determined as a function of $P$ alone.  In fact, if $\mathbf a$
  is the $m$th occurrence of a constant of type $\realT$ in $L$ (so
  that $\theta (x_m) = \mathbf a$; recall that $x_m$ is the $m$th
  template variable), and $\FV (P) = \{x_1, \cdots, x_n\}$ with
  $n \geq m$, then $\FV (Q) = \{x_1, \cdots, x_{n-1}\}$, and
  $\theta' = \theta \circ \sigma$, where $\sigma$ maps $x_1$ to $x_1$,
  \ldots, $x_{m-1}$ to $x_{m-1}$, and $x_m$ to $x_{m+1}$, \ldots,
  $x_{n-1}$ to $x_n$.  The map $\mathbf a \mapsto |\mathbf a|$ is
  Scott-continuous from $\IR$ to $\creal$.  Product is a
  Scott-continuous map on $\creal$, from which we obtain easily that
  $(a, \nu) \mapsto a \nu$ is Scott-continuous from
  $\creal \times \Val \Config$ to $\Val \Config$.  Using that
  $\eta_\Config$ is also Scott-continuous, the map
  $(\mathbf a, \theta') \in \IR \times \IR^{\FV (Q)} \mapsto |\mathbf
  a| . \delta_{Q \theta'}$ is Scott-continuous, and therefore
  $S \colon \theta \in \IR^{\FV (P)} \mapsto |\theta (x_m)|
  . \delta_{Q (\theta \circ \sigma)}$ is also Scott-continuous.  We
  then define $U$ as
  $\{P \theta \mid \theta \in S^{-1} (\mathcal V)\}$.  This is open,
  contains $L$, and its image by $Next$ in included in $\mathcal V$ by
  definition.

  In the case of (\ref{rule:sample}), finally,
  $L = E^0 [\sample [\mu]]$.  This time, $L = P\theta$ for some unique
  configuration $(P, \theta)$.  For every $a \in \real$,
  $f_{E^0} (a) = E^0 [\retkw \underline a]$ can be written as
  $Q \theta' (a)$, where $(Q, \theta' (a))$ is a configuration such
  that $Q$ is obtained from $P$ alone, as in previous cases.
  Additionally, there is a number $m$ (the number of the template
  variable that is replaced by $\underline a$ in $Q$), the template
  variables of $P$ are $x_1, \cdots, x_n$ with $n \geq m$, the
  template variables of $Q$ are $x_1, \cdots, x_{n+1}$, and $\theta'$
  is obtained as $S (\theta, i (a))$, where the map $S$ is defined by:
  for all $\theta \in \IR^{\FV (P)}$ and $\mathbf a \in \IR$,
  $S (\theta, \mathbf a)$ maps $x_m$ to $\underline {\mathbf a}$,
  every $x_i$ with $i < m$ to $\theta (x_i)$, and every $x_i$ with
  $i > m$ to $\theta (x_{i-1})$.  $S$ is clearly Scott-continuous.  In
  particular, $f_{E^0} = Q\;S (\theta, i (\_))$ is lower semicontinuous,
  so $f_{E^0} [\mu]$ makes sense.  Moreover, the map
  $\theta \mapsto f_{E^0} (a) = Q\;S (\theta, i (a))$ is also
  Scott-continuous for every $a \in \real$, so the map that sends
  every $\theta$ to
  $f_{E^0} [\mu] = \lambda V \in \Open \Config . \mu \{a \in \real \mid
  Q\;S (\theta, i (a)) \in V\}$ is also Scott-continuous.  Therefore
  $U \eqdef \{P \theta \in \IR^{\FV (P)} \mid f_{E^0} [\mu] \in \mathcal
  V\}$ is open. By construction, $L$ is in $U$, and every element of
  $U$ maps to an element of $\mathcal V$ by $Next$.
\end{proof}
As a corollary, $Next$ is Borel measurable from $\Config$ to $\Val
\Config$, and is therefore a \emph{kernel} in the sense of
\cite{VKS:SFPC}; not just a measurable kernel, but a continuous kernel.

\subsection{Soundness}
\label{sec:soundness}

\begin{definition}[Redex, blocked term, normal form]
  \label{defn:normal}
  A \emph{redex} is a generalized ISPCF term that is of any of the ten
  forms found at the left of the rules of
  Figure~\ref{fig:ideal:opsem}.  A \emph{blocked term} is a
  generalized ISPCF term of the form
  $E [\underline f\; \underline {a_1} \cdots \underline {a_k}]$ where
  $E$ is an evaluation context and $f (a_1) \cdots (a_k)$ is not
  observable.  A \emph{normal form} is a generalized ISPCF term that
  is neither a redex nor a blocked term.
\end{definition}
All three categories are disjoint.  For every blocked term or normal
form $L$, $Next (L)$ is the zero valuation.

We will also need the following notions.
\begin{definition}[Weak reduction, weak normal forms]
  \label{defn:normal:weak}
  A \emph{weak} instance of a rule of Figure~\ref{fig:ideal:opsem} is
  one where $E$, not just $E^0$, is a weak evaluation context.  A
  \emph{weak redex} is a left-hand side of a weak instance of a rule.
  A \emph{weak blocked term} is a generalized ISPCF term of the form
  $E^0 [\underline f\; \underline {a_1} \cdots \underline {a_k}]$
  where $E^0$ is a weak evaluation context and $f (a_1) \cdots (a_k)$
  is not observable.  A \emph{weak normal form} is a generalized ISPCF
  term that is neither a weak redex nor a weak blocked term.
\end{definition}

\begin{lemma}
  \label{lemma:normal:weak}
  Let $\Sigma$ consist of strictly observable first-order constants.
  The weak normal forms are exactly the lazy values.
\end{lemma}
\begin{proof}
  It is clear that every lazy value is in weak normal form.
  Conversely, we show that every weak normal form $M$ is a lazy value
  $V^0$, by induction on the size of $M$.

  We start with the case where $M$ is a constant $\underline f$.  If
  $\alpha (\underline f)=0$, then $M$ is a zero-ary constant of
  observable type.  By the last condition of
  Definition~\ref{defn:assum:strict}, it is of basic type, and
  therefore it is a lazy value.  If $\alpha (\underline f) \neq 0$,
  then it is a lazy value of the form
  $\underline f \;V_1\; \cdots\; V_i$, with
  $i < \alpha (\underline f)$ (namely with $i=0$).

  If $M$ is of the form $\score\; N$, then $N$ is weakly normal, since
  $\score []$ is a weak evaluation context and $M$ is weakly normal.
  We apply the induction hypothesis to $N$, and we realize that
  because of typing, $N$ must be of the form $\underline {\mathbf a}$
  for some zero-ary constant $\underline {\mathbf a}$ of type
  $\realT$.  However, $M = \score \;\underline {\mathbf a}$ is not
  weakly normal, so $M$ cannot be of the form $\score\; N$ after all.
  
  $M$ cannot be of the form $\reckw\; N$ or $\sample [\mu]$, which are
  not weakly normal.  Lambda-abstractions, and terms of the form
  $\retkw N$, $\langle N, P \rangle$, $\iota_1 N$, or $\iota_2 N$ are
  all lazy values.

  If $M$ is of the form $\pi_1 N$, then $N$ is weakly normal, so
  by induction hypothesis and typing, $N$ is a pair, which would
  contradict the normality of $M$; similarly if $M = \pi_2 N$.  If $M$
  is of the form $\casekw N P Q$, then $N$ is weakly normal, so by
  induction hypothesis and typing it is of the form $\iota_1 N_1$ or
  $\iota_2 N_2$, and that would contradict the normality of $M$ as well.

  If $M$ is of the form $\dokw {x \leftarrow N}; P$, then by induction
  hypothesis and typing $N$ is of the form $\retkw P$, which would
  contradict the fact that $M$ is weakly normal.

  If $M$ is an application $N P$, then by induction hypothesis and
  typing, either $N$ is a $\lambda$-abstraction (which is impossible
  since $M$ is weakly normal), or $N$ is of the form
  $\underline f \;M_1\; \cdots\; M_i$ with $i < \alpha (\underline f)$
  and where $M_1$, \ldots, $M_i$ are weakly normal.  Then
  $M = \underline f \;M_1\; \cdots\; M_i P$, where $M_1$, \ldots,
  $M_i$ and $P$ are weakly normal, hence are lazy values.  By the
  first condition of Definition~\ref{defn:assum:strict}, they are lazy
  values of basic types, hence are constants of arity $0$.  If
  $i+1 = \alpha (\underline f)$, either rule (\ref{rule:f}) would
  apply or $M$ would be blocked.  Hence $i+1 < \alpha (\underline f)$,
  showing that $M$ is a lazy value.
\end{proof}

Let $Norm_\tau$ be the set of normal forms of type $\tau$.
\begin{lemma}
  \label{lemma:normal}
  Let $\Sigma$ consist of strictly observable first-order constants.
  For every type $\tau$, $Norm_\tau$ is exactly the set of values of
  type $\tau$.
\end{lemma}
\begin{proof}
  It is clear that every value is in normal form.  Conversely, we
  argue that every element $M$ of $Norm_\tau$ is a value, by induction
  on $M$.  Every normal form is weakly normal, since every weak
  evaluation context $E^0$ is a context $E$; so $M$ is a lazy value.
  In order to show that it is a value, we only need to show that if it
  is of the form $\retkw N$, $\langle N, P \rangle$, $\iota_1 N$ or
  $\iota_2 N$, then $N$ and $P$ are values.  If $M = \retkw N$, this
  is because $\retkw []$ is an evaluation context, so that $N$ is a
  value by induction hypothesis.  Similarly when $M = \iota_1 N$ or
  $M = \iota_2 N$.  When $M = \langle N, P \rangle$, since
  $\langle [], P\rangle$ is an evaluation context, by induction
  hypothesis $N$ is a value $V$.  Since $\langle V, [] \rangle$ is a
  context, $P$ is a value.
\end{proof}

The following lemma will apply to $\kappa \eqdef Next$, where
$X \eqdef \Config$, and where $Norm$ is the set of normal forms. For
every $x \in X$, for every open subset $U$ of $Norm$,
$Next^{\leq n} (x) (U)$ is the `probability' that we will reach $U$,
starting from $x$, in at most $n$ $Next$ steps, and $Next^* (x) (U)$
is the `probability' that we will reach $U$ starting from $x$ in any
number of steps.  (The quotes around `probability' reflect the fact
that those probabilities need not be bounded by $1$.)  This way of
defining `probabilities' of reaching $U$ is inspired from Lemma~3.9
and Equation~(6) of \cite{EPT:PPCF}.  Let us also recall the
restriction $\mu_{|U}$ of a continuous valuation to an open subset
$U$, defined by $\mu_{|U} (V) \eqdef \mu (U \cap V)$.

\begin{lemma}
  \label{lemma:iter}
  Let $X$ be a topological space, and let $Norm$ and $\overline{Norm}$
  be two complementary open subsets of $X$.  Let $\kappa \colon X \to \Val X$ be
  a continuous map, and let us assume that for every $x \in
  Norm$, $\kappa (x)$ is the zero valuation.  We define:
  \begin{itemize}
  \item $\kappa^{\leq 0}$ as mapping every $x \in Norm$ to
    $\delta_x$, and every $x \in \overline{Norm}$ to the zero valuation;
  \item $\kappa^{\leq 1} \eqdef \kappa + \kappa^{\leq 0}$;
  \item $\kappa^{\leq n+1} \eqdef (\kappa^{\leq n})^\dagger \circ
    \kappa^{\leq 1}$, for every $n \geq 1$.
  \end{itemize}
  Then:
  \begin{enumerate}
  \item the maps $\kappa^{\leq n}$ are continuous from $X$ to $\Val X$;
  \item for every $x \in X$, the family
    ${(\kappa^{\leq n} (x)_{|Norm})}_{n \geq 1}$ is monotonically
    increasing;
  \item the formula
    $\kappa^* (x) \eqdef \sup_{n \geq 1} \kappa^{\leq
      n} (x)_{|Norm}$ defines a
    continuous map $\kappa^*$ from $X$ to $\Val X$.
  \end{enumerate}
\end{lemma}
\begin{proof}
  1.   Since $X$ is homeomorphic to the topological coproduct $\overline{Norm} +
  Norm$, the continuity of $\kappa^{\leq 0}$ reduces to that
  the unit $\eta_{Norm}$ and of the constant zero map.  Then
  $\kappa^{\leq 1}$ is continuous since $+$ is Scott-continuous on
  $\creal$, and by induction on $n \geq 1$, $\kappa^{\leq n+1}$ is
  continuous as a composition of continuous maps.

  2. We show that $\kappa^{\leq n} (x) (U) \leq \kappa^{\leq n+1} (x)
  (U)$ for every $x \in X$ and every open subset $U$ of
  $Norm$ by induction on $n \in \nat$.

  In the base case $n=0$, this is obvious. Otherwise, $n \geq 1$, and
  $\kappa^{\leq n+1} (x)
  (U) = \int_{y \in X} \kappa^{\leq n} (y) (U)
  d\kappa^{\leq 1} (x)$.   For every $y \in X$,
  $\kappa^{\leq n} (y) (U) \geq \kappa^{\leq n-1} (y) (U)$ by
  induction hypothesis, so $\kappa^{\leq n+1} (x)
  (U) \geq \int_{y \in X} \kappa^{\leq n-1} (y) (U) d\kappa^{\leq 1}
  (x) = \kappa^{\leq n} (x) (U)$.

  3. $\Val X$ is a dcpo, and suprema in the dcpo of continuous maps
  from $X$ to $\Val X$ are computed pointwise.
\end{proof}

From now on, we let $Norm$ be the subset of configurations
$(M, \theta)$ such that $M \theta$ is a normal form.  Being a normal
form is a property of shapes, namely any two generalized ISPCF terms
with the same shape will both be normal, or neither of them will
be. It follows that $Norm$ and its complement $\overline{Norm}$ are
open in $\Config$. Let us give another characterization of
$Next^{\leq n}$.

\begin{lemma}
  \label{lemma:Next*}
  Let $\Sigma$ consist of strictly observable first-order constants.  For every
  generalized ISPCF term $M$, for every $n \in \nat$, we have:
  \begin{enumerate}
  \item if $M \in Norm$, then $Next^{\leq n} (M) = \delta_M$;
  \item if $M = E^0 [\sample [\mu]]$, and $n \geq 1$, then for every
    $U \in \Open \Config$,
    $Next^{\leq n} (M) (U) = \int_{a \in \real} Next^{\leq n-1} (E^0
    [\retkw \underline a]) (U) d\mu$;
  \item if $M = E^0 [\score\; \underline{\mathbf a}]$ and $n \geq 1$,
    then
    $Next^{\leq n} (M) = |\mathbf a| . Next^{\leq n-1} (E^0 [\retkw
    \underline *])$;
  \item if
    $M = E [\underline f\; \underline {a_1} \cdots \underline {a_k}]$
    and $n \geq 1$, then $Next^{\leq n} (M)$ is equal to
    $Next^{\leq n-1}\allowbreak (E [\underline {f (a_1) \cdots
      (a_k)}])$ if $f (a_1) \cdots (a_k)$ is observable, to the zero
    valuation otherwise;
  \item for every instance of a rule $L \to R$ of
    Figure~\ref{fig:ideal:opsem} except (\ref{rule:sample}),
    (\ref{rule:score}) and (\ref{rule:f}), if $n \geq 1$ then
    $Next^{\leq n} (L) = Next^{\leq n-1} (R)$.
  \end{enumerate}
\end{lemma}
\begin{proof}
  1. By induction on $n$.  If $n=0$, then $Next^{\leq 0} (M) =
  \delta_M$ since $M \in Norm$.  If $n=1$, then $Next (M)=0$, so
  $Next^{\leq 1} (M) = Next (M) + Next^{\leq 0} (M) = \delta_M$. If
  $n \geq 2$, then
  $Next^{\leq n} (M) = (Next^{\leq n-1})^\dagger (Next^{\leq 1} (M)) =
  (Next^{\leq n-1})^\dagger (\delta_M) = Next^{\leq n-1} (M) =
  \delta_M$, by induction hypothesis.

  2. Since $M = E^0 [\sample [\mu]] \in \overline{Norm}$,
  $Next^{\leq 1} (M) = Next (M)$ is equal to $f_{E^0} [\mu]$, where
  $f_{E^0} (a) \eqdef E^0 [\retkw \underline a]$ for every
  $a \in \real$.  Therefore
  $Next^{\leq n} (M) (U) = (Next^{\leq n-1})^\dagger (f_{E^0} [\mu])
  (U) = \int_{N \in \Config} Next^{\leq n-1} (N) (U) df_{E^0} [\mu]$.
  By the change-of-variables formula, this is equal to
  $\int_{a \in \real} Next^{\leq n-1} (f_{E^0} (a)) \allowbreak (U)
  d\mu$, hence to
  $\int_{a \in \real} Next^{\leq n-1} (E^0 [\retkw \underline a]) (U)
  d\mu$.

  3. Since
  $M = E^0 [\score\; \underline{\mathbf a}] \in \overline{Norm}$,
  $Next^{\leq 1} (M) = Next (M) = |\mathbf a| . \allowbreak
  \delta_{E^0 [\retkw \underline *]}$.  Then, for every
  $U \in \Open \Config$,
  $Next^{\leq n} (M) (U) = (Next^{\leq n-1})^\dagger (|\mathbf a|
  . \allowbreak \delta_{E^0 [\retkw \underline *]}) (U) = \int_{N \in
    \Config} Next^{\leq n-1} (N) (U) d |\mathbf a| .  \allowbreak
  \delta_{E^0 [\retkw \underline *]}$, and this is equal to $|\mathbf a| . Next^{\leq n-1}
  \allowbreak (E^0 \allowbreak [\retkw \underline *]) (U)$.

  4. Let
  $M = E [\underline f\; \underline {a_1} \cdots \underline
  {a_k}]$. If $f (a_1) \cdots (a_k)$ is not observable, then $M$ is
  blocked, so $Next^{\leq 1} (M)=0$. For every $U \in \Open \Config$,
  $Next^{\leq n} (M) (U) = (Next^{\leq n-1})^\dagger (0) = \int_{N \in
    \Config} Next^{\leq n-1} (N) (U) d0 = 0$, so
  $Next^{\leq n} (M) = 0$.

  If $f (a_1) \cdots (a_k)$ is observable, then
  $Next^{\leq 1} (M) = Next (M) = \delta_{E [\underline {f (a_1)
      \cdots (a_k)}]}$. It follows that
  $Next^{\leq n} (M) = (Next^{\leq n-1})^\dagger (\eta_\Config (E
  [\underline {f (a_1) \cdots (a_k)}])) = Next^{\leq n-1} \allowbreak
  (E [\underline {f (a_1) \cdots (a_k)}])$.

  5. This is similar to the latter case.   
\end{proof}

We extend the denotational semantics of Figure~\ref{fig:sem} to
generalized ISPCF terms by positing:
\begin{align*}
  \Eval {\underline {\mathbf a}} \rho & \eqdef \mathbf a,
\end{align*}
for every $\mathbf a \in \IR$.  We note that $\Eval M \rho$ is
independent of $\rho$.  In the sequel, we will therefore simply write
$\Eval M$ for the semantics of the generalized ISPCF term $M$,
disregarding the useless $\rho$.  Alternatively, every generalized
ISPCF term $M$ can be written as $P \theta$ for some unique configuration
$(P, \theta)$, and $\Eval M$ is also equal to $\Eval P \theta$,
reading $\theta$ as an environment.  The Scott-continuity of $\Eval P$
immediately entails the following.
\begin{lemma}
  \label{lemma:Eval:cont}
  Let $\Config_\tau$ be the subspace of $\Config$ consisting of
  generalized ISPCF terms of type $\tau$.
  The map $\Eval \_ \colon \Config_\tau \to \Eval \tau$ is
  Scott-continuous.
\end{lemma}

It follows that, given any Scott-open subset $U$ of $\Eval \tau$, for
any type $\tau$, the set $\overline U$ of normal forms $M$ of type
$\tau$ such that $\Eval M \in U$ is Scott-open in $\Config_\tau$.
We make that into a definition, and we also introduce another open set
$\rover U$.
\begin{definition}[The open sets $\overline U$, $\rover U$]
  \label{defn:rover}
  For every type $\tau$, for every Scott-open subset $U$ of
  $\Eval\tau$, let $\overline U$ be the Scott-open set of normal forms
  $V$ of type $\tau$ such that $\Eval V \in U$, and $\rover U$ be the
  Scott-open set of terms of the form $\retkw V$ with
  $V \in \overline U$.
\end{definition}
It is clear that $\rover U$ is Scott-open: by definition of the
topology of $\Config$, the map $M \mapsto \retkw M$ is a homeomorphism
of $\Config$ onto the subspace of those configurations starting with
$\retkw$.

\begin{lemma}
  \label{lemma:E0:D}
  The only weak evaluation contexts $E^0$ whose type is of the
  form $D \sigma \vdash \tau$ are those of the form:
  \[
    \dokw x_n \leftarrow (\dokw x_{n-1} \leftarrow \cdots (\dokw x_1
    \leftarrow [];\allowbreak N_1); \cdots ; N_{n-1}); N_n,
  \]
  and then $\tau$ is a distribution type $D \tau'$.
\end{lemma}
\begin{proof}
  By induction on the size of $E^0$.  This is obvious if $E^0 = []$,
  otherwise we see that $E^0$ must be of the form
  ${E'}^0 [\dokw {x \leftarrow []}; N]$, by inspection of the
  evaluation context formation rules, using typing to rule out all the
  other possibilities.
\end{proof}

The following linearity property is crucial for soundness.  The
importance of such linearity properties have already been made in
\cite[Proposition~5.1]{JGL-lics19}, where probabilistic choice was
discrete, and in \cite[Lemma~7.6]{EPT:PPCF}.  Given any evaluation
context $E$, we write $\Eval E$ for $\Eval {\lambda x . E [x]}$.
\begin{lemma}
  \label{lemma:E0:lin}
  Assume that all constants in $\Sigma$ are first-order constants.
  For every weak evaluation context $E^0$ of type
  $D \sigma \vdash D \tau$, $\Eval {E^0}$ is \emph{linear}:
  \begin{enumerate}
  \item for all $\mu, \mu' \in \Val {\Eval \sigma}$, for every
    $a \in \realp$, $\Eval {E^0} (a . \mu) = a \Eval {E^0} (\mu)$ and
    $\Eval {E^0} (\mu+\mu') = \Eval {E^0} (\mu) + \Eval {E^0} (\mu')$;
  \item for every dcpo $X$, for every Scott-continuous map
    $f \colon X \to \Val {\Eval \sigma}$,
    $\Eval {E^0} \circ f^\dagger = (\Eval {E^0} \circ f)^\dagger$.
  \end{enumerate}
\end{lemma}
\begin{proof}
  By Lemma~\ref{lemma:E0:D}, $E^0$ is of the form
  $\dokw x_n \leftarrow (\dokw x_{n-1} \leftarrow \cdots (\dokw x_1
  \leftarrow [];\allowbreak N_1); \cdots ; N_{n-1}); N_n$.  We prove
  both claims by induction on $n$, the \emph{length} of $E^0$.

  1. If $n=0$, this is obvious.  Otherwise, we can write $E^0$ as
  $\dokw {x_n \leftarrow {E'}^0}; N_n$ where ${E'}^0$ has length
  $n-1$.  We note that, for every continuous map
  $f \colon X \to \Val Y$, $f^\dagger$ is linear, in the sense that
  $f^\dagger (a . \mu) = a . f^\dagger (\mu)$ and
  $f^\dagger (\mu+\mu') = f^\dagger (\mu) + f^\dagger (\mu')$.  This
  is clear from the definition of $f^\dagger$, see formula
  (\ref{eq:dagger:def}).  Then:
  \begin{align*}
    \Eval {E^0} (a . \mu)
    & = (\Eval {\lambda x_n . N_n})^\dagger (\Eval {{E'}^0} (a . \mu)) \\
    & = (\Eval {\lambda x_n . N_n})^\dagger (a . \Eval {{E'}^0} (\mu))
    & \text{by induction hypothesis} \\
    & = a . (\Eval {\lambda x_n . N_n})^\dagger (\Eval {{E'}^0} (\mu))
      = a . \Eval {E^0} (\mu)\mskip-40mu
  \end{align*}
  and similarly for $\Eval E (\mu+\mu')$.

  2. Again, this is clear when $n=0$.  When $n \geq 1$, we have:
  \begin{align*}
    \Eval {E^0} \circ f^\dagger
    & = (\Eval {\lambda x_n . N_n})^\dagger \circ \Eval {{E'}^0} \circ
      f^\dagger \\
    & = (\Eval {\lambda x_n . N_n})^\dagger \circ \left(\Eval {{E'}^0} \circ
      f\right)^\dagger
    & \text{by induction hypothesis} \\
    & = \left((\Eval {\lambda x_n . N_n})^\dagger \circ \Eval {{E'}^0} \circ
      f\right)^\dagger
    & \text{by the equation }(g^\dagger \circ f)^\dagger = g^\dagger \circ f^\dagger\\
    & = (\Eval {E^0} \circ f)^\dagger.
  \end{align*}

\end{proof}

\begin{proposition}[Soundness]
  \label{prop:sound}
  Let $\Sigma$ consist of strictly observable first-order constants.  For every
  generalized ISPCF term $M \colon D \tau$, where $\tau$ is any type,
  for every open subset $U$ of $\Eval \tau$,
  $\Eval M (U) \geq Next^* (M) (\rover U)$.
\end{proposition}
\begin{proof}
  We show that $\Eval M \rho (U) \geq Next^{\leq n} (M) (\rover U)$
  for every $n \in \nat$. If $M \in Norm$, then
  $Next^{\leq n} (M) = \delta_M$ by Lemma~\ref{lemma:Next*}, item~1.
  By Lemma~\ref{lemma:normal} (and typing), $M = \retkw V$ for some
  value $V$ of type $\tau$.  Hence $\Eval M (U) = \delta_{\Eval V} (U)$,
  which is equal to $1$ if $\Eval V \in U$, and to $0$ otherwise.
  Since $\Eval V \in U$ if and only if $V \in \overline U$ if and only
  if $\retkw V \in \rover U$, $\Eval M (U) = Next^{\leq n} (M) (\rover U)$.

  Henceforth, we assume that $M$ is not normal, and we prove the claim
  by induction on $n$.
  If $n=0$, then $Next^{\leq 0} (M) = 0$, so the claim is clear. Let
  therefore $n \geq 1$, and let us look at the possible shapes of $M$.

  IF $M = E^0 [\sample [\mu]]$, then:
  \begin{align*}
    Next^{\leq n} (M) (\rover U)
    & = \int_{a \in \real} Next^{\leq n-1} (E^0 [\retkw \underline a]) (\rover
      U) d\mu
    & \text{by Lemma~\ref{lemma:Next*}, item~2} \\
    & \leq \int_{a \in \real} \Eval {E^0 [\retkw \underline a]} (U) d\mu
    & \text{by induction hypothesis} \\
    & = \int_{\mathbf a \in \IR_{\bot}} \Eval {E^0} (\eta_{\Eval {\realT}}
      (\mathbf a)) (U) di [\mu] \\
    & \qquad\qquad\qquad\text{by the change-of-variables formula}\mskip-150mu \\
    & = (\Eval {E^0} \circ \eta_{\Eval {\realT}})^\dagger (i [\mu]) (U)
    & \text{by definition of }^\dagger\\
    & = \Eval {E^0} (\eta_{\Eval {\realT}}^\dagger (i [\mu])) (U)
    & \text{by Lemma~\ref{lemma:E0:lin}, item~2} \\
    & = \Eval {E^0} (i [\mu]) (U) & \text{since $\eta_X^\dagger =
                                \identity {\Val X}$} \\
    & = \Eval {{E^0} [\sample [\mu]]} (U) = \Eval
      M (U).
  \end{align*}

  If $M = E^0 [\score\; \underline {\mathbf a}]$, then:
  \begin{align*}
    Next^{\leq n} (M) (\rover U)
    & = |\mathbf a| . Next^{\leq n-1} (E^0 [\retkw \underline *]) (\rover
      U) d\mu
    & \text{by Lemma~\ref{lemma:Next*}, item~3} \\
    & \leq |\mathbf a| .  \Eval {E^0 [\retkw \underline *]} (U)
    & \text{by induction hypothesis} \\
    & = |\mathbf a| . \Eval {E^0} (\delta_*) (U) \\
    & = \Eval {E^0} (|\mathbf a| . \delta_*) (U)
    & \text{by Lemma~\ref{lemma:E0:lin}, item~1} \\
    & = \Eval {E^0 [\score\; {\underline {\mathbf a}}]} (U) = \Eval
      M (U).
  \end{align*}
  If
  $M = E [\underline f\; \underline {a_1} \cdots \underline {a_k}]$,
  then we consider two cases.  If $f (a_1) \cdots (a_k)$ is not
  observable, then $Next^{\leq n} (M) (\rover U) = 0$ by
  Lemma~\ref{lemma:Next*}, item~4, and the claim is clear.  If
  $f (a_1) \cdots (a_k)$ is observable, then:
  \begin{align*}
    Next^{\leq n} (M) (\rover U)
    & = Next^{\leq n-1} (E [\underline {f (a_1) \cdots (a_k)}]) (\rover U)
    & \text{by Lemma~\ref{lemma:Next*}, item~4} \\
    & \leq \Eval {E [\underline {f (a_1) \cdots (a_k)}]} (U)
    & \text{by induction hypothesis} \\
    & = \Eval {E [\underline f\; \underline {a_1} \cdots \underline
      {a_k}]} (U) = \Eval M (U).
  \end{align*}
  
  Finally, if $M$ an instance $L$ of a rule $L \to R$ of
  Figure~\ref{fig:ideal:opsem} except (\ref{rule:sample}),
  (\ref{rule:score}) and (\ref{rule:f}), then by inspection
  $\Eval L = \Eval R$, so:
  \begin{align*}
    Next^{\leq n} (M) (\rover U)
    & = Next^{\leq n-1} (R) (\rover U)
    & \text{by Lemma~\ref{lemma:Next*}, item~5} \\
    & \leq \Eval {R} (U)
    & \text{by induction hypothesis} \\
    & = \Eval L (U)  = \Eval M (U).
  \end{align*}
\end{proof}


\subsection{Adequacy}
\label{sec:adequacy}

Adequacy means that the inequality of Proposition~\ref{prop:sound} can
be reinforced to an equality, provided that $\tau$ is an observable
  type.

We recall that the \emph{deterministic rules} are all those of
Figure~\ref{fig:ideal:opsem} except the first three, and that the
\emph{weak rules} are those where the evaluation context is weak.
\begin{lemma}
  \label{lemma:det}
  Let $\Sigma$ consist of strictly observable first-order constants.  For every
  string of rewrite steps
  $M = M_0 \to M_1 \to \cdots \to M_n \to \cdots$ by deterministic
  rules, $Next^* (M) = Next^* (M_i)$ for every $i$.
\end{lemma}
\begin{proof}
  It suffices to show that if $M \to N$ by some deterministic rule,
  then $Next^* (M) = Next^* (N)$.  This follows from the fact that,
  for every $n \geq 1$, $Next^{\leq n} (M) = Next^{\leq n-1} (N)$
  (Lemma~\ref{lemma:Next*}, items~4 and~5), and by taking suprema over
  $n$.
\end{proof}

We say that $M \to^* N$ \emph{by deterministic rules} (resp., \emph{by
  weak deterministic rules}) if and only if one can connect $M$ to $N$
by a finite string of instances of deterministic rules (resp., weak
deterministic rules).
\begin{lemma}
  \label{lemma:det:value}
  Let $\Sigma$ consist of strictly observable first-order constants.
  \begin{enumerate}
  \item For every generalized ISPCF term $M$ of type $\tau$, there is
    at most one maximal string of instances of weak deterministic
    rules $M_0 \eqdef M \to M_1 \to \cdots \to \cdots M_n \cdots$.
    Every $M_i$ has type $\tau$.  If $\tau$ is an observable type, and
    if that string stops at rank $n$, then $M_n$ is a weak value
    $V^0$.
  \item For every generalized ISPCF term $M$ of type $\tau$, there is
    at most one maximal string of instances of deterministic rules
    $M_0 \eqdef M \to M_1 \to \cdots \to \cdots M_n \cdots$.  Every
    $M_i$ has type $\tau$.  If $\tau$ is an observable type, either
    that string stops at rank $n$, then $M_n$ is a value $V$, and
    $Next^* (M) = \delta_V$; or else $Next^* (M)=0$.
  \end{enumerate}
\end{lemma}
\begin{proof}
  The unicity of any maximal string of (weak) deterministic rules
  starting from $M$ is clear, as well as the fact that every $M_i$ has
  type $\tau$.  We now assume that $\tau$ is an observable type.

  1. We consider a maximal string of weak deterministic rules stopping
  at rank $n$, namely $M=M_0 \to^* M_n$.  We claim that $M_n$ is
  weakly normal.  If $M_n$ is of the form $E^0 [\sample [\mu]]$ or
  $E^0 [\score\;\underline{\mathbf a}]$ or
  $E^0 [\dokw {x \leftarrow \retkw M}; N]$, then $E^0$ has a type of
  the form $D \sigma \vdash \tau$.  By Lemma~\ref{lemma:E0:D}, that
  would imply that $\tau$ is of the form $D \tau'$ for some type
  $\tau'$, which is impossible since $\tau$ is observable.  Since no
  weak deterministic rule applies to $M_n$, no weak rule at all applies, so
  $M_n$ is weakly normal.  By Lemma~\ref{lemma:normal}, it is a weak
  value $V^0$.

  2. We reason similarly, assuming a maximal string of deterministic
  rules stopping at rank $n$, $M=M_0 \to^* M_n$.  Then $M_n$ is a
  value $V$.  By Lemma~\ref{lemma:det},
  $Next^* (M) = Next^* (M_n) = Next^* (V)$, and this is equal to
  $\delta_V$ by Lemma~\ref{lemma:Next*}, item~1.

  If instead the maximal string $M=M_0 \to M_1 \to \cdots$ does not
  stop, then for every $m \in \nat$, we have
  $Next^{\leq m} (M) = Next^{\leq m-1} (M_1) = \cdots = Next^{\leq 0}
  (M_m)$, by Lemma~\ref{lemma:Next*}, items~4 and~5.  Since $M_m$ is
  not normal, this is equal to $0$.  Taking suprema over $m \in \nat$,
  we obtain that $Next^* (M) = 0$.
\end{proof}

In order to establish adequacy, we will use a logical relation
${(\R_\tau)}_{\tau\text{ type}}$ in the same style as those used in
\cite{jgl-jlap14,JGL-lics19}.  Before we do this, we need the
following notion.
\begin{definition}[Observable open set]
  \label{defn:partial:open}
  An \emph{observable open subset} of $\Eval \beta$ is a Scott-open
  subset of $\Eval \beta$ consisting of observable elements.
\end{definition}

For each type $\tau$, $\R_\tau$ will be
a binary relation between generalized ISPCF terms of type $\tau$ and
elements of $\Eval \tau$.  We define it with the help of auxiliary
relations $\R_{\tau \vdash D \beta}^\perp$.  The latter relates
evaluation contexts of type $D \tau \vdash D \beta$ and
Scott-continuous maps from $\Eval \tau$ to $\Eval {D \beta}$, where
$\beta$ is an observable type.  We will use short phrases such as
``for all $E^0 \R_{\tau \vdash D \beta}^\perp h$'' instead of ``for
every observable type $\beta$, for every weak evaluation context $E^0$
of type $D \tau \vdash D \beta$, for every Scott-continuous map $h$
from $\Eval \tau$ to $\Eval {D \beta}$, if
$E^0 \R_{\tau \vdash D \beta}^\perp h$ then''.
\begin{definition}[Logical relation, $\R_\tau$, $\R_{\tau \vdash
    D\beta}^\perp$]
  \label{defn:R}
  We define:
  \begin{itemize}
  \item for every \emph{basic} type $\tau$, $M \R_\tau a$ if and only
    if either $a=\bot$, or $M \to^* \underline b$ by weak
    deterministic rules, for some zero-ary constant $\underline b$
    such that $b \geq a$.
  \item $M \R_{\tau_1 \times \tau_2} a$ if and only if $a=\bot$, or
    $a = (a_1, a_2)$ where
    $M \to^* \langle M_1, M_2 \rangle$ by weak deterministic rules,
    for some $M_1 \R_{\tau_1} a_1$ and $M_2 \R_{\tau_2} a_2$.
  \item $M \R_{\tau_1 + \tau_2} a$ if and only if $a=\bot$, or
    $M \to^* \iota_i N$ by weak deterministic rules and $a = (i, b)$,
    for some $i \in \{1, 2\}$ and some $N \R_{\tau_i} b$ with
    $b \neq \bot$.
  \item $M \R_{\sigma \to \tau} g$ if and only if for all
    $N \R_\sigma a$, $MN \R_\tau g (a)$.
  \item $M \R_{D \tau} \mu$ if and only if for all
    $E^0 \R_{\tau \vdash D \beta}^\perp h$, for every observable open
    subset $U$ of $\Eval \beta$,
    $Next^* (E^0 [M]) (\rover U) \geq h^\dagger (\mu) (U)$;
  \item $E^0 \R_{\tau \vdash D \beta}^\perp h$ if and only if for all
    $M \R_\tau a$, for every observable open subset $U$ of $\Eval \beta$,
    $Next^* (E^0 [\retkw M]) (\rover U) \geq h (a) (U)$.
  \end{itemize}
\end{definition}
This is a well-founded definition by induction on the size of $\tau$.
Note that $\R_{D \tau}$ is defined in terms of
$\R_{\tau \vdash D\beta}^\perp$, which itself depends on $\R_\tau$;
there is no dependency on, say, $\R_{D \beta}$, which would make the
definition ill-founded.

\begin{lemma}
  \label{lemma:R:red}
  Let $\Sigma$ consist of strictly observable first-order constants.  Let $\tau$ be
  any type.  If $M \to^* M'$ by weak deterministic rules, then
  $M \R_\tau a$ if and only if $M' \R_\tau a$.
\end{lemma}
\begin{proof}
  It suffices to show this when $M \to M'$ by one weak deterministic
  rule, and we show this by induction on $\tau$.

  When $\tau$ is a basic type, and if $M' \R_\tau a$, then either
  $a=\bot$, or $M' \to^* \underline b$ by weak deterministic rules,
  with $b \geq a$.  In the second case, $M \to M' \to^* b$ by weak
  deterministic rules, so $M \R_\tau a$.  In the reverse direction, if
  $M \R_\tau a$, either $a=\bot$, or $M \to^* \underline b$ by
  weak deterministic rules, with $b \geq a$.  Since $M \to M'$ by a weak
  deterministic rule, since $\underline b$ is (weakly) normal, and
  since weak deterministic rules are deterministic by construction,
  the string of reduction steps $M \to^* \underline b$ is of the form
  $M \to M' \to^* \underline b$.  Therefore $M \R_\tau a$.

  The argument is similar 
  when $\tau$ is a product type $\tau_1 \times \tau_2$.
  The key is that, if $M \to^* \langle M_1, M_2 \rangle$ by weak
  deterministic rules, and if $M \to M'$ by a weak deterministic rule,
  then $M$ is not weakly normal, so the first string of rewrite steps
  must be of the form $M \to M' \to^* \langle M_1, M_2 \rangle$.
  Hence, if $M \to M'$ by a weak deterministic rule, then
  $M \to^* \langle M_1, M_2 \rangle$ by weak deterministic rules if
  and only if $M' \to^* \langle M_1, M_2 \rangle$ by weak
  deterministic rules.  We reason similarly if $\tau$ is a sum type
  $\tau_1+\tau_2$, where if $M \to M'$ by a weak deterministic rule,
  then $M \to^* \iota_i N$ by weak deterministic rules if and only if
  $M' \to^* \iota_i N$ by weak deterministic rules.

  For arrow types $\sigma \to \tau$, we assume $M$ and $M'$ of type
  $\sigma \to \tau$, such that $M \to M'$ by a weak deterministic
  rule.  Then, for every $N \colon \sigma$, $MN \to M'N$ by the same
  weak deterministic rule.  If $M \R_{\sigma \to \tau} f$, then for
  all $N \R_\tau a$, $MN \R_\tau f (a)$, and since $MN \to M'N$, by
  induction hypothesis $M'N \R_\tau f (a)$; therefore
  $M' \R_{\sigma \to \tau} f$.  The fact that
  $M' \R_{\sigma \to \tau} f$ implies $M \R_{\sigma \to \tau} f$ is
  similar.

  For distribution types $\tau \eqdef D \sigma$, we use a variant of
  that argument.  We assume $M, M' \colon D \sigma$ such that
  $M \to M'$ by some weak deterministic rule.  Then, for every weak
  evaluation context $E^0$, $E^0 [M] \to^* E^0 [M']$ by the same weak
  deterministic rule.  By Lemma~\ref{lemma:det},
  $Next^* (E^0 [M]) = Next^* (E^0 [M'])$, from which it follows
  immediately that $M \R_{D \sigma} \mu$ if and only if $M' \R_{D
    \sigma} \mu$, for any $\mu$.
\end{proof}

We let $M \R_\tau$ abbreviate the set of all values $a \in \Eval \tau$
such that $M \R_\tau a$.
\begin{lemma}
  \label{lemma:MR:closed}
  Let $\Sigma$ consist of strictly observable first-order constants.  For every
  type $\tau$, for every generalized ISPCF term $M$ of type $\tau$,
  $M \R_\tau$ is a Scott-closed subset of $\Eval \tau$ that contains
  $\bot$.
\end{lemma}
\begin{proof}
  By induction on $\tau$.  That $M \R_\tau$ contains $\bot$ is clear
  except perhaps when $\tau$ is a distribution type $D \tau'$.  In
  that case, the $\bot$ element of $\Eval \tau$ is the zero valuation
  $0$, and $h^\dagger (0) (U)$ is equal to $0$ for every
  $h \in \Eval {\tau \vdash D \beta}$, which trivially entails the
  statement.

  If $\tau$ is a basic type, then $M \R_\tau$ is the closed set
  $\{\bot\}$ if the maximal string of weak deterministic rewriting
  steps starting from $M$ does not stop, or $\dc b$ if
  $M \to^* \underline b$ by weak deterministic rules.  Then, $\dc b$
  is closed because downward closures of points are closed.

  In the case of product types, either $M \to^* V^0$ by weak
  deterministic rules for no lazy value $V^0$, and then
  $M \R_{\tau_1 \times \tau_2}$ is equal to $\{\bot\}$; or
  $M \to^* V^0$ for some uniquely determined lazy value $V^0$ by
  Lemma~\ref{lemma:det:value}, item~1.  By typing and
  Lemma~\ref{lemma:normal}, $V^0 = \langle M_1, M_2 \rangle$ for some
  $M_1 \colon \tau_1$ and $M_2 \colon \tau_2$.  Then
  $M \R_{\tau_1 \times \tau_2}$ is equal to
  $(M_1 \R_\sigma) \times (M_2 \R_\sigma)$ union $\{\bot\}$.  The
  latter is a finite union of Scott-closed sets, hence is
  Scott-closed.

  In the case of sum types, we reason similarly.  Either
  $M \R_{\tau_1 + \tau_2}$ is equal to $\{\bot\}$, or
  $M \to^* \iota_i N$ by weak deterministic rules, and then
  $M \R_{\tau_1 + \tau_2}$ is the set of elements
  $\{\bot\} \cup \{(i, b) \mid b \in N \R_{\tau_i}\}$.  It is easy to
  see that the latter is downwards closed and closed under directed
  suprema, hence is Scott-closed.

  In the case of arrow types $\sigma \to \tau$,
  $M \R_{\sigma \to \tau}$ is equal to the intersection over all
  $N \R_\sigma a$ of the sets $\App (\_, a)^{-1} (MN \R_\tau)$.
  ($\App$ is the application morphism $(g, a) \mapsto g (a)$.)  This
  set is closed because $\App$ and therefore $\App (\_, a)$ is
  continuous for every $a$, and intersections of closed sets are closed.

  In the case of distribution types $D \tau$, $M \R_{D \tau}$ is the
  intersection of the sets
  ${(h^\dagger (\_) (U))}^{-1} (\dc Next^* (E^0 [M]) (\rover U))$ over all
  $E^0 \R_{\tau \vdash D \beta}^\perp h$ and all observable open subsets $U$
  of $\Eval \beta$.  We conclude because downward closures of points
  are closed, and since $h^\dagger$ is Scott-continuous, hence also
  $h^\dagger (\_) (U)$ for every $U$.
\end{proof}


The fundamental lemma of logical relations states that terms are
related to their semantics.  We prove the variant suited to our case
below.  Given an ISPCF term $M$, a \emph{substitution} $\theta$ for
$M$ is a map from the free variables of $M$ to generalized ISPCF terms
of the same type, and $M \theta$ denotes the result of application
$\theta$ to $M$.  Given a substitution $\theta$ and an environment
$\rho$ with the same domain $D$, we write $\theta \R_* \rho$ to mean
that $\theta (x_\sigma) \R_\sigma \rho (x_\sigma)$ for every
$x_\sigma \in D$.

\begin{lemma}
  \label{lemma:R:obs:val}
  Let $\Sigma$ consist of strictly observable first-order constants.  For every
  observable type $\beta$, for every observable element $a$ of
  $\Eval \beta$, $\underline a \R_\beta a$.
\end{lemma}
\begin{proof}
  We recall that $\underline a$ was defined in
  Definition~\ref{defn:underline}.  We show this by induction on
  $\beta$.  If $\beta$ is a basic type, this follows from the fact
  that $\underline a \to^* \underline a$ (in $0$ step) and
  $\Eval {\underline a} = a$.
  If $\beta$ is a product type $\beta_1 \times \beta_2$,
  then $a$ is of the form $(a_1, a_2)$ with $a_1$ observable in $\Eval
  {\beta_1}$ and $a_2$ observable in $\Eval {\beta_2}$.  Then
  $\underline a \to^* \langle \underline a_1, \underline a_2 \rangle$
  (in $0$ step), while $\underline a_1 \R_{\beta_1} a_1$ and
  $\underline a_2 \R_{\beta_2} a_2$ by induction hypothesis.  The case
  of sum types is similar.
\end{proof}
  
\begin{proposition}
  \label{prop:basic:lemma}
  Let $\Sigma$ consist of strictly observable first-order constants.
  For every ISPCF term $M \colon \tau$, for all $\theta \R_* \rho$, $M
  \theta \R_\tau \Eval M \rho$.
\end{proposition}
\begin{proof}
  By induction on the size of $M$.

  $\bullet$ If $M$ is a constant $\underline f$, say of type
  $\sigma_1 \to \cdots \to \sigma_k\to \tau$ where every $\sigma_i$ is
  a basic type and $\tau$ is an observable type, then we must show
  that for all $N_1 \R_{\sigma_1} a_1$, \ldots,
  $N_k \R_{\sigma_k} a_k$,
  $\underline f\; N_1 \cdots N_k \R_\tau f (a_1) \cdots (a_k)$.  If
  $f (a_1) \cdots (a_k) = \bot$, this is clear, so let us assume
  $f (a_1) \cdots (a_k) \neq \bot$.  By the second condition of
  Definition~\ref{defn:assum:strict}, $f (a_1) \cdots (a_k)$ is
  observable, and every $a_i$ is different from $\bot$.  Since every
  $\sigma_i$ is a basic type, we have $N_i \to^* \underline {a'}_i$ by
  weak deterministic rules for some zero-ary constant
  $\underline {a'}_i$ such that $a'_i \geq a_i$, for every $i$ with
  $1\leq i \leq k$.  It is easy to see that
  $\underline f\; N_1 \cdots N_k \to^* \underline f\;\underline{a'}_1
  \; \cdots \; \underline{a'}_k \to \underline {f (a'_1) \cdots
    (a'_k)}$ by weak deterministic rules; the last step is justified
  by the fact that $f (a'_1) \cdots (a'_k) \geq f (a_1) \cdots (a_k)$,
  that $f (a_1) \cdots (a_k)$ is observable, and that every element
  larger than an observable element is itself observable.  By
  Lemma~\ref{lemma:R:obs:val},
  $\underline {f (a'_1) \cdots (a'_k)} \R_\tau \Eval {\underline {f
      (a'_1) \cdots (a'_k)}} \rho = f (a'_1) \cdots (a'_k)$.  Using
  Lemma~\ref{lemma:R:red}, we deduce that
  $\underline f\; N_1 \cdots N_k \R_\tau f (a'_1) \cdots (a'_k)$.
  Since $f (a'_1) \cdots (a'_k) \geq f (a_1) \cdots (a_k)$, and since
  $\underline f\; N_1 \cdots N_k \R_\tau$ is Scott-closed, hence
  downwards-closed (Lemma~\ref{lemma:MR:closed}), we obtain that
  $\underline f\; N_1 \cdots N_k \R_\tau f (a_1) \cdots (a_k)$, as
  desired.

  $\bullet$ If $M$ is a variable, this is by the assumption
  $\theta \R_* \rho$.

  $\bullet$ If $M = \sample [\mu]$, we must show that for all
  $E^0 \R_{\realT \to D\beta}^\perp h$, for every observable open
  subset $U$ of $\Eval \beta$,
  $Next^* (E^0 [\sample [\mu]]) (\rover U) \geq h^\dagger (\mu) (U)$.
  By Lemma~\ref{lemma:Next*}, item~2, and taking suprema over
  $n \geq 1$, $Next^* (E^0 [\sample [\mu]]) (\rover U)$ is equal to
  $\int_{a \in \real} Next^* (E^0 [\retkw \underline a]) (\rover U)
  d\mu$.  By Lemma~\ref{lemma:R:obs:val},
  $\underline {\mathbf a} \R_{\realT} \mathbf a$ for every
  $\mathbf a \in \IR$.  In particular,
  $\underline a \R_{\realT} i (a)$ for every $a \in \real$.  Since
  $E^0 \R_{\realT \to D\beta}^\perp h$,
  $Next^* (E^0 [\retkw \underline a]) (\rover U) \geq h (a) (U)$ for
  every $a \in \real$.  From this, we conclude that
  $Next^* (E^0 [\sample [\mu]]) (\rover U) \geq \int_{a \in \real} h
  (a) (U) d\mu = h^\dagger (\mu) (U)$.

  $\bullet$ If $M = \score \; N$, we must show that for all
  $E^0 \R_{\unitT \to D\beta}^\perp h$, for every observable open subset
  $U$ of $\Eval \beta$,
  $Next^* (E^0 [M \theta]) (\rover U) \geq h^\dagger (|\mathbf a|
  . \delta_*) (U)$ where $\mathbf a \eqdef \Eval N \rho$.  We observe
  that $h^\dagger (|\mathbf a| . \delta_*) = |\mathbf a| . h (*)$.
  Hence the claim is clear if $|\mathbf a|=0$, notably if
  $\mathbf a = \bot$.  Let us assume that $\mathbf a \neq \bot$.  Then
  $N \theta \R_{\realT} \mathbf a$, by induction hypothesis.  By
  definition of $\R_{\realT}$,
  $N \theta \to^* \underline {\mathbf {a'}}$ by weak deterministic
  rules, for some zero-ary constant $\underline{\mathbf {a'}}$ such
  that $\mathbf {a'} \geq \mathbf a$.

  Since $E^0 \R_{\unitT \to D\beta}^\perp h$, and since
  $\underline * \R_{\unitT} *$ (by the case of zero-ary constants,
  already treated),
  $Next^* (E^0 [\retkw {\underline *}]) (\rover U) \geq h (*) (U)$.
  We have $N \theta \to^* \underline {\mathbf {a'}}$ by weak
  deterministic rules, so
  $E^0 [M \theta] = E^0 [\score \; N\theta] \to^* E^0 [\score \;
  \underline {\mathbf {a'}}]$ by weak deterministic rules.  By
  Lemma~\ref{lemma:det},
  $Next^* (E^0 [M \theta]) (\rover U) = Next^* \allowbreak (E^0
  [\score \; \underline {\mathbf {a'}}]) (\rover U)$, which is equal
  to $|\mathbf {a'}| . Next^* (E^0 [\retkw \underline *]) (\rover U)$
  by Lemma~\ref{lemma:Next*}, item~3, hence is larger than or equal to
  $|\mathbf {a'}| .  h (*) (U)$.  We verify that
  $|\mathbf {a'}| \geq |\mathbf a|$, owing to the fact that
  $\mathbf {a'} \geq \mathbf a$.  Hence
  $Next^* (E^0 [M \theta]) (\rover U) \geq |\mathbf a| .  h (*) (U) =
  h^\dagger (|\mathbf a| . \delta_*) (U)$.

  $\bullet$ If $M = \lambda x . N \colon \sigma \to \tau$, then we
  must show that for all $P \R_\sigma a$,
  $(M\theta) P \R_\tau \Eval M \rho (a)$.  By induction hypothesis, we
  have $N \theta [x:=P] \R_\tau \Eval N \rho [x \mapsto a]$.  Since
  $(M\theta) P \to N \theta [x:=P]$, we conclude by
  Lemma~\ref{lemma:R:red}.

  $\bullet$ The case where $M$ is an application is immediate.

  $\bullet$ If $M = \reckw\; N$, where $N \colon \sigma \to \sigma$,
  then let $f \eqdef \Eval N \rho$.  By induction hypothesis,
  $N \R_{\sigma \to \sigma} f$, so: $(*)$ for all $P \R_\sigma a$, we
  have $(N \theta) P \R_\sigma f (a)$.  We show that
  $M \theta \R_\sigma f^n (\bot)$ for every $n \in \nat$, by induction
  on $n$.  If $n=0$, then $f^0 (\bot) = \bot$ is in
  $M \theta \R_\sigma$ by Lemma~\ref{lemma:MR:closed}.  If $n \geq 1$,
  then $N \theta (M \theta) \R_\sigma f (f^{n-1} (\bot))$, using $(*)$
  and the induction hypothesis.  Now
  $M \theta = \reckw (N \theta) \to N \theta (M \theta)$ by a
  weak deterministic rule, so $M \theta \R_\sigma f^n (\bot)$, by
  Lemma~\ref{lemma:R:red}.  Now that we have shown that
  $M \theta \R_\sigma f^n (\bot)$, namely that $f^n (\bot)$ is in
  $M \theta \R_\sigma$ for every $n \in \nat$, we use the fact that
  $M \theta \R_\sigma$ is Scott-closed (Lemma~\ref{lemma:MR:closed})
  and we conclude that
  $\Eval M \rho = \lfp (f) = \sup_{n \in \nat} f^n (\bot)$ is in
  $M \theta \R_\sigma$.

  $\bullet$ If $M = \retkw N$, with $N \colon \sigma$, then we must
  show that for all $E^0 \R_{\sigma \vdash D \beta}^\perp h$, for
  every observable open subset $U$ of $\Eval \beta$,
  $Next^* (E^0 [M \theta]) (\rover U) \geq h^\dagger (\delta_{\Eval N
    \rho}) (U)$.  By induction hypothesis,
  $N \theta \R_\sigma \Eval N \rho$, so by definition of
  $\R_{\sigma \vdash D \beta}^\perp$, we obtain
  $Next^* (E^0 [\retkw N \theta]) (\rover U) \geq h (\Eval N \rho)
  (U)$; this is exactly what we want to prove.

  $\bullet$ If $M = \dokw {x_\sigma \to N}; P \colon D \tau$, then we
  must show that for all $E^0 \R_{\tau \vdash D \beta}^\perp h$, for
  every observable open subset $U$ of $\Eval \beta$,
  $Next^* (E^0 [M \theta]) (\rover U) \geq h^\dagger (\Eval M \rho)
  (U)$.  We note that
  $h^\dagger (\Eval M \rho) (U) = h^\dagger ((\Eval {\lambda x_\sigma
    . P} \rho)^\dagger \allowbreak (\Eval N \rho)) (U) = (h^\dagger
  \circ \Eval {\lambda x_\sigma . P} \rho)^\dagger \allowbreak (\Eval
  N \rho)) (U)$, where the last equality is by the monad equation
  $g^\dagger \circ f^\dagger = (g^\dagger \circ f)^\dagger$.

  We consider the weak evaluation context
  ${E'}^0 \eqdef E^0 [\dokw {x_\sigma \to []}; P\theta]$, and we claim
  that
  ${E'}^0 \R_{\sigma \vdash D \beta}^\perp h^\dagger \circ \Eval
  {\lambda x_\sigma . P} \rho$.  This means showing that for all
  $Q \R_\sigma b$, for every observable open subset $V$ of
  $\Eval {D \beta}$,
  $Next^* ({E'}^0 [\retkw Q]) (\rover V) \geq h^\dagger (\Eval
  {\lambda x_\sigma . P} \rho (b)) (V)$.  The right-hand side is equal
  to $h^\dagger (\Eval P \rho [x \mapsto b]) (V)$.  Since
  ${E'}^0 [\retkw Q] = E^0 [\dokw {x_\sigma \to \retkw Q}; P\theta]
  \to E^0 [P \theta [x:=Q]]$ by a weak deterministic rule, and by
  Lemma~\ref{lemma:det}, the left-hand side
  $Next^* ({E'}^0 [\retkw Q]) (\rover V)$ is equal to
  $Next^* (E^0 [P \theta [x:=Q]]) (\rover V)$.  The latter is larger
  than or equal to $h^\dagger (\Eval P \rho [x \mapsto b]) (V)$ since
  $P \theta [x:=Q] \R_{D \tau} \Eval P \rho [x \mapsto b]$ by
  induction hypothesis and since
  $E^0 \R_{\tau \vdash D \beta}^\perp h$.

  Since
  ${E'}^0 \R_{\sigma \vdash D \beta}^\perp h^\dagger \circ \Eval
  {\lambda x_\sigma . P} \rho$ and since by induction hypothesis
  $N \theta \R_{D \sigma} \Eval N \rho$, we have
  $Next^* ({E'}^0 [N \theta]) (\rover U) \geq (h^\dagger \circ \Eval
  {\lambda x_\sigma . P} \rho)^\dagger (\Eval N \rho) (U)$, and this
  is exactly what we had to prove.

  $\bullet$ If $M = \langle M_1, M_2 \rangle$, with
  $M_1 \colon \tau_1$ and $M_2 \colon \tau_2$, then
  $M \theta \to^* \langle M_1 \theta, M_2 \theta \rangle$ by weak
  deterministic rules (vacuously), and
  $M_1 \theta \R_{\tau_1} \Eval {M_1} \rho$,
  $M_2 \theta \R_{\tau_2} \Eval {M_2} \rho$, by induction hypothesis.

  $\bullet$ If $M = \pi_1 N$, with $N \colon \tau_1 \times \tau_2$,
  then $N \theta \R_{\tau_1 \times \tau_2} \Eval N \rho$ by induction
  hypothesis.  If $\Eval N \rho = \bot$, then $\Eval M \rho = \bot$,
  and therefore $M \theta \R_{\tau_1 \times \tau_2} = \Eval M \rho$,
  by Lemma~\ref{lemma:MR:closed}.  Otherwise, $\Eval N \rho$ is of the
  form $(a_1, a_2)$ with $a_1 \neq \bot$ or $a_2 \neq \bot$.  By
  definition of $\R_{\tau_1 \times \tau_2}$,
  $N \theta \to^* \langle N_1, N_2 \rangle$ by weak deterministic
  rules in such a way that $N_1 \R_{\tau_1} a_1$ and
  $N_2 \R_{\tau_2} a_2$.  We note that
  $M \theta = \pi_1 N \theta \to^* \pi_1 \langle N_1, N_2 \rangle \to
  N_1$ by weak deterministic rules, so
  $M \theta \R_{\tau_1} a_1 = \Eval M \rho$, using
  Lemma~\ref{lemma:R:red}.

  The case $M = \pi_2 N$ is similar.

  $\bullet$ If $M = \iota_1 N \colon \tau_1 + \tau_2$, then vacuously
  $M \theta \to^* \iota_1 N \theta$ by weak deterministic rules.  By
  induction hypothesis, $N \theta \R_{\tau_1} \Eval N \rho$, so
  $M \theta \R_{\tau_1 + \tau_2} (1, \Eval N \rho) = \Eval M \rho$.

  The case $M = \iota_2 N$ is similar.

  $\bullet$ If $M = \casekw N P_1 P_2 \colon \tau$, with
  $N \colon \tau_1+\tau_2$, then by induction hypothesis
  $N \theta \R_{\tau_1+\tau_2} \Eval N \rho$.  If
  $\Eval N \rho = \bot$, then $\Eval M \rho = \bot$, so
  $M \theta \R_\tau \Eval M \rho$, by Lemma~\ref{lemma:MR:closed}.
  Otherwise, $\Eval N \rho = (i, b)$ for some $i \in \{1, 2\}$ and
  $b \in \Eval {\tau_i}$.  By definition of
  $\R_{\tau_1+\tau_2}$, $N \theta \to^* \iota_i N'$ by weak
  deterministic rules, in such a way that $N' \R_{\tau_i} b$.  Then
  $M \theta \to^* \casekw {(\iota_i N')} {(P_1\theta)} {(P_2\theta)}
  \to P_i\theta N'$ by weak deterministic rules.  By induction
  hypothesis, $P_i \theta \R_{\tau_i \to \tau} \Eval {P_i} \rho$, so
  $P_i \theta N' \R_\tau \Eval {P_i} (\rho) (b)$.  By
  Lemma~\ref{lemma:R:red},
  $M \theta \R_\tau \Eval {P_i} (\rho) (b) = \Eval M \rho$.
\end{proof}

We obtain the generalized ISPCF terms as $M \theta$, where
$(M, \theta)$ ranges over $\Config$.  In that case, for every variable
$x$ in the domain of $\theta$, $x$ has type $\realT$, and if we write
$\underline {\mathbf a}$ for $\theta (x)$, then we may define
$\rho (x)$ as $\mathbf a$.  Then $\Eval M \rho$ is the semantics
$\Eval {M \theta}$ of the generalized ISPCF term $M \theta$.  By
Lemma~\ref{lemma:R:obs:val},
$\underline {\mathbf a} \R_{\realT} \mathbf a$, so $\theta \R_* \rho$.
Proposition~\ref{prop:basic:lemma} then implies the following.
\begin{corollary}
  \label{corl:basic:lemma}
  Let $\Sigma$ consist of strictly observable first-order constants.  For every
  generalized ISPCF term $M \colon \tau$, $M \R_\tau \Eval M$.
\end{corollary}

\begin{lemma}
  \label{lemma:R:obs}
  Let $\Sigma$ consist of strictly observable first-order constants.  For every
  observable type $\beta$, for every generalized ISPCF term
  $N \colon \beta$, for every $b \in \Eval \beta$, if $N \R_\beta b$
  then either $b$ is not observable, or $b$ is observable and
  $N \to^* V$ by deterministic rules for some value $V \colon \beta$
  such that $\Eval V \geq b$.
\end{lemma}
Note that $N \to^* V$ by deterministic rules, but not necessarily weak
rules.

\begin{proof}
  We show that $N \R_\beta b$ and the fact that $b$ is observable
  together imply the existence of a value $V \colon \beta$ such that
  $\Eval V \geq b$, by induction on $\R_\beta$.  If $\beta$ is a basic
  type, then this is by definition of $\R_\beta$.  If $\beta$ is a
  product type $\beta_1 \times \beta_2$, then let us write $b$ as
  $(b_1, b_2)$, where $b_1 \neq \bot$ and $b_2 \neq \bot$ since $b$ is
  observable.  Since $b \neq \bot$, by definition of
  $\R_{\beta_1 \times \beta_2}$, $M \to^* \langle M_1, M_2 \rangle$ by
  weak deterministic rules, for some $M_1 \R_{\tau_1} b_1$ and
  $M_2 \R_{\tau_2} b_2$.  By induction hypothesis, $M_1 \to^* V_1$ and
  $M_2 \to^* V_2$ by weak deterministic rules, for some values $V_1$
  and $V_2$ such that $\Eval {V_1} \geq b_1$, $\Eval {V_2} \geq b_2$.
  Then
  $M \to^* \langle M_1, M_2 \rangle \to^* \langle V_1, M_2 \rangle
  \to^* \langle V_1, V_2 \rangle$ by (non-weak) deterministic rules.
  The term $V \eqdef \langle V_1, V_2 \rangle$ is a value, and
  $\Eval V \geq (b_1, b_2) = b$.  The case of sum types is similar.
\end{proof}

We finally obtain the promised adequacy theorem.
\begin{theorem}[Adequacy]
  \label{thm:adeq}
  Let $\Sigma$ consist of strictly observable first-order constants.  For every
  observable type $\beta$, for every generalized ISPCF term
  $M \colon D \beta$, for every observable open subset $U$ of
  $\Eval \beta$, $\Eval M (U) = Next^* (M) (\rover U)$.
\end{theorem}
\begin{proof}
  We consider the empty weak evaluation context $E^0 \eqdef []$, and
  we claim that
  $E^0 \R_{\beta \vdash D \beta}^\perp \eta_{\Eval \beta}$.  In order
  to show this, we need to show that for all $N \R_\beta b$, for every
  observable open subset $U$ of $\Eval \beta$,
  $Next^* (E^0 [\retkw N]) (\rover U) \geq \delta_b (U)$.

  If $b \not\in U$, then $\delta_b (U) = 0$, and this is obvious.

  Otherwise, $\delta_b (U) = 1$.  Since $b \in U$ and since $U$ is
  observable, $b$ is an observable element of $\Eval \beta$.  Since
  $N \R_\beta b$, Lemma~\ref{lemma:R:obs} tells us that $N \to^* V$ by
  deterministic rules for some value $V \colon \beta$ such that
  $\Eval V \geq b$.  Since $U$ is upwards closed, $\Eval V$ is in $U$.
  Therefore $V$ is in $\overline U$, and $\retkw V$ is in $\rover U$.

  We note that since $N \to^* V$ by deterministic rules,
  $\retkw N \to^* \retkw V$ by deterministic rules as well.  By
  Lemma~\ref{lemma:det},
  $Next^* (E^0 [\retkw N]) = Next^* \allowbreak (\retkw N)$ is equal
  to $Next^* (\retkw V)$.  By Lemma~\ref{lemma:Next*}, item~1, (and
  taking suprema over all $n$,) since $\retkw V$ is a normal form,
  $Next^* (\retkw V) (\rover U) = \delta_{\retkw V} (\rover U) = 1 =
  \delta_b (U)$.

  Now that we have proved
  $E \R_{\beta \vdash D \beta}^\perp \eta_{\Eval \beta}$, we apply it
  to $M \R_{D \beta} \Eval M$ (Corollary~\ref{corl:basic:lemma}), and,
  noting that $E = []$ is a weak evaluation context, we obtain that
  $Next^* (E [M]) (\rover U) \geq (\eta_{\Eval \beta})^\dagger (\Eval
  M) (U) = \Eval M (U)$.  The converse inequality is by soundness
  (Proposition~\ref{prop:sound}).
\end{proof}

\begin{remark}
  \label{rem:adeq:D}
  Why do we restrict $M$ to be of type $D \beta$ in
  Theorem~\ref{thm:adeq}?  Restricting it to be of observable type
  $\beta$ would be silly, because only deterministic computations
  happen at type $\beta$.  In order to be able to do any probabilistic
  computation at all, $M$ has to be of distribution type.  A similar
  choice was made in \cite{jgl-jlap14,JGL-lics19}.
\end{remark}

\begin{remark}
  \label{rem:adeq:proper}
  The reason why Theorem~\ref{thm:adeq} only considers
  \emph{observable} open subsets $U$ is that adequacy is concerned
  with \emph{terminating} computations, and terms of observable type
  $\beta$ terminate on values $V$ (if they do terminate at all), whose
  semantics is observable.  For non-observable open subsets $U$, the
  denotational semantics gives us more information.  For example, when
  $U$ is the whole of $\Eval \beta$, $Next^* (M) (\rover U)$ is the
  `probability' that $M$ terminates at all, while $\Eval M (U)$ is
  larger in general, and represents the `probability' of all
  executions starting from $M$, including those that do not terminate.
  For example, let $M \eqdef \retkw (\reckw (\lambda n_{\intT} . n))$,
  then $Next^* (M) (\rover {\Eval {\intT}}) = 0$, while
  $\Eval M (\Eval {\intT}) = \delta_\bot (\Eval {\intT}) = 1$.  For a
  product type $\beta_1 \times \beta_2$, the Scott-open set
  $U \eqdef \Eval {\beta_1} \times U'$, where $U'$ is a non-empty
  observable subset of $\Eval {\beta_2}$, is not observable either,
  but is smaller than $\Eval {\beta_1 \times \beta_2}$.  In that case,
  $Next^* (M) (\rover U)$ is the `probability' that $M$ terminates on
  a value $\retkw \langle V_1, V_2 \rangle$ with
  $V_2 \in \overline {U'}$, but $\Eval M (U)$ is the `probability'
  that $\pi_2 M$ terminates with a value in $U'$, thereby also
  counting all computations for which $\pi_1 M$ does not terminate.
  Another view to the question is given by Lemma~\ref{lemma:rover} below.
\end{remark}

\begin{remark}
  \label{rem:infprod}
  The reader interested in extending the adequacy theorem to
  non-observable open subsets $U$ should be warned of the following
  difficulty.  Let $U$ be the whole of $\Eval \beta$.  While
  $\Eval M (U)$ is perfectly well-defined, what would be the
  operational meaning of ``the `probability' of all executions
  starting from $M$''?  This is difficult, already for executions not
  involving $\sample [\mu]$ and rule (\ref{rule:sample}), but
  involving $\score$ and rule (\ref{rule:score}).  Imagine an infinite
  execution starting for $M$ and calling $\score$ on real values
  $a_1$, $a_2$, \ldots, in succession.  One would imagine that the
  `probability' of the unique execution starting from $M$ is
  $\prod_{i=1}^{+\infty} a_i$.  However, that is ill-defined in
  general.  For an example, consider $a_i \eqdef 2$ for $i$
  odd, $1/2$ for $i$ even.
\end{remark}

We say that $\Sigma$ is \emph{terse} if and only if for every
$n \in \Z$, there is a unique zero-ary constant $\underline n$ of type
$\intT$ such that $\Eval {\underline n} = n$, and $\underline *$ is
the only zero-ary constant of type $\unitT$.  In this case, we can
avoid the use of sets of the form $\rover U$ in the definition of
$\precsim^{app}$ and $\cong^{app}$, because of the following lemma.
Note that we do not need to require a similar condition of terseness
on $\voidT$ (if $\Sigma$ consists of strictly observable first-order constants,
then there can be no zero-ary constant of type $\voidT$), or on
$\realT$ (because of the way we defined generalized ISPCF terms, see
Section~\ref{sec:ideal-oper-semant}).

\begin{lemma}
  \label{lemma:rover}
  Let $\Sigma$ be terse and consist of strictly observable first-order constants.
  For every observable type $\beta$, the map $M \mapsto \Eval M$ is an
  order isomorphism from $Norm_\beta$ to the subdcpo $Obs_\beta$ of
  observable elements in $\Eval \beta$.  The open subsets of
  $Norm_\beta$ are exactly the sets of the form $\overline U$, where
  $U$ is a observable Scott-open subset of $\Eval \beta$.  Those of
  $Norm_{D \beta}$ are exactly the sets of the form $\rover U$, where
  $U$ is a observable Scott-open subset of $\Eval \beta$.
\end{lemma}
\begin{proof}
  The first claim is clear.  The inverse map is the map
  $a \mapsto \underline a$ of Definition~\ref{defn:underline}.  For
  the second claim, $Obs_\beta$ is a Scott-open subset of
  $\Eval \beta$, and is in particular a subdcpo (i.e., directed
  suprema are computed as in $\Eval \beta$).  Let $U'$ be any open
  subset of $Norm_\beta$.  Its image by the map
  $a \mapsto \underline a$ is an open subset of $\Eval \beta$ included
  in $Obs_\beta$, hence is Scott-open in $Obs_\beta$, since
  $Obs_\beta$ is a subdcpo of $\Eval \beta$.  The third claim follows
  from the fact that $M \mapsto \retkw M$ is an isomorphism between
  $Norm_\beta$ and $Norm_{D \beta}$.
\end{proof}

\subsection{Contextual equivalence and the applicative preorder}
\label{sec:observ-equiv-observ}

An important, and expected, consequence of adequacy, is that equality
of denotations implies contextual equivalence, namely that for any two
generalized ISPCF terms $M$, $N$ of any type $\tau$, if
$\Eval M=\Eval N$ then $M$ and $N$ are contextually equivalent,
meaning that every observable type $\beta$, for every context
$\mathcal E$ of type $\tau \vdash D \beta$, the `probability' that
$\mathcal E [M]$ evaluates to a value in any given open set $U$ of
normal forms is equal to the `probability' that $\mathcal E [N]$
evaluates to a value in $U$.

We define the required contexts $\mathcal E$ as generalized ISPCF
terms with arbitrarily many holes.  This represents the fact that one
may replace $M$ at any number of positions in a program by a
contextually equivalent program $N$ and not change the observable
behavior of the program.

In order to formalize this, there is no need to invent a new notion of
multi-hole context.  Indeed, replacing all the holes in $\mathcal E$
with a term $M$ gives a term with the same semantics as the
application $QM$, where $Q \eqdef \lambda x . \mathcal E [x]$.  Hence
the notion of contextual equivalence we are looking is the following
one, usually called \emph{applicative equivalence}.  This is the
equivalence relation associated with a more primitive notion of
\emph{applicative preorder}, which we now define.
\begin{definition}[Applicative preorder $\precsim^{app}$ and equivalence]
  \label{defn:app}
  For any two generalized ISPCF terms $M$, $N$ of the same type
  $\tau$, $M \precsim^{app} N$ (resp., $M \cong^{app} N$) if and only
  if, for every observable type $\beta$, for every generalized ISPCF
  term $Q \colon \tau \to D \beta$, for every open subset $U$ of
  $Norm_{D \beta}$, $Next^* (QM) (U) \leq Next^* (QN) (U)$ (resp.,
  $=$).
\end{definition}

The adequacy Theorem~\ref{thm:adeq} allows us to simplify the
definition of $\precsim^{app}$ and $\cong^{app}$ as follows.  Note
that, if $Q = \lambda x_\tau . \mathcal E [x]$, where $\mathcal E$ is
one of our purported multi-hole contexts, then
$\Eval {QM} = \Eval {\mathcal E [M]}$, justifying the view that the
applicative preorder is really about comparing programs with zero,
one, or several occurrences of $M$ replaced by another term $N$.
\begin{proposition}
  \label{prop:precsim}
  Let $\Sigma$ be terse and consist of strictly observable first-order constants.
  For all generalized ISPCF terms $M$, $N$ of the same type $\tau$,
  $M \precsim^{app} N$ (resp., $M \cong^{app} N$) if and only if, for
  every observable type $\beta$, for every generalized ISPCF term
  $Q \colon \tau \to D \beta$, for every observable Scott-open subset $U$
  of $\Eval \beta$, $\Eval {QM} (U) \leq \Eval {QN} (U)$.
\end{proposition}
\begin{proof}
  By Lemma~\ref{lemma:rover}, one may replace the quantification over
  open subsets $U$ of $Norm_{D \beta}$ by a quantification over open
  subsets of the form $\rover U$, where $U$ ranges over the observable
  Scott-open subsets.  Then, by Theorem~\ref{thm:adeq}, $Next^* (QM)
  (\rover U) = \Eval {QM} (U)$, and similarly with $QN$.
\end{proof}

Importantly, we obtain the following easy result.
\begin{proposition}
  \label{prop:adeq:equiv}
  Let $\Sigma$ be terse and consist of strictly observable first-order constants.  Let $M$ and
  $N$ be two generalized ISPCF terms of the same type $\tau$.  The
  following implications hold:
  \begin{enumerate}
  \item $\Eval M \leq \Eval N \limp M \precsim^{app} N $;
  \item $\Eval M = \Eval N \limp M \cong^{app} N$.
  \end{enumerate}
\end{proposition}
\begin{proof}
  If $\Eval M \leq \Eval N$, then for every observable type $\beta$,
  for every generalized ISPCF term $Q \colon \tau \to D \beta$,
  $\Eval {QM} = \Eval Q (\Eval M) \leq \Eval Q (\Eval N) = \Eval
  {QN}$.  Using Proposition~\ref{prop:precsim}, $M \precsim^{app} N$.  The
  second part of the theorem is an easy consequence of the first part.
\end{proof}
The converse implications would be a form of full abstraction, and are
quite probably hopeless.  Therefore we will not bother with them.

\begin{remark}
  \label{rem:equiv}
  Following up on Remark~\ref{rem:comm}, it follows that when $x$ is
  not free in $N$ and $y$ is not free in $M$, the terms
  $\dokw {x \leftarrow M}; \dokw {y \leftarrow N}; P$ and
  $\dokw {y \leftarrow N}; \dokw {x \leftarrow M}; P$ are contextually
  equivalent.
\end{remark}

\begin{remark}
  \label{rem:equiv:normal}
  Similarly, the terms $\mathtt{normal}$,
  $\dokw {\langle x, y \rangle \leftarrow \mathtt{box\_muller}}; \retkw
  x$, and
  $\dokw {\langle x, y \rangle \leftarrow \mathtt{box\_muller}}; \retkw
  y$ (see Section~\ref{sec:box-muller-algorithm}) are contextually
  equivalent.  The terms $\mathtt{box\_muller}$ and
  $\mathtt{box\_muller'}$ are contextually equivalent, and so are the
  terms $\mathtt{expo}$, $\mathtt{expo}'$ and $\mathtt{von\_neumann}$
  of Section~\ref{sec:gener-expon-distr}.
\end{remark}

\section{A Precise Operational Semantics}
\label{sec:meas-oper-semant}

We give a second operational semantics that is closer to the intent of
\cite{VKS:SFPC,DLH:geom:bayes,EPT:PPCF}.  This requires slightly
stricter assumptions than simply having strictly observable first-order constants.
Those assumptions are always met in practice, as far as we know.  Note
that $\exp$, $\sin$, $\cos$, but also $\posop$ and $\log$, are
precise, for example.  One instance of a non-precise constant would be
the sign map, if we decide to define it as $\hat g$ (see
Section~\ref{sec:semantics}), where $g$ maps every negative number to
$0$ and every non-negative number to $1$; indeed, $\hat g$ maps every
interval $[a, b]$ with $a \leq 0\leq b$ to $[0, 1]$.  We can define
the (precise) sign map as
$\lambda x_{\realT} . \ifkw (\poskw\;x)\;\underline{1.0}\;\underline
{0.0}$ instead.

\begin{definition}
  \label{defn:assum:meas}
  The set $\Sigma$ consists of \emph{precise first-order constants} if
  and only if $\Sigma$ has strictly observable first-order constants, and
  additionally, for each constant $\underline f \colon \sigma_1 \to
  \cdots \to \sigma_k \to \tau$ of arity $k$ in $\Sigma$, $f$ is
  \emph{precise}, namely:
  \begin{itemize}
  \item for all maximal elements $v_1 \in \Eval {\sigma_1}$, \ldots,
    $v_k \in \Eval {\sigma_k}$, $f (v_1) \cdots (v_k)$ is either
    maximal in $\Eval \tau$ or equal to $\bot$.
  \end{itemize}
\end{definition}
When $\sigma_1=\cdots=\sigma_k=\tau = \realT$, this means that $f$
encodes a partial continuous map from $\real^k$ to $\real$.  If we
equate $i (v) = [v,v]$ with the real number $v$, the domain
$\{(v_1, \cdots, v_k) \in \real^k \mid f (v_1) \cdots (v_k) \neq
\bot\}$ of $f$ is an open subset of $\real^k$, and $f$ is required to
be continuous on that domain.

A \emph{precise} configuration $(M, \theta)$ is a configuration where
for every $x \in \FV (M)$, $\theta (x)$ is a maximal element of
$\IR_\bot$, namely a real number.

A \emph{precise} generalized ISPCF term is a term $M \theta$ obtained
from a precise configuration $(M, \theta)$.  Given a shape $M$, one
can equate $\theta$ with an element of $\real^{\FV (M})$, namely with
a finite tuple of real numbers.  $\real^{\FV (M)}$, and by extension
the set of precise configurations with shape $M$, is equipped with its
usual metric topology, which we will call its \emph{Euclidean topology}.
\begin{definition}[$\Config^0$]
  \label{defn:Config0}
  We topologize 
  $\Config^0$ as the set of all precise configurations, topologized as
  the coproduct over all shapes $M$ of the spaces of precise
  configurations with shape $M$, each being given the Euclidean topology.
\end{definition}
On the one hand, $\Config^0$ is a much more familiar space to the
ordinary topologist.  It is a countable coproduct of Polish spaces,
hence it is itself Polish.  On the other hand, $\Config^0$ is also a
topological subspace of $\Config$.  We write $\iota$ for the subspace
embedding; we have $\iota (M, \theta) = (M, i \circ \theta)$, where
$i$ is the familiar subspace embedding $r \mapsto [r, r]$ of $\real$
into $\IR_\bot$.

\begin{lemma}
  \label{lemma:Next0}
  Assume that $\Sigma$ consists of precise first-order constants.
  There is a unique continuous map $Next^0 \colon \Config^0 \to \Val
  \Config^0$ such that, for every $(M, \theta) \in \Config^0$, $\iota
  [Next^0 (M, \theta)] = Next (\iota (M, \theta))$.
\end{lemma}
\begin{proof}
  By inspection of Definition~\ref{defn:opsem}.  The key point is the
  case of $\underline f$, where the result holds because $\Sigma$ has
  precise first-order constants.
\end{proof}

\begin{remark}
  \label{rem:Config0:min}
  In general, $Next (M \theta)$ is a minimal valuation on $\Config$.
  However, $Next^0 (M, \theta)$ need not be a minimal valuation.  For
  example, $Next^0 (x, \allowbreak [x:=\sample [\lambda]])$ is, up to
  isomorphism, the Lebesgue valuation on $\real$, which is not minimal
  by Lemma~\ref{lemma:lambda:notmin}.
\end{remark}

$Next^0$ defines a kernel.  Using Lemma~\ref{lemma:iter}, we obtain a
continuous map $Next^{0,*} \colon \Config^0 \to \Val \Config^0$,
provided we replace $Norm$ by the set
$Norm^0 \eqdef Norm \cap \Config^0$.  By induction on $n$, it is easy
to see that
$\iota [Next^{0, \leq n} (M, \theta)] = Next^{\leq n} (\iota (M,
\theta))$.  It follows:
\begin{lemma}
  \label{lemma:Next0*}
  Assume that $\Sigma$ consists of precise first-order constants.
  For every precise configuration $(M, \theta)$,
  $\iota [Next^{0,*} (M, \theta)] = Next^* (\iota (M, \theta))$.
\end{lemma}

Proposition~\ref{prop:sound} and Theorem~\ref{thm:adeq} then
immediately imply the following.  We equate $\Config_0$ with a
subspace of $\Config$, and omit any mention of the embedding $\iota$.
\begin{theorem}
  \label{thm:precise:adeq}
  Let $\Sigma$ consist of precise first-order constants.  For every
  precise generalized ISPCF term $M \colon D \tau$, where $\tau$ is
  any type, for every open subset $U$ of $\Eval \tau$,
  $\Eval M (U) \geq Next^{0,*} (M) (\rover U \cap \Config^0)$.  This
  inequality is an equality if $\tau$ is an observable type and $U$ is
  observable.
\end{theorem}

\section{A Sampling-Based Operational Semantics}
\label{sec:effectivity}

An operational semantics is usually an abstract view of an
implementation.  However, operational semantics such as those of
previous sections (and such as those of
\cite{VKS:SFPC,DLH:geom:bayes,EPT:PPCF}, which even go as far as
manipulating \emph{true} real numbers) are far from an implementation.
Park, Pfenning and Thrun \cite{PPT:sampling}, and later Dal Lago and
Hoshino \cite{DLH:geom:bayes} introduce a more concrete
\emph{sampling-based} operational semantics.  Roughly speaking, it is
more concrete in the sense that it no longer has to sample from
arbitrary measures on $\real$, and instead draws arbitrarily long
strings of independent, uniform random bits, which are drawn at random
prior to execution.

We use a similar idea here.  This will apply under some reasonable
assumptions, the most important one being that we only have one term
of the form $\sample [\mu]$, namely $\sample [\lambda_1]$.  As we have
hinted in Section~\ref{sec:examples}, one can define quite a number of
other distributions from that one alone.

\subsection{Random strings}
\label{sec:random-strings}

Our sampling-based operational semantics will be a
\emph{deterministic} semantics, parameterized by an infinite string of
bits, which we can think as having been drawn at random, uniformly, in
advance.  This trick is routinely used in complexity theory, for
example, where a popular definition of a randomized Turing machine is
a Turing machine with an additional read-only \emph{random tape} on
which the head can only move right.  It is also at the heart of the
definition of computable probability theory, where a \emph{computable
  random variable} on a computable metric space $S$ is a measurable
map $X \colon \{0, 1\}^\nat \to S$ such that $X$ is computable on a
subset of $\{0, 1\}^\nat$ of $\upsilon$-measure one
\cite[Definition~II.12]{Roy:PhD}.  We have also used similar tricks in
Section~\ref{sec:more-distr-high}, too.

We will reuse several notions and notations from
Section~\ref{sec:distr-high-order}.  Notably, $\upsilon$ is the
uniform measure $\bin [\lambda_1]$ on $\{0, 1\}^\nat$.  We also recall
the maps $\bin \colon \real \to \{0, 1\}^\nat$,
$\num \colon \{0, 1\}^\nat \to \real$,
$split \colon \{0,1\}^\nat \to (\{0,1\}^\nat)^2$, the notation
$s [m] \eqdef {(s_{\langle m, n\rangle})}_{n \in \nat}$ and the shift
$\upc s$.


\subsection{The sampling-based operational semantics}
\label{sec:sampl-based-oper}

Our next semantics will be a partial map
$Next' (s) \colon \Config' \to \Config'$, parameterized by a fixed
random string $s$, and where $\Config'$ is a space of so-called
\emph{enriched configurations}.  This is somewhat similar to the
sampling-based operational semantics of
\cite[Section~3.3]{DLH:geom:bayes}.  An enriched configuration is a
tuple $(M, i, r)$ where $M$ is a generalized ISPCF term, $i \in \nat$,
and $r \in \realp$: $i$ is the number of the next random string
$s [i]$ to use in order to implement rule (\ref{rule:sample}), and $r$
is the product of the score values produced by rule (\ref{rule:score})
so far.  We let $\Config'$ be the space of all indexed configurations.
This is a product dcpo $\Config \times \nat \times \realp$, where
$\nat$ is ordered by equality and $\realp$ by its usual ordering.  The
Scott topology on the product coincides with the product topology,
since this is a finite product of continuous posets
\cite[Proposition~5.1.54]{goubault13a}; as
such, $\Config'$ is also a continuous poset.  Random strings are
elements of the topological product $\{0, 1\}^\nat$, where $\{0, 1\}$
has the discrete topology.
\begin{definition}[Sampling-based operational semantics]
  \label{defn:seeded}
  Assume that all constants in $\Sigma$ are first-order constants, and
  that the only term of the form $\sample [\mu]$ is
  $\sample [\lambda_1]$, where $\lambda_1$ is Lebesgue measure on
  $[0, 1]$.  The \emph{sampling-based operational semantics} of ISPCF
  is the partial map
  $Next' \colon \{0, 1\}^\nat \to \Config' \to \Config'$ defined by:
  \begin{itemize}
  \item for every instance of a rule $L \to R$ of
    Figure~\ref{fig:ideal:opsem}
    except (\ref{rule:sample}), (\ref{rule:score}) and (\ref{rule:f}),
    $Next' (s) (L, i, r) = (R, i, r)$;
  \item (case of $\sample$)
    $Next' (s) (E [\sample [\lambda_1]], i, r) \eqdef (E [\retkw
    \underline {\num(s [0])}], i+1, r)$;
  \item (case of $\score$)
    $Next' (s) (E [\score\; \underline {\mathbf a}], i, r) \eqdef (E [\retkw
      \underline *], i, |\mathbf a| . r)$;
    \item (case of $\underline f$) If
      $k = \alpha (\underline f) \neq 0$, then
      $Next' (s) (E [\underline f\; \underline a_1 \cdots \underline
      a_k], i, r) \eqdef (E [\underline {f (a_1) \cdots (a_k)}], i,
      r)$ if $f (a_1) \cdots (a_k)$ is observable;
  \item $Next' (s) (M, i, r)$ is undefined otherwise.
  \end{itemize}
\end{definition}
Let us write $(M, i, r) \to_s (M', i', r')$ for
``$Next' (s) (M, i, r)$ is defined and equal to $(M', i', r')$''.  Let
also $\to^*_s$ be the reflexive-transitive closure of $\to_s$.  For
every generalized ISPCF term $M$, and every open set $U$ of normal
forms, there is at most one enriched configuration $(M', i', r')$ such
that $(M, 0, 1) \to^*_s (M', i', r')$ and $M' \in U$.  If it exists,
we call that $r'$ the \emph{score} $\sigma [U] (M, s)$ of $M$ relative
to $s$ and $U$; otherwise we let $\sigma [U] (M, s) \eqdef 0$.  Note
that the score depends on $M$, on $U$, and on $s$.

One can refine this as follows.  This will be needed below.  Let
$\to^{\leq k}$ denote reachability in at most $k$ $Next'$ steps.
\begin{definition}
  \label{defn:sigmak}
  For each $k \in \nat$, let $\sigma_k [U] (M, s)$ be the unique value
  $r'$ such that $(M, 0, 1) \allowbreak \to^{\leq k}_s (M', i', r')$
  for some $M' \in U$, if that exists, and $0$ otherwise.
\end{definition}
Then $\sigma [U]$ is the supremum of the directed family of maps
${(\sigma_k [U])}_{k \in \nat}$.

\begin{lemma}
  \label{lemma:shift}
  For every $k \in \nat$, for all $i, i'' \in \nat$,
  $r, r' \in \realp$, for all generalized ISPCF terms $M$ and $M'$,
  $(M, i, r) \to_{\upc s}^{\leq k} (M', i', r')$ if and only if
  $(M, i+1, r) \to_s^{\leq k} (M', i'+1, r')$.
\end{lemma}
\begin{proof}
  It suffices to show that $(M, i, r) \to_{\upc s} (M', i', r')$ if
  and only if $(M, i+1, r) \to_s (M', i'+1, r')$, which is an easy
  verification, and to induct on $k$.
\end{proof}

The following is easily proved, too.
\begin{lemma}
  \label{lemma:scale}
  For every $k \in \nat$, for all $i, i'' \in \nat$,
  $r, r' \in \realp$, $\alpha \in \realp$, for all generalized ISPCF
  terms $M$ and $M'$, if $(M, i, 0) \to_s^{\leq k} (M', i', r')$ then
  $r'=0$.  If $(M, i, r) \to_s^{\leq k} (M', i', r')$ then
  $(M, i+1, \alpha.r) \to_s^{\leq k} (M', i'+1, \alpha.r')$.  The
  converse implication holds if $\alpha \neq 0$.
\end{lemma}

\begin{lemma}
  \label{lemma:score:cont}
  Under the assumptions of Definition~\ref{defn:seeded}, for every
  open subset $U$ of $Norm$, the maps $\sigma_k [U]$ ($k \in \nat$)
  and $\sigma [U]$ are lower semicontinuous from
  $\Config \times \{0, 1\}^\nat$ to $\realp$.
\end{lemma}
\begin{proof}
  We recall that lower semicontinuous simply means continuous,
  provided that we equip $\realp$ with its Scott topology, which we
  now assume.
  
  In order to simplify the proof, we use the following result by
  Ershov \cite[Proposition~2]{Ershov:aspace:hull}: given a c-space $X$
  (an $\alpha$-space in Ershov's terminology, e.g., a continuous poset
  with its Scott topology), and two topological spaces $Y$ and $Z$, a
  function $f \colon X \times Y \to Z$ is separately continuous if and
  only if it is jointly continuous.  Joint continuity is ordinary
  continuity from $X \times Y$ with the product topology.  Separate
  continuity means that $f (x, \_)$ and $f (\_, y)$ are continuous for
  fixed $x$ and $y$ respectively.  This applies here, as $\Config'$ is
  a continuous poset.

  Since suprema of lower semicontinuous maps are lower semicontinuous,
  it suffices to show that $\sigma_k [U]$ is lower semicontinuous for
  every $k \in \nat$, which we do by induction on $k$.

  We fix a template $M_0$, and we show that $\sigma_k [U]$ is
  continuous from the product of the space of generalized ISPCF terms
  of shape $M_0$ with $\{0, 1\}^\nat$ to $\creal$.  We only need
  to prove separate continuity.
  
  If $k=0$, then $\sigma_0 [U] (M, s)=1$ if $M$ is normal, $0$
  otherwise.  For $M$ fixed, this is constant hence continuous in $s$.
  For $s$ fixed, we note that normality only depends on the shape
  (here, $M_0$).  Hence $\sigma_0 [U] (\_, s)$ is constant, hence
  continuous, on the subspace of generalized ISPCF terms of shape
  $M_0$.

  Let us assume $k \geq 1$.  If the unique $\to_s$ step from
  $(M, i, r)$ is by applying a rule $L \to R$ of $L \to R$ of
  Figure~\ref{fig:ideal:opsem} except (\ref{rule:sample}),
  (\ref{rule:score}) and (\ref{rule:f}), then we note that the same
  rule would apply if we replaced $M$ by any other term with the same
  shape $M_0$.  Then $\sigma_k [U] (M, s) = \sigma_{k-1} [U] (R, s)$,
  and we apply the induction hypothesis.

  In the case of $\sample$, $M = E [\sample [\lambda_1]]$, and if
  $(M, 0, 1) \to^{\leq k}_s (M', i', r')$ for some $M' \in U$, then
  $(M, 0, 1) \to_s (E [\retkw \underline {\num(s [0])}], 1, 1)
  \to^{\leq k-1}_s (M', i', r')$, where $r' = \sigma_k [U] (M, s)$.
  We then have
  $(E [\retkw \underline {\num(s [0])}], \allowbreak 0, 1) \to^{\leq
    k-1}_{\upc s} (M', i'-1, r')$ by Lemma~\ref{lemma:shift}, hence
  $r' = \sigma_k [U] (M, \upc s) = \sigma_{k-1} [U] (E [\retkw
  \underline {\num(s [0])}], \upc s)$.  For $M$ fixed, we see that the
  composition of the maps $s \mapsto s [0] \mapsto \num(s [0])$ is
  continuous from $\{0, 1\}^\nat$ to $\real$, using
  Lemma~\ref{lemma:num}; also, $s \mapsto \upc s$ is continuous.  As
  in the proof of Proposition~\ref{prop:next:cont}, the map
  $a \in \real \mapsto E [\retkw \underline {a}]$ is continuous, so
  $\sigma_k [U] (M, s)$ defines a continuous function of $s$.  For $s$
  fixed, this is the composition of the continuous function that maps
  every term $M = E [\sample [\lambda_1]]$ of shape $M_0$ to
  $E [\retkw \underline a]$ (where $a \eqdef \num(s [0])$) with
  $s_{k-1} [U] (\_, \upc s)$, hence is continuous as well.

  In the case of $\score$, we have
  $M = E [\score\;\underline{\mathbf a}]$, and assuming that
  $(M, 0, 1) \to^{\leq k}_s (M', i', r')$ for some $M' \in U$, we have
  $(M, 0, 1) \to_s (E [\retkw \underline *], 0, \allowbreak|\mathbf
  a|) \to^{\leq k-1}_s (M', i', r')$, and $r' = \sigma_k [U] (M, s)$.
  In particular, if $|\mathbf a|\neq 0$ then
  $(E [\retkw \underline *], 0, 1) \to^{\leq k-1}_s (M', i',
  r'/|\mathbf a|)$ by Lemma~\ref{lemma:scale} (final part), so
  $\sigma_{k-1} [U] (E [\retkw \underline *], s) = r' / |\mathbf a|$.
  Therefore
  $\sigma_k [U] (M, s) = |\mathbf a| . \sigma_{k-1} [U] (E [\retkw
  \underline *], s)$.  We see that this holds also when
  $|\mathbf a|=0$, using the first part of Lemma~\ref{lemma:scale},
  and also when $(M, 0, 1) \to^{\leq k}_s (M', i', r')$ for no
  $M' \in U$, in which case both sides of the equality are $0$.  By
  induction hypothesis, $\sigma_{k-1}$ is continuous.  The map
  $M = E [\score\;\underline{\mathbf a}] \mapsto E [\retkw \underline
  *]$ is continuous, by a similar argument as in
  Proposition~\ref{prop:next:cont}.  The map
  $M = E [\score\;\underline{\mathbf a}] \mapsto \mathbf a$ is also
  clearly continuous, the map $\mathbf a \mapsto |\mathbf a|$ is
  continuous from $\IR_\bot$ to $\creal$, and product is
  Scott-continuous on $\creal$, so $\sigma_k [U] (\_, s)$ is
  continuous, for every fixed $s$.  For
  $M = E [\score\;\underline{\mathbf a}]$ fixed,
  $\sigma_k [U] (M, \_) = |\mathbf a| . \sigma_{k-1} [U] (E [\retkw
  \underline *], \_)$ is obviously continuous.  Hence $\sigma_k [U]$
  itself is continuous.

  In case
  $M = E [\underline f\; \underline a_1 \cdots \underline a_k]$ where
  $k = \alpha (\underline f) \neq 0$, then either
  $f (a_1) \cdots (a_k)$ is observable and
  $\sigma_k [U] (M, s) = \sigma_{k-1} [U] (E [\underline {f (a_1)
    \cdots (a_k)}], s)$, or else $\sigma_k [U] (M, s) = 0$.  For $M$
  fixed, this defines a continuous map in $s$.  We now consider $s$
  fixed.  The function that maps every term $M$ as above, of shape
  $M_0$, to the tuple of (denotations of) constants of type $\realT$
  among $\underline a_1$, \ldots, $\underline a_k$ is continuous,
  hence also the function that maps $M$ (of given shape $M_0$) to
  $f (a_1) \cdots (a_k)$.  In particular the set $V$ of generalized
  ISPCF terms $M$ of shape $M_0$ such that $f (a_1) \cdots (a_k)$ is
  observable is open.  Let $g$ be the function that maps every
  $M = E [\underline f\; \underline {\mathbf a_1} \cdots \underline
  {\mathbf a_k}]$ in $V$ to
  $E [\underline {f (\mathbf a_1) \cdots (\mathbf a_k)}]$.  This is
  continuous.  The inverse image of every non-trivial Scott-open
  subset $]a, +\infty]$ of $\creal$ (namely, $a > 0$) by
  $\sigma_k [U] (\_, s)$ is equal to
  $g^{-1} (\sigma_{k-1} [U] (\_, s)^{-1} ]a, +\infty])$ is then open,
  showing that $\sigma_k [U] (\_, s)$ is continuous.
\end{proof}

We define the `probability' of reaching a normal form in some open
subset $U$ of $Norm$ through the sampling-based operational semantics
as:
\begin{align*}
  P [M \to U] & \eqdef \int_{s \in \{0, 1\}^\nat} \sigma [U] (M, s) d\upsilon.
\end{align*}
In other words, $P [M \to U]$ is the average score of all strings of
rewrite steps starting from $M$, where $s$ is drawn uniformly at
random.  The integral makes sense by
Proposition~\ref{lemma:score:cont}.  We now verify that $P [M \to U]$
is exactly the `probability' $Next^* (M) (U)$ that $M$ eventually
evaluates to a normal form in $U$ according to the operational
semantics of Section~\ref{sec:ideal-oper-semant}.
\begin{proposition}
  \label{prop:PU}
  Assume that all constants in $\Sigma$ are first-order constants, and
  that the only available term of the form $\sample [\mu]$ is
  $\sample [\lambda_1]$, where $\lambda_1$ is Lebesgue measure on
  $[0, 1]$.  For every open subset $U$ of $Norm$, for every
  generalized ISPCF term $M$, $P [M \to U] = Next^* (M) (U)$.
\end{proposition}
\begin{proof}
  We show that
  $\int_{s \in \{0, 1\}^\nat} \sigma_k [U] (M, s) d\upsilon =
  Next^{\leq k} (M) (U)$ by induction on $k \in \nat$.  The result
  will follow by taking suprema over $k$.

  When $k=0$, $Next^{\leq 0} (M) (U)$ is equal to $\delta_M (U)$ if
  $M \in Norm$, to $0$ otherwise, hence to $\chi_U (M)$ in all cases.
  For every $s \in [0, 1]^\nat$, $\sigma_0 [U] (M, s)$ is equal to
  $\chi_U (M)$, and since $\upsilon$ is a probability distribution,
  $\int_{s \in \{0, 1\}^\nat} \sigma_k [U] (M, s) d\upsilon = \chi_U
  (M)$.

  When $k \geq 1$, we look at the shape of $M$, and we use
  Lemma~\ref{lemma:Next*} to evaluate $Next^{\leq k} (M) (U)$.

  If $M \in Norm$, then
  $Next^{\leq k} (M) (U) = \delta_M (U) = \chi_U (M)$.  Also,
  $\sigma_k [U] (M) = \chi_U (M)$, so
  $\int_{s \in \{0, 1\}^\nat} \sigma_k [U] (M, s) d\upsilon = \chi_U
  (M)$.

  We arrive at the main case of the proof.  If
  $M = E [\sample [\mu]]$, where necessarily $\mu = \lambda_1$, then,
  as in the proof of Lemma~\ref{lemma:score:cont},
  $\sigma_k [U] (M, \upc s) = \sigma_{k-1} [U] (E [\retkw \underline
  {\num(s [0])}], \upc s)$, so:
  \begin{align*}
    \int_{s \in \{0, 1\}^\nat} \sigma_k [U] (M, s) d\upsilon
    & = \int_{s \in \{0, 1\}^\nat} \sigma_{k-1} [U] (E [\retkw \underline
      {\num(s [0])}], \upc s) d \upsilon.
  \end{align*}
  We now use the change-of-variables formula with respect to the
  homeomorphism
  $split \colon s \in \{0, 1\}^\nat \mapsto (s [0], \upc s) \in (\{0,
  1\}^\nat)^2$ (see Lemma~\ref{lemma:split}).  Since $split [\upsilon]
  = \upsilon \otimes \upsilon$, and using Proposition~\ref{prop:fubini:top}:
  \begin{align*}
    \int_{s \in \{0, 1\}^\nat} \sigma_k [U] (M, s) d\upsilon
    & = \int_{t \in \{0, 1\}^\nat} \left(\int_{(s' \in \{0, 1\}^\nat} \sigma_{k-1} [U] (E
      [\retkw \underline {\num(t)}], s') d\upsilon\right) d \upsilon \\
    & = \int_{t \in \{0, 1\}^\nat} Next^{\leq k-1} (E [\retkw
      \underline {\num(t)}]) (U) d\upsilon \\
    & \text{by induction hypothesis} \\
    & = \int_{a \in \real} Next^{\leq k-1} (E [\retkw \underline a])
      (U) d\lambda_1 \\
    & \text{by the change-of-variables formula and
      Lemma~\ref{lemma:num}} \\
    & = Next^{\leq k} (M) (U),
  \end{align*}
  by Lemma~\ref{lemma:Next*}, item~2.

  If $M = E [\score \; \underline {\mathbf a}]$, then, as in the proof
  of Lemma~\ref{lemma:score:cont},
  $\sigma_k [U] (M, s) = |\mathbf a| . \sigma_{k-1} [U] (E [\retkw
  \underline *], s)$.  Hence:
  \begin{align*}
    \int_{s \in \{0, 1\}^\nat} \sigma_k [U] (M, s) d\upsilon
    & = \int_{s \in \{0, 1\}^\nat} |\mathbf a| . \sigma_{k-1} [U] (E [\retkw
      \underline *], s) d\upsilon \\
    & = |\mathbf a| . Next^{\leq k-1} (E [\retkw
      \underline *], s) (U) = Next^{\leq k} (M) (U),
  \end{align*}
  by induction hypothesis and Lemma~\ref{lemma:Next*}, item~3.  The
  remaining cases are equally easy, and left to the reader.
\end{proof}

Combining Proposition~\ref{prop:PU} with Theorem~\ref{thm:adeq}, it
follows that the denotational, operational, and sampling-based operational
semantics all agree at observable types.
\begin{theorem}
  \label{thm:adeq:final}
  Let $\Sigma$ consist of strictly observable first-order constants, and let us
  assume that the only available term of the form $\sample [\mu]$ is
  $\sample [\lambda_1]$, where $\lambda_1$ is Lebesgue measure on
  $[0, 1]$.  For every observable type $\beta$, for every generalized
  ISPCF term $M \colon D \beta$, for every observable open subset $U$ of
  $\Eval \beta$,
  \[
    \Eval M (U) = Next^* (M) (\rover U) = P [M \to
    \rover U].
  \]
\end{theorem}

\subsection{A few notes on implementation}
\label{sec:few-notes-impl}

We finish this section by briefly mentioning that the deterministic
relation $\to_s$ is easily implemented.  One only needs to be able to
compute on exact real numbers.  There are quite a number of proposals
in order to do this
\cite{BC:exact:real,lester92,sunderhauf95c,escardo96a,escardo96,edalat97,Plume:ERC,marcial04,dGL:ERA,lambov05,Ho:ERCE}.
The only modification one needs to make is that some of the exact real
numbers $\theta (x) \in \IR$ we have to consider were drawn by
sampling from $\sample [\lambda_1]$, i.e., were obtained as
$\num (s [i])$ for some $i \in \nat$.  Reading that number up to a
specified accuracy $\epsilon > 0$ means reading its first $k$ bits,
where $1/2^k < \epsilon$.  In an actual implementation, we would
therefore only need (arbitrarily large) finite prefixes of an
(arbitrarily large) finite number of substrings $s [i]$, and its bits
can be drawn at random, and stored, on demand.

Explicitly, in the redundant signed-digit representation once proposed
by Boehm and Cartwright in \cite{BC:exact:real} and later by di
Gianantonio and Lanzi \cite{dGL:ERA}, and by Ho \cite{Ho:ERCE}, real
numbers between $-1$ and $1$ are infinite streams of digits in
$\{-1, 0, 1\}$, to be read in base $2$.  Those infinite streams are
read on demand, and can represent $s [i]$ verbatim.  In Boehm and
Cartwright's second representation, real numbers $x$ are encoded as
functions $\overline x$ mapping each natural number $k$ to an integer
$\overline x (k)$ such that $\overline x (k)/4^k$ is at distance at
most $1/4^k$ from $x$.  One can then represent $s [i]$ by maintaining
a table $T_i$ of the first few bits of $s [i]$ in mutable store.  When
$\overline {s [i]} (k)$ is required, we make sure that at least the
first $2k$ entries of $T_i$ are populated, possibly drawing fresh
bits, uniformly and independently at random, and storing them in
$T_i$; then we return $\sum_{j=0}^{2k-1} T_i [j]/2^{j+1}$.

The other primitives working on real numbers that we have considered
are well covered in the literature, except perhaps for $\bin'$ and
$\num'$, mentioned in Section~\ref{sec:distr-high-order}.  We have
already described how $\bin'$ can be implement, right before
Remark~\ref{rem:bin:unfold}.  In a redundant signed-digit
representation, Ho explains how one can compute series of the form
$\sum_{k=0}^{+\infty} x_k /2^{k+1}$, where
$(x_0, x_1, \cdots, x_k, \cdots)$ is a stream of exact real numbers in
$[-1, 1]$ \cite[Section~3.1.2]{Ho:ERCE}, based on unpublished work by
Adam Scriven.  (This is also the denotation of Simpson's
$\mathtt{coerce}$ operation \cite[Figure~3]{simpson98}.)  From this,
$\num'$ is easily implemented.  In Boehm and Cartwright's second
representation, one can implement $\num' (s)$ as the function that
maps each $k \in \nat$ to the partial, finite sum
$\sum_{k'=0}^{2k-1} s_{k'}/2^{k'+1}$ if $s_0$, \ldots, $s_{2k-1}$ are
all different from $\bot$, and as $\bot$ otherwise.
While this definition of $\num'$ differs slightly from
(\ref{eq:num'}), it shares with it the property that $\num' (s)$
computes the exact real representation of $\num (s)$ if no bit $s_k$
is equal to $\bot$, which is the only thing we require from it.

      %

\section{Related work}
\label{sec:related-work}


We have already cited a few papers on statistical programming
languages \cite{GMRBT:Church,WvdMM:Anglican,GS:webppl,MSP:venture}
and, what is of more interest to us, on \emph{semantics} of such
languages \cite{VKS:SFPC,HKSY:qBorel,DLH:geom:bayes,EPT:PPCF}.  We
will not cite the even larger body of literature that deals with the
simpler case of probabilistic, higher-order languages without real
numbers, continuous distributions, or soft constraints.  As only
exception, we mention \cite{PPT:sampling}, which presents a
sampling-based operational semantics for a higher-order probabilistic
call-by-value language with continuous distributions and true real
numbers, and which can be seen as a precursor of some of the
previously mentioned papers.

We have insisted on obtaining a commutative monad of continuous
$R$-valuations, and we have noted that this amounts to forms of the
Fubini-Tonelli theorem.  There are several such commutative monads in
the literature already, including some that predate the novel
commutative monad of \cite{VKS:SFPC}.  On the category of measurable
spaces, the monad of spaces of measures, popularized by M. Giry
\cite{giry82}, is commutative, as first proved by L. Tonelli
\cite{Tonelli09}.  (Tonelli's theorem applies to non-negative
measurable maps, while the Fubini-Tonelli theorem applies to
measurable maps whose absolute value is integrable.)  Vickers
introduced a valuation monad on the category of locales, and proved
that it is commutative, too \cite{Vickers:val:loc}.  The topological
counterpart of that result is Proposition~\ref{prop:fubini:top}.  We
have argued in Section~\ref{sec:fubini-theorem-what} why this does not
imply a corresponding Fubini-Tonelli theorem on the category $\Dcpo$
of dcpos.

We get around this problem by considering \emph{minimal} valuations
instead of general continuous valuations.  The minimal
\emph{subprobability} valuations are an instance of a more general
construction of Jia and Mislove \cite{jiamis}, based on Keimel and
Lawson's $K$-completions \cite{keimel08}.  Heckmann's point-continuous
valuations \cite{heckmann95} are another instance of that
construction.  It was realized independently by the first and second
authors in 2020 that such constructions provide for commutative monads
on $\Dcpo$.  Explicitly, Jia realized that the constructions of
\cite{jiamis} all produced commutative monads, while Goubault-Larrecq
realized that the monad of point-continuous valuations was commutative
on $\Dcpo$, and made it an exercise in a book in preparation
\cite[Exercise~12.8.11]{GL:val}.  The two groups decided to write up
on this, in different settings and with different applications in
mind.  This led to the paper \cite{JLM:prob:quant} by the Tulane
group, and to the present paper.

The general idea of using a form of completion of simple objects
(simple valuations, here) in order to obtain objects with better
properties is, naturally, a very common idea.  A recent and topical
example is the construction of the commutative monad $T$ of
\cite[Section~4]{VKS:SFPC}.  There, the authors consider the
continuation monad $J$ with answers in $[0, +\infty]$.  The analogue
in our case would be to define $J X$ as
$[[X \to \creal] \to \creal] = [\Lform X \to \creal]$, and one should
realize that $\Val X$ embeds into $J X$ by mapping each continuous
valuation $\nu$ to its integration functional $\int \_ d\nu$.  Then,
they consider the so-called randomisable expectation operators, which
are obtained as the (integration functionals of) the image measures of
Lebesgue measure on $\real$.  Finally, $T$ is the smallest full
submonad of $J$ that contains all the randomisable expectation
operators.  A simpler, and more classical example of completion is the
notion of \emph{s-finite} measure
\cite{Sharpe:markov,Getoor:exc:meas}, which are the suprema of
countable chains of bounded measures; s-finite measures satisfy
Fubini-Tonelli, but are not known to form a monad on the category of
measurable sets and measurable maps.  Instead, s-finite kernels form a
commutative monad on that category, which Staton used to give
semantics to a simple first-order probabilistic language
\cite{Staton:sfinite}.

\section{Conclusion}
\label{sec:conc}

We have given a simple domain-theoretic, denotational semantics for a
simple statistical programming language, based on the novel notion of
\emph{minimal} valuation, forming a commutative monad over $\Dcpo$.
We believe that the notion is natural, and that the resulting
denotational semantics is simple.  Through an extensive list of
examples, we hope to have demonstrated that a language based on this
semantics is able to implement a rich set of distributions, including
some non-trivial distributions on higher-order objects.  We have given
three operational semantics, which are all sound and adequate, under
various, natural sets of assumptions on the semantics of available
constants.  Further work include formally verified efficient
implementations, for example extending the particle Monte Carlo Markov
chain algorithm of Anglican \cite{WvdMM:Anglican}.

\section*{Acknowledgments}
\label{sec:acknowledgments}

We took some advice from the anonymous referees on an earlier, and
very different version of this paper submitted at the LICS'21
conference.  We thank Vladimir Zamdzhiev, who spotted a mistake in
the claim of soundness in one version of this paper.  We thank Daniel
Roy for mentioning theories of computable probability distributions to
us, and Sam Staton for correcting a misunderstanding we had about
s-finite measures.


\bibliographystyle{apalike}
\ifarxiv

\else
\bibliography{refs}
\fi

\appendix




\end{document}